\documentclass[11pt]{article}
\usepackage{amsmath, amssymb, amsthm}
\usepackage{graphicx,psfrag}
\usepackage{mathtools}
\usepackage{bm}
\usepackage{bbm}
\usepackage{enumerate}
\usepackage{subfigure}
\usepackage{natbib}
\usepackage[mathscr]{euscript}
\usepackage{url} 
\usepackage[bottom]{footmisc}
\usepackage{booktabs}
\usepackage{multirow}
\usepackage{color}
\usepackage[hypertexnames=false]{hyperref}
\hypersetup{
colorlinks = true,
allcolors = blue
}
\usepackage{caption}
\usepackage{adjustbox}

\usepackage[a4paper, 
            left=1.1in,right=1.1in,top=1in,bottom=1.2in,%
            footskip=.7in]{geometry}

\usepackage[toc,page,header]{appendix}
\usepackage{minitoc}   
\usepackage{xcolor}
\usepackage[linesnumbered,ruled,vlined]{algorithm2e}

\usepackage{cleveref}
\usepackage{wrapfig}

\makeatletter
\newcommand{\algorithmfootnote}[2][\footnotesize]{%
  \let\old@algocf@finish\@algocf@finish
  \def\@algocf@finish{\old@algocf@finish
    \leavevmode\rlap{\begin{minipage}{\linewidth}
    #1#2
    \end{minipage}}%
  }%
}
\makeatother
\def\ind{\mathbbm{1}}   

\newcommand{\bbR}{\mathbb{R}} 
\newcommand{\bbS}{\mathbb{S}} 
\newcommand{\bbP}{\mathbb{P}} 

\newcommand{\id}{\mathrm{id}}

\newcommand{\bK}{\mathbf{K}} 
\newcommand{\betamin}{C_{\mathrm{min}}}

\def\din{d_{\mathrm{in}}}              
\def\post{\pi_n} 			 
\def\Pa{\mathrm{Pa}} 
\def\Ch{\mathrm{Ch}}  

\newcommand{\se}{\mathsf{e}} 
\newcommand{\cE}{\mathcal{E}}  
\newcommand{\cG}{\mathcal{G}}  

\newcommand{\cA}{\mathcal{A}}
\newcommand{\cB}{\mathcal{B}}
\newcommand{\cC}{\mathcal{C}}
\newcommand{\cD}{\mathcal{D}}

\newcommand{\lmin}{\lambda_{\mathrm{min}}}
\newcommand{\lmax}{\lambda_{\mathrm{max}}}

\def\vmin{\underline{\nu}}
\def\vmax{\overline{\nu}}
\def\umin{\underline{\mu}}
\def\umax{\overline{\mu}}
\newcommand{\cN}{\mathcal{N}}  
\newcommand{\proj}{\Phi}  
\newcommand{\oproj}{\Phi^{\perp}}  

\DeclareMathOperator\argmax{arg\,max}
\DeclareMathOperator\argmin{arg\,min}
\DeclareMathOperator\diag{diag}
\DeclareMathOperator\trace{tr}

\def\sX{\mathsf{X}} 
\def\dag{\mathcal{G}_p}

\def\model{\mathcal{M}_p}

\newcommand{\cZ}{\mathcal{Z}}
\newcommand{\cJ}{\mathcal{J}}
\newcommand{\cK}{\mathcal{K}}

\def\X{X}

\def\RRS{\xi}

\def\MAP{\mathrm{MAP}}
\def\T{\mathrm{T}} 
\def\SSE{\mathrm{RSS}}
\def\minus{\text{-}}

\def\c{\mathrm{c}} 

\SetKwInput{KwInput}{Input}                
\SetKwInput{KwOutput}{Output}              

\theoremstyle{remark}
\newtheorem{remark}{Remark}
\newtheorem{assumption}{Assumption}
\newtheorem{proposition}{Proposition}
\newtheorem{theorem}{Theorem}
\newtheorem{corollary}{Corollary}
\newtheorem{lemma}{Lemma}

\title{Order-based Structure Learning without Score Equivalence} 
\usepackage{authblk}
\author{Hyunwoong Chang$^{1}$, James Cai$^{2}$ and  Quan Zhou$^{1,}$\thanks{Corresponding author: quan@stat.tamu.edu}}
\date{}
\affil{$^{1}$Department of Statistics, Texas A\&M University\\
$^{2}$Department of Veterinary Integrative Bioscience, Texas A\&M University }

\begin{document}
\maketitle

\begin{abstract}
We propose an empirical Bayes formulation of the structure learning problem, where the prior specification assumes that all node variables have the same error variance, an assumption known to ensure the identifiability of the underlying causal directed acyclic graph (DAG).  
To facilitate efficient posterior computation, we approximate the posterior probability of each ordering by that of a best DAG model, which naturally leads to an order-based Markov chain Monte Carlo (MCMC) algorithm. 
Strong selection consistency for our model in high-dimensional settings is proved under a condition that allows heterogeneous error variances, and the mixing behavior of our sampler is theoretically investigated.  
Further, we propose a new iterative top-down algorithm, which quickly yields an approximate solution to the structure learning problem and can be used to initialize the MCMC sampler.   
We demonstrate that our method outperforms other state-of-the-art algorithms under various simulation settings, and conclude the paper with a single-cell real-data study illustrating practical advantages of the proposed method. 
\end{abstract}

\section{Introduction}\label{sec:intro}
\normalsize
We consider Bayesian structure learning of a directed acyclic graph (DAG) model from observational data. 
Bayesian algorithms for structure learning are often classified as score-based in the literature, since they assign a posterior probability to each candidate DAG, the logarithm of which can be interpreted as a score~\citep{drton2017structure}.  
A Markov equivalence class is a set of all DAGs that encode the same set of conditional independence relations among node variables. Without a priori knowledge, we cannot distinguish between two Markov equivalent DAGs using only observational data~\citep{koller2009probabilistic}. 
If a Bayesian model yields the same score for DAGs in the same equivalence class, we say it is score equivalent, which is widely considered a desirable property~\citep{andersson1997characterization}.  
Most Bayesian structure learning methods used in practice are score equivalent~\citep{geiger2002parameter}.    

Since the number of $p$-node DAGs grows super-exponentially with $p$, an exact evaluation of the posterior distribution is impossible unless $p$ is extremely small, and Markov chain Monte Carlo (MCMC) methods are commonly employed to generate samples from the posterior distribution.   
As a classical example, structure MCMC, which was proposed in the seminal work of~\citet{madigan1995bayesian}, is a random walk Metropolis-Hastings algorithm on the DAG space that uses single-edge addition, deletion, and reversal as proposal moves. However, it is known that this algorithm can often suffer from computational inefficiency due to the considerable time it spends sampling DAGs within the same equivalence class~\citep{andersson1997characterization, chickering2002learning}. Even if the data is very informative on the conditional independence relations among all variables, we are only able to learn the equivalence class of the underlying true DAG model, which can easily be very large and takes the chain a long time to explore.  
In order to overcome slow mixing behavior caused by equivalence classes, many DAG MCMC samplers have been proposed, which typically introduce new DAG operations that can realize jumps between very different DAGs, enabling the chain to move more efficiently across equivalence classes~\citep{grzegorczyk2008improving, su2016improving}. 
Another strategy is to devise MCMC samplers on some other spaces that might be easier to explore than the DAG space. Indeed, one can directly search on the equivalence class space so that redundant moves between Markov equivalent DAGs are avoided~\citep{castelletti2018learning, zhou2021complexity}. But this approach is not commonly used in the Bayesian literature, and one likely reason is that, unlike DAG MCMC samplers, the implementation of graph operations for equivalence classes can be highly complicated.

A more popular approach is to perform MCMC sampling on the order space~\citep{friedman2003being, agrawal2018minimal, kuipers2022efficient}. Due to the acyclicity constraint, every $p$-node DAG has at least one consistent ordering of the $p$ nodes  such that node $i$ precedes node $j$  whenever the edge $i \rightarrow j$ is in the DAG. 
Order-based MCMC methods are largely motivated by the following observation: the main computational challenge in structure learning lies in the uncertainty of order estimation, since once the ordering of variables is fixed, structure learning can be reduced to a collection of variable selection problems that are often considered to have a much smaller complexity. 
It is generally believed that the mixing of order MCMC is better than that of structure MCMC, because the search space is smaller and the posterior distribution on the order space tends to be smoother~\citep{friedman2003being}. 
However, the problem of traversing large equivalence classes still exists. 
To see this, assume again that all conditional independence relations can be learned from the data so that the posterior concentrates on one equivalence class. But any two DAGs in this equivalence class must have different orderings since at least one edge is flipped. This implies that the posterior distribution on the order space concentrates on a set at least as large as this equivalence class. 

To mitigate the potential mixing problem caused by traversing large equivalence classes, we propose to impose identifiability conditions  so that within each equivalence class, the posterior mass tends to concentrate on only one DAG. Consequently, the overall posterior distribution tends to have less and sharper modes. To this end, we follow the work of~\citet{peters2014identifiability} to consider Gaussian structural equation models with equal error variances.  
Intuitively, by assuming equal error variances, the data becomes informative on edge directions so that an MCMC sampler can quickly learn the best DAG in its equivalence class. For example, consider two correlated variables $\sX_1, \sX_2$. The DAGs $\sX_1 \rightarrow \sX_2$ and $\sX_2 \rightarrow \sX_1$ are Markov equivalent, and in general, we cannot determine the causal direction if only observational data is available. But the equal variance assumption forces  the posterior score to favor $\sX_1 \rightarrow \sX_2$ if $\sX_2$ has a larger marginal variance than $\sX_1$.  Though a score equivalent Bayesian procedure allows us to make posterior inferences by averaging over Markov equivalent DAGs, 
this advantage is often merely theoretical due to its slow convergence, even when dealing with a moderately large number of node variables. 
Our simulation study and real data analysis will show that the use of equal variance assumption does provide practical advantages, and it improves the posterior inference accuracy unless there is a huge degree of heterogeneity among error variances. 

There is a rapidly growing literature on the identifiability conditions for structure learning~\citep{shimizu2006linear, hoyer2008nonlinear, peters2011identifiability, peters2014identifiability, strieder2021confidence, drton2017structure, glymour2019review}. In particular, two deterministic search algorithms have been proposed recently for structure learning with equal error variances~\citep{ghoshal2018learning,chen2019causal}, and they are shown to be advantageous in terms of computational cost and scale well to high-dimensional data. 
But to our knowledge, the corresponding Bayesian theory and methodology is largely underdeveloped. 
Aiming to fill this gap, we formulate an empirical Bayes model under the equal variance assumption and obtain a posterior score that distinguishes between Markov equivalent DAGs. 
We prove a strong selection consistency result for our model, which shows that the posterior probability of the true DAG tends to one in probability under mild high-dimensional conditions. 
In particular, while our prior distribution encodes the equal variance constraint, the consistency result holds under a weaker assumption known as the minimum-trace condition~\citep{aragam2019globally}. 
Further, we extend the consistency result to cases where errors follow sub-Gaussian distributions, which include more interesting settings such as mixed discrete-Gaussian DAG models. 

The posterior score derived from our model is non-decomposable (see Remark~\ref{rmk:decomp}), which is expected since, under the equal variance assumption, the marginal likelihood of a DAG model should depend on how close the residual variances of the $p$ nodes are to each other. 
This poses new computational challenges and again makes our method very different from the existing Bayesian literature, where decomposable scores are almost always used because the decomposability enables one to evaluate the posterior probability of a DAG by local calculations at each node~\citep{chickering2002learning}. 
  
To numerically evaluate the posterior distribution of our empirical Bayes model, in the same spirit of the minimal I-MAP MCMC of~\citet{agrawal2018minimal}, we approximate the posterior probability of an ordering by that of the best consistent DAG and then build a sampling algorithm on the order space. 
We show that, under some conditions on the edge weights, the chain will never get stuck at a sub-optimal local mode for exponentially many iterations in expectation, which partially explains why this order MCMC scheme may perform well in practice. 
Further, we propose a generalized iterative version of the top-down algorithm of~\citet{chen2019causal}. This algorithm is deterministic and quickly finds a likely ordering of the variables, which can be used as a warm start for our order MCMC sampler. 
When estimating edge inclusion probabilities, we tune our estimators via a conditional expectation calculation so that we can reduce the estimation variance caused by picking one single best DAG for each ordering.  
Lastly, though the non-decomposable score of our model cannot be evaluated locally, we are able to devise an implementation strategy that  makes the posterior evaluation for our model as efficient as that with a decomposable score. The key idea is to store the search paths of the forward-backward stepwise selection at each node, which can be reused in finding the best DAG consistent with a given ordering.  

\section{An empirical Bayes model for order-based structure learning}\label{sec:model}

\subsection{Notation and terminology}\label{subsec:note}
We set up the notation and terminology to be used throughout the paper. Let $G = (V, E)$ denote a DAG, where $V$ is a node set and $E \subset V \times V$ is a set of directed edges that form no cycle. Without loss of generality, for a $p$-node DAG, we assume $V = [p] = \{1, \dots, p\}$. 
For ease of notation, we write $ \{i \rightarrow j \} \in G$ to mean that $(i, j) \in E$, and use $ G \cup  \{i \rightarrow j \} $ (respectively $G \setminus  \{i \rightarrow j \} $) to denote the DAG obtained by adding (respectively removing) the edge $i \rightarrow j$. 
We use $|G|$ to denote then number of edges in $G$. 
We denote by $\bbS^p$ the set of all bijections from $[p]$ to $[p]$. 
An element $\sigma \in \bbS^p$ is said to be a topological ordering for a DAG $G$ if the following holds: for any indices $k < l$, the edge between the nodes $\sigma(k)$ and $\sigma(l)$ is  directed as $\sigma(k) \rightarrow \sigma(l)$, if it exists in $G$. 
Let $\sigma^{-1}$ denote the inverse function of $\sigma$, and for each node $j \in [p]$, let 
\begin{align}\label{eq:potential}
    P_j^{\sigma} = \{i\colon  \sigma^{-1}(i) < \sigma^{-1}(j)\}
\end{align}
denote the set of potential parents of node $j$ under the ordering $\sigma$, i.e., all nodes preceding $j$ in $\sigma$. 
Let $\cG_p$ be the collection of all $p$-node DAGs and $\cG_p^\sigma$ be the collection of all $p$-node DAGs consistent with topological ordering $\sigma$; that is,
$  \cG_p^\sigma = \{ G  \in \cG_p \colon  \{i \rightarrow j \} \in G  \text{ implies } \sigma^{-1}(i) < \sigma^{-1}(j) \}$. 
Given a node $j$, we use $\Pa_j(G)$ and $\Ch_j(G)$ to denote the set of its parent nodes and that of its child nodes, respectively, in the DAG $G$.  If the underlying DAG is clear from the context, we simply write $\Pa_j$ and $\Ch_j$.  
Finally, given a matrix $A \in \bbR^{a \times b} $, $j \in [b]$, $J \subseteq [b]$ and $I \subseteq [a]$, $A_j$ denotes the $j$-th column of $A$, $A_J$ denotes the submatrix of $A$ containing columns indexed by $J$, and $A_{I, j}$ denotes the subvector of $A_j$ with entries $\{A_{ij} \colon i \in I \}$. 
We use $|J|$ to denote the cardinality of the set $J$. 
 
\subsection{Model specification}\label{subsec:model}
Let $\sX = (\sX_1, \dots, \sX_p)$ denote a $p$-dimensional random vector, and denote by $X$ an $n \times p$ data matrix, each row of which is an independent copy of $\sX$. 
For each  $\sigma \in \bbS^p$  and $G \in \cG_p^\sigma$, consider the following structural equation model for the random vector $\sX$, 
\begin{equation}~\label{eq:ln.str.eq}
    \sX_j = B_{\Pa_j(G), j}^\T \sX_{\Pa_j(G)} + \se_j, \quad \se_j  \mid \omega \overset{\text{i.i.d}}{\sim}  N(0, \omega) \text{ for } j = 1, \dots, p, 
\end{equation}
where $\Pa_j(G) \subseteq P_j^\sigma$ for each $j$, and  $B$ is a $p \times p$ matrix. Entries of $B$ that are not involved in~\eqref{eq:ln.str.eq}   are set to zero. 
$B$ can be seen as the weighted adjacency matrix of the DAG $G$ such that  $\{i \rightarrow j\} \in G$ if  $|B_{ij}| > 0$.  

We use the following empirical prior on the parameter $(\sigma, G, B, \omega)$, where $\pi_0$ denotes the prior density function: 
\begin{align}
    B_{\Pa_j(G), j} \mid  G,  \omega \overset{\mathrm{ind}}{\sim} \;&  N_{|\Pa_j(G)|}\left(\hat{ B}_{\Pa_j(G), j},\frac{\omega}{\gamma}  (X_{\Pa_j(G)}^\T X_{\Pa_j(G) })^{-1}\right), \quad \forall \, j \in [p],  \label{prior:B}\\ 
    \pi_0(\omega \mid \sigma) \propto \;& \omega^{-\frac{\kappa}{2}-1}, \label{prior:omega} \\
    \pi_0(G, \sigma) \propto \;&  \left(p^{c_0} \right)^{-|G|} \ind_{ \{ \hat{G}_\sigma \} }(G), \label{prior:G}  
\end{align}
where $\hat{B}_{\Pa_j(G),j}$ is the least-squares estimator of $B_{\Pa_j(G),j}$, $c_0, \gamma, \kappa$ are hyperparamters of the prior, and $\hat{G}_\sigma$ in~\eqref{prior:G} is the best estimate for $G$ among $\cG_p^\sigma$; we will detail how to obtain $\hat{G}_\sigma$ later. 
This prior is doubly empirical. 
First, given $G$ and $\omega$, we use an empirical prior on $B_{\Pa_j(G), j}$ for each $j$ in~\eqref{prior:B}, where the conditional prior mean  depends on the data.   
Following~\citet{martin2017empirical} and~\citet{lee2019minimax}, when computing the posterior distribution, we raise the data likelihood to the power of $\alpha$, where $\alpha \in (0,1)$ is a constant, so that we can reduce the influence of the data that is inflated by the usage of the empirical prior.  
\citet{lee2019minimax} suggests setting $\alpha$ close to 1 to make the $\alpha$-likelihood behave similarly to the standard likelihood in finite sample scenarios.
Observe that the covariance in~\eqref{prior:B} is identical to that of Zellner's g-prior,   proportional to the inverse Fisher information matrix for $B_{\Pa_j(G),j}$~\citep{tadesse2021handbook}. 
An alternative approach to specifying the prior is to use the fractional Bayes factor~\citep{carvalho2009objective, castelletti2021bayesian}.  This yields a fractional posterior with the value of $\alpha$ determined automatically, but the resulting posterior is more difficult to calculate than the proposed posterior.  
Second, according to~\eqref{prior:G},  
the conditional prior distribution of $G$ given $\sigma$ is again empirical: it assigns  unit mass to some $\hat{G}_\sigma$ that  can be seen as the solution to a DAG selection problem given ordering $\sigma$. 
This implies that the marginal prior distribution of $G$ has support $\hat{\cG} = \{\hat{G}_\sigma \colon \sigma \in \bbS^p \}$. 
For moderately large $p$, searching the entire space $\cG_p$ is impossible, but the empirical prior~\eqref{prior:G}  reduces the size of the search space to that of the order space $\bbS^p$. Unfortunately,   $| \bbS^p | = p!$  is still super-exponential in $p$, making it challenging to devise an efficient MCMC sampler.

\begin{remark}
The use of the empirical prior~\eqref{prior:G} makes our approach very different from traditional Bayesian structure learning methods, where posterior inference is performed by averaging over all DAG models that satisfy certain sparsity constraints.  
The seminal order-based MCMC sampler of~\citet{friedman2003being} imposes a uniform conditional prior given $\sigma$ on all DAGs satisfying degree constraints in $\cG_p^\sigma$.  
But calculating the un-normalized marginal posterior probability of an ordering requires summation over all possible DAGs, which is infeasible unless  $p$ is small or the degree constraint is highly demanding. Further, the technique used in~\citet[Eq. (8)]{friedman2003being} to expedite this calculation is not applicable in our case since our score is not decomposable; see Remark~\ref{rmk:decomp}.  
Therefore, we prefer using the empirical prior~\eqref{prior:G} for its computational efficiency. 
A similar approach is taken in \citet{agrawal2018minimal}, which uses empirical conditional independence tests to construct a minimal independence map for each ordering and restricts the search space to the set of minimal independence maps.   
Henceforth, we will always use DAG selection to refer to the problem of identifying the best DAG with given ordering.    
\end{remark}

Let $\post$ denote the posterior distribution given the observed data matrix $X$. By a standard normal-inverse-gamma calculation that integrates out the parameters $B$ and $\omega$, we get 
\begin{equation}\label{def.posterior}
    \post(G, \sigma) \propto  e^{ \phi(G)} \ind_{ \{ \hat{G}_\sigma \} }(G), 
\end{equation}
where  $\phi(G)$ is called the score of $G$ and is given by
\begin{equation}~\label{eq:post_score}
\begin{aligned}
&  \phi(G)   =  -|G| c_0 \log p  - \frac{|G|}{2} \log [ (1 + \alpha/\gamma) ] -  \frac{\alpha p n + \kappa}{2} \log  \left(\sum_{j=1}^p \SSE_j(G) \right),  \\
&  \text{where } \SSE_j(G)  = X_j^\T \oproj_{\Pa_j(G)} X_j, \quad 
  \oproj_S = I -  X_S(X_S^\T X_S)^{-1}X_S.
\end{aligned}
\end{equation}
We will also sometimes refer to $\phi(G)$ as the posterior score. 
For a detailed derivation of \eqref{def.posterior}, see Section~\ref{subsec:post} in the supplementary material. The marginal posterior probability of an ordering $\sigma$ and that of a DAG $G$ are
\begin{equation}\label{eq:marginal.post}
    \post(\sigma) \propto e^{\phi(\hat{G}_\sigma )}, \quad \post(G) \propto e^{ \phi(G ) }  \sum_{\sigma \in \bbS^p} \ind_{\{ \hat{G}_\sigma \}} (G).
\end{equation} 
For our model, $ \post(G)$ is not exactly proportional to the exponentiation of the score of $G$ due to the factor $\sum_{\sigma \in \bbS^p} \ind_{\{ \hat{G}_\sigma \}} (G)$, and in our high-dimensional analysis we will show this term is negligible under mild assumptions. 

In the rest of this work, we consider the following choice for $\hat{G}_\sigma$, 
\begin{equation}\label{def.map}
    \hat{G}_{\sigma}^{\MAP}(\din ) = \argmax_{G \in \cG_p^\sigma(\din)} \phi(G),   \quad \forall \, \sigma \in \bbS^p, 
\end{equation}
where $\cG_{p}^{\sigma}(\din) = \{G \in \cG_{p}^\sigma : |\Pa_j(G)| \leq \din \text{ for all } j\in [p]\}$ is the collection of all $p$-node DAGs with maximum in-degree bounded by $\din$.  
For our high-dimensional analysis, we will impose the condition $\din \log p = o(n)$, 
which is commonly used in the literature on high-dimensional DAG selection~\citep{cao2019posterior, lee2019minimax}. 
The superscript MAP indicates that $\hat{G}_{\sigma}^{\MAP} $ is the DAG with the largest posterior score among $\cG^\sigma_p(\din)$, i.e., the maximum a posteriori estimate. 
 
\begin{remark}\label{rmk:decomp}
In most existing methods for Bayesian structure learning, the posterior score of a DAG $G$ takes a decomposable form in the sense that it can be written as the sum of $p$ terms, where the $i$-th term only involves node $i$ and its parent set and thus can be evaluated locally. 
But our posterior score given in~\eqref{eq:post_score} is not decomposable due to the equal variance assumption used in the prior: integrating out $\omega$ results in the logarithm of the sum of $p$ residual sum of squares  ($\SSE$) terms  in~\eqref{eq:post_score}.   
This non-decomposable score is able to discriminate between Markov equivalent DAGs, and as we will prove shortly, given sufficiently large sample size, the posterior distribution of our model concentrates on only the unique true DAG.   
\end{remark}

\subsection{Strong model selection consistency}\label{subsec:consistency}
We consider a high-dimensional setting where $n$ tends to infinity and both $p = p(n)$ and $\din = \din(n)$ may grow with $n$.  Strong model selection consistency means that the posterior probability of the true model converges to $1$ in probability with respect to the true probability measure from which the data is generated. This is often regarded as one of the most important theoretical guarantees for a high-dimensional Bayesian model selection procedure. In the DAG literature, it was proven for DAG selection with known ordering~\citep{cao2019posterior, lee2019minimax} and structure learning up to equivalence class~\citep{zhou2021complexity}.  
To the best of our knowledge, there is no strong selection consistency result on Bayesian structure learning under an identifiability condition. 

Though the equal variance assumption was used in the prior specification, for our consistency analysis, we consider a more general setting. 
Assume the data is generated according to the structural equation model 
\begin{equation}\label{eq:true.sem}
    \sX_j = (B^*_{\Pa_j(G^*), j})^\T \sX_{\Pa_j(G^*)} + \se_j, \quad \se_j  \sim  N(0, \omega^*_j) \text{ for } j = 1, \dots, p, 
\end{equation}
where $G^*, B^*, \{\omega_j^*\}_{j=1}^p$ denote the true parameter values, and we assume $B^*_{ij} \neq 0$ if and only if $\{i \rightarrow j \} \in G$. 
Define $\Omega^* = \diag(\omega_1^*, \dots, \omega_p^*)$. 
Let $[\sigma^*]$ denote the set of all  orderings consistent with $G^*$, where $\sigma^*$ is some element in $[\sigma^*]$ interpreted as the true ordering. 
Thus, $G^* \in \cG^\sigma_p$ if and only if $\sigma \in [\sigma^*]$.   
Let $\bbP^*$ denote the probability measure corresponding to the structural equation model~\eqref{eq:true.sem}. 
Observe that  the covariance matrix of the random vector $\sX$ can be written as $\Sigma^* =  \Sigma(B^*, \Omega^*)$, where  
\begin{align}\label{eq:mcd} 
\Sigma(B, \Omega) =  (I_p - B^\T)^{-1} \Omega (I_p - B)^{-1}. 
\end{align}
This is known as the modified Cholesky decomposition.  
This decomposition of $\Sigma^*$ is not unique, as we explain in the following remark.   

\begin{remark}\label{remark:minimap}
For each ordering $\sigma \in \bbS^p$, there exists a unique tuple $(B_\sigma^*, \Omega_\sigma^*)$ such that  $B_\sigma^*$ is the weighted adjacency matrix of a DAG in $\cG^\sigma_p$, $\Omega_\sigma^*$ is a diagonal matrix with all diagonal entries being strictly positive, and $\Sigma^* = \Sigma(B_\sigma^*, \Omega_\sigma^*)$. 
Write $\Omega^*_\sigma = \diag(\omega_1^\sigma, \dots, \omega_p^\sigma)$ and use $G_\sigma^*$ to denote the DAG with edge set $E_\sigma^* = \{(i,j): |(B_{\sigma}^*)_{ij}| > 0\}$ and define 
\begin{equation}\label{eq:def.dstar}
 d^* =  \max_{\sigma \in \bbS^p} \max_{j \in [p]} | \Pa_j(G^*_\sigma)|. 
\end{equation} 
\end{remark}

To prove that the empirical Bayes model specified in Section~\ref{sec:model} has strong model selection consistency in high-dimensional settings, we make the following two  assumptions. 

\renewcommand*{\theassumption}{\Alph{assumption}}

\begin{assumption}[Minimum-trace  condition]\label{A:omega}
There exists a universal constant $\eta \in (0, \infty)$ such that $  \min_{\sigma \notin [\sigma^*]}   \trace(\Omega_\sigma^*) / \trace(\Omega^*)   > 1 + \eta^{-1},$ where $\trace$ denotes the trace.  
\end{assumption}

\begin{assumption}[Consistency of DAG selection given true ordering]\label{A:freq}
The estimator $\hat{G}_{\sigma}$ satisfies $\bbP^*( \cap_{\sigma \in [\sigma^*]} \{ \hat{G}_{\sigma} = G^*  \}  ) \geq 1 - \zeta(p)$ for some $\zeta(p) \rightarrow 0$. 
\end{assumption}

The first assumption includes the equal variance assumption as a special case. 
To see this, suppose that $\Omega^* = \diag(\omega^*, \dots, \omega^*)$ for some  $\omega^* > 0$. 
Since the determinant of $\Sigma^*$ satisfies $\mathrm{det}(\Sigma^*) =  (\omega^*)^p = \prod_{j=1}^p \omega_j^\sigma$ for all $\sigma \in \bbS^p$, we have $p \omega^* \leq \sum_{j=1}^p \omega_j^\sigma$ by the inequality of arithmetic and geometric means. That is,  the true ordering $\sigma^*$ satisfies $\trace (\Omega_{\sigma^*}^*) = \min_{\sigma} \trace (\Omega_\sigma^*)$. Hence, there always exists some $\eta(n)$  such that $\min_{\sigma \notin [\sigma^*]}   \trace(\Omega_\sigma^*) / \trace(\Omega^*)  > 1 + \eta(n)^{-1}.$ 
Assumption~\ref{A:omega} just requires that $\eta(n)^{-1}$ can be bounded away from zero so that we can replace it with some universal constant $\eta$. 
Under the equal variance assumption, we can rewrite Assumption~\ref{A:omega} as follows, which  has been used in~\citet{van2013ell} and is known as the omega-min condition. 

\setcounter{assumption}{0}
\renewcommand*{\theassumption}{\Alph{assumption}'}
\begin{assumption}[Assumption A with equal variances] \label{A1} 
Suppose $\Omega^* \!= \!\diag(\omega^*, \dots, \omega^*)$, where $\omega^* > 0$ is the  error variance shared by all node variables.
There exists a universal constant $\eta \in (0, \infty)$ such that $\min_{\sigma \notin [\sigma^*]}   p^{-1}  \sum_{j = 1}^p  ( \omega_j^{\sigma} /  \omega^* )   > 1 + \eta^{-1}$. 
\end{assumption}
\renewcommand*{\theassumption}{\Alph{assumption}}
\setcounter{assumption}{2}

\begin{remark}\label{min.trace}
Recall our score function  given in~\eqref{eq:post_score} and that $\SSE_j / n$ is an estimate of the error variance $\omega^\sigma_j$.  So our method essentially aims to select the DAG that provides the tightest fit to the data. 
More precisely, the score~\eqref{eq:post_score} aims to learn the best DAG in $\cG_{p}^\sigma$ where $\sigma$ minimizes $\trace (\Omega_\sigma^*)$, the sum of error variances; such a DAG is called the minimum-trace DAG. 
Our strong consistency result, which only requires Assumption~\ref{A:omega} instead of Assumption~\ref{A1}, confirms that though the equal variance assumption was used to derive~\eqref{eq:post_score}, our method has the theoretical guarantee under a more general setting.  
We refer readers to~\citet{aragam2019globally} for a general theory on structure learning using  minimum-trace DAGs. 
\end{remark}

\begin{remark}\label{rmk:conjecture}
An interesting open question is, without  the equal variance assumption, what choices of $(B^*, \Omega^*)$ can satisfy the minimum-trace condition so that the true model is identifiable.  
We conjecture that if for some  $\sigma^* \in \bbS^p$, we have  $\omega^{\sigma^*}_{\sigma^*(1)}\leq \omega^{\sigma^*}_{\sigma^*(2)} \leq \dots \leq \omega^{\sigma^*}_{\sigma^*(p)}$, then  $\trace (\Omega_{\sigma^*}^*) = \min_{\sigma} \trace (\Omega_\sigma^*)$. 
This weakly increasing variance condition falls under the broader identifiability conditions presented in~\cite{park2020identifiability}, which extend beyond the equal variance assumption.
We have conducted extensive numerical experiments, which suggest that the conjecture is likely to be true, but a proof for every $p \geq 2$  seems highly challenging. Simulation studies are presented in Section C.2 of the supplement.  
\end{remark}

The second assumption says that when we are given an ordering $\sigma \in [\sigma^*]$, the pre-specified DAG selection procedure is able to identify the true DAG with high probability. 
This is a very mild assumption since if the ordering is known, one can often apply an existing consistent algorithm for high-dimensional variable selection to select the parent set of node $j$ for each $j \in [p]$ separately~\citep{ben2011high, yu2017learning, shojaie2010penalized, cao2019posterior, lee2019minimax}. 
We do not need any assumption on the behavior of $\hat{G}_\sigma$ when $\sigma \notin [\sigma^*]$.  
Among many possible DAG selection methods, we use the  estimator defined in~\eqref{def.map} for the following reason. If some other DAG selection method is used, for any $\sigma \notin [\sigma^*]$, there is no guarantee that $\hat{G}_\sigma$ has a sufficiently large posterior score compared with other DAGs in $\cG_p^\sigma$, and the resulting posterior distribution on the order space $\bbS^p$ could be very irregular and contain more sub-optimal local modes. 
However, no existing consistency result can be readily applied to the estimator~\eqref{def.map} due to the non-decomposable posterior score it uses. 
We prove in the following proposition that it does have strong consistency for DAG selection, and it satisfies Assumption~\ref{A:freq} with $\zeta(p) = 4p^{-1}$.  
All the three conditions assumed in Proposition~\ref{prop:freq} are commonly used in the literature: (C\ref{A:eigen}) is known as the restricted eigenvalue condition, (C\ref{A:prior}) assumes prior parameters are properly chosen, and (C\ref{A:beta-min}) is often called the $\beta$-min condition~\citep{lee2019minimax}. 
Except universal constants, all parameters are allowed to depend on $n$. 

\begin{proposition}\label{prop:freq} 
Suppose $\max_j |\Pa_j(G^*)| \leq \din$, and the following conditions hold. 
\begin{enumerate}[({C}1)]
    \item There exist $\vmin, \vmax > 0$ and a universal constant $\delta > 0$ such that
    \begin{align*}
        \frac{\vmin}{(1-\delta)^2} \leq \lmin(\Sigma^*) \leq \lmax(\Sigma^*)\leq \frac{\vmax}{(1+\delta)^2}, 
    \end{align*}
    where $\lmin, \lmax$ are the smallest and largest eigenvalues, respectively.  \label{A:eigen}
    \item The sparsity parameter $\din$ satisfies $\din \log p  = o(n)$, and prior parameters satisfy that $\kappa \leq n p, 0 \leq \alpha/\gamma \leq p^2 -1,$ $c_0 > \rho (\alpha+1) \, \max_{i\neq j} (\omega_j^*/ \omega_i^*)$, and $\rho > 4\din + 6$.  \label{A:prior}
    \item For the true weighted adjacency matrix $B^*$, \label{A:beta-min}
\begin{align*}
    \betamin =  \min \{|(B^*)_{ij}|^2 : (B^*)_{ij} \neq 0 \} \geq 16 c_0 \frac{\vmax^2 \log p}{\alpha \vmin^2 n}. 
\end{align*} 
\end{enumerate}
Consider the posterior score given in~\eqref{eq:post_score} and the  estimator  defined in~\eqref{def.map}.
For sufficiently large $n$, with probability at least $1 - 4p^{-1}$, all the following three events happen. 
\begin{enumerate}[(i)]
    \item For any $\sigma \in [\sigma^*]$, $G \in \cG_p^\sigma(2 \din)$,  $j \in [p]$ such that $ \Pa_j(G^*) \subset \Pa_j(G)$, there exists some $G' \in \cG^\sigma_p$ such that $\phi(G') > \phi(G)$ and $G' = G \setminus \{i \rightarrow j\}$ for some $i \in [p]$. 
    \item For any $\sigma \in [\sigma^*]$, $G \in \cG_p^\sigma(2\din)$, $j \in [p]$ such that $\Pa_j(G^*)  \not\subseteq \Pa_j(G)$, there exists some $G' \in \cG^\sigma_p$ such that $\phi(G') > \phi(G)$ and $G' = G \cup \{i \rightarrow j\}$ for some $i \in [p]$. 
    \item For any $\sigma \in [\sigma^*]$,  $\hat{G}_{\sigma}^\MAP = G^*$. 
\end{enumerate}
\end{proposition}
\begin{proof}
See Section~\ref{proof:freq} in the supplementary material. 
\end{proof}


\begin{remark}\label{rmk:stepwise1}
For computational efficiency, to estimate $\hat{G}_\sigma^\MAP$, one may use a forward-backward stepwise selection to find $\Pa_j$ for each $j$ separately. 
This is outlined in Algorithm~\ref{alg:nodewise_step} in Section~\ref{subsec:stepwise} of the supplementary material. 
Since the posterior score is not decomposable, the stepwise selection at node $j$ depends on the values of $\{\SSE_i \colon i \neq j \}$. A simple solution is to  estimate $\SSE_i$ by  $X_i^\T X_i$ for each $i \neq j$. Then, parts~(i) and~(ii) of Proposition~\ref{prop:freq} imply that this procedure is consistent as long as for each $j$,  $|\Pa_j|$ is bounded by $\din$ at the end of the forward phase in Algorithm~\ref{alg:nodewise_step}. As shown in~\citet{an2008stepwise} and~\citet{zhou2010thresholded}, this condition on the output of forward selection can often be satisfied, with high probability, by choosing some $\din = O(\max_j |\Pa_j(G^*)|)$; i.e., $\din$ has the same order as the maximum in-degree of $G^*$. Actually, Proposition~\ref{prop:freq} implies that the following procedure is also consistent: starting from an arbitrary DAG $G$ with maximum in-degree bounded by $\din$, 
one performs stepwise selection at each node $j$ by setting 
$\SSE_i = \SSE_i(G)$ for each $i \neq j$. 
\end{remark}

\begin{remark}\label{rmk:stepwise2}
An alternative approach to performing forward-backward DAG selection with given ordering is to consider all the $p$ nodes jointly; see Algorithm~\ref{alg:dagwise} in Section~\ref{subsubsec:effective} of the supplementary material. 
In the forward phase, we add one best edge consistent with the given ordering in each iteration, while in the backward phase, we remove one edge in each iteration.  
Proposition~\ref{prop:freq} implies that this algorithm is also consistent for $\sigma \in [\sigma^*]$, provided that the maximum in-degree of any DAG on the search path is bounded by $\din$. 
\end{remark}     
  
The main result of this section is given in the following theorem. 

\begin{theorem}[Strong selection consistency]\label{thm.consistency}
Suppose Assumption~\ref{A:omega},~\ref{A:freq} hold, and assume that $d^* \leq \din$ and $\din \log p = o(n)$. 
Then $\post(G^*)$ converges in probability to 1 with respect to $\bbP^*$, where $\post$ is as given in~\eqref{eq:marginal.post}.   
\end{theorem}
\begin{proof}
See Section~\ref{proof:consistency} in the supplementary material.
\end{proof}
\begin{remark}
The proof can be further extended to cases where the errors $\se_j$, $j=1, \dots, p$ in~\eqref{eq:true.sem} 
follow a sub-Gaussian distribution. As any bounded random variable is sub-Gaussian, this relaxation covers scenarios where some variables are normally distributed and others are discrete and bounded~\citep{lauritzen1992propagation}. 
The proof is given in Section~\ref{proof:subgaussian} in the supplementary material. 
Some inequalities cannot be obtained as sharply as in the Gaussian case, because zero correlation does not imply independence in the sub-Gaussian case. 
\end{remark}
Consider the marginal posterior distribution on the order space $\bbS^p$. The following corollary shows  that the posterior mass concentrates on the set of orderings consistent with $G^*$, and the posterior probabilities of all other orderings vanish.  

\begin{corollary}\label{cor.consistency}
Under the setting of Theorem~\ref{thm.consistency}, $ \post([\sigma^*])$ converges in probability to 1 with respect to $\bbP^*$. 
\end{corollary}

\begin{proof}
This follows from Theorem~\ref{thm.consistency} and $\post(G^*) =  \sum_{\sigma \in [\sigma^*]} \post(G^*, \sigma) =  \sum_{\sigma \in [\sigma^*]} \post(\sigma)$. 
\end{proof}

\section{Posterior sampling via order MCMC}\label{sec:implement} 

\subsection{Metropolis-Hastings algorithms on the order space}\label{subsec:MH} 
To generate posterior samples for our model, we use random walk Metropolis-Hastings algorithms on the order space $\bbS^p$. 
For each $\sigma \in \bbS^p$, let  $\mathbf{K}(\sigma,  \cdot)$ denote the proposal distribution at state $\sigma$. 
We consider three types of random walk proposals: adjacent transposition, which is a standard choice for order-based MCMC methods~\citep{friedman2003being, agrawal2018minimal},  random transpositions and random-to-random shuffles, which are more commonly seen in the literature on random walks on symmetric groups~\citep{levin2017markov, bernstein2019cutoff}. 
All three types of proposals correspond to defining $\mathbf{K}(\sigma,  \cdot)$ by 
\begin{equation}\label{eq:proposal.matrix}
    \mathbf{K}(\sigma,  A ) = \frac{ | \cN(\sigma) \cap  A| }{|\cN(\sigma)|}, \quad \forall \, A \subseteq \bbS^p, 
\end{equation}
for some set $\cN(\sigma) \subset \bbS^p$.  
We refer to $\cN(\sigma)$ as the neighborhood of $\sigma$, and now we formally define this set for each type of proposal. 
Let $(\cdot)_{\mathrm{c}}$ denote an ordering in the cycle notation; for example, $\mu = (a, b, c)_{\mathrm{c}}$ is the ordering given by $\mu(a) = b, \mu(b) = c, \mu(c) =a $ and $\mu(k) = k$ for every $k \notin \{a, b, c\}$. 
Let $\circ$ denote the composition of two orderings; that is, $\tau = \sigma \circ \mu$ is defined by $\tau(i) = \sigma(\mu(i))$.   
Then, we can use  $\sigma \circ (i, j)_{\mathrm{c}}$ to denote the ordering obtained by interchanging the $i$-th and the $j$-th elements of $\sigma$ while keeping the others unchanged. 
Let $ \sigma \circ \RRS(i, j)$ denote the ordering obtained by inserting the $i$-th element of $\sigma$ to the $j$-th position, where $\RRS(i, j)$ is defined by  
$  \RRS(i, j) = (i, i + 1, \dots, j)_{\mathrm{c}}$ if $i < j$, and $\RRS(i, j) =(i, i - 1, \dots, j)_{\mathrm{c}}$ if $i > j$. 
Define the adjacent transposition neighborhood by 
\begin{align*}
    \cN_{\mathrm{adj}}(\sigma) &\; = \{  \sigma' \in \bbS^p 
    \mid   \sigma' = \sigma \circ  (i, i + 1)_{\mathrm{c}} ,  \; 
    i \in [p - 1] \};
\end{align*}
that is, $\cN_{\mathrm{adj}}(\sigma)$ is the set of all orderings that can be obtained from $\sigma$ by one adjacent transposition. 
Similarly, we denote the neighborhood corresponding to random transpositions by  $\cN_{\mathrm{rtp}}$ and that corresponding to random-to-random shuffles by  $\cN_{\mathrm{rrs}}$, which are defined by 
\begin{align*} 
    \cN_{\mathrm{rtp}}(\sigma) &\; = \{ \sigma' \in \bbS^p 
    \mid  \sigma' = \sigma \circ  (i, j)_\mathrm{c}, \; i < j, \, \text{and } i,j \in [p] \}, \\
    \cN_{\mathrm{rrs}}(\sigma) &\; = \{ \sigma' \in \bbS^p 
    \mid \sigma' = \sigma \circ  \RRS(i,j), \;  i \neq j, \, \text{and } i,j \in [p] \}. 
\end{align*}  
We provide an illustration of the three proposals in the supplementary material A.4. Observe that all the three neighborhood relations defined above are symmetric: if $\sigma' \in \cN(\sigma)$, then $\sigma \in \cN(\sigma')$. 
Therefore, by the Metropolis rule, the transition matrix of the algorithm can be calculated by 
\begin{align}\label{eq:transition}
    \mathbf{P}(\sigma, \sigma') = 
    \begin{cases}
     \mathbf{K}(\sigma, \sigma') \min\left\{1, \frac{\post(\sigma') \mathbf{K}(\sigma', \sigma) }{\post(\sigma) \mathbf{K}(\sigma, \sigma')}\right\}, \quad & \text{ if } \sigma' \neq \sigma, \\
     1 - \sum_{\tau \neq \sigma} \mathbf{P}(\sigma, \tau), \quad & \text{ if } \sigma' = \sigma,
    \end{cases}
\end{align}
where $\post(\sigma)$ is the marginal posterior probability and also the stationary probability of $\sigma$. The Hastings ratio $\bK(\sigma',\sigma) / \bK(\sigma,\sigma') = 1$ for all the three neighborhood relations we consider.  
As explained in Section~\ref{sec:model}, once we select an ordering $\sigma \in \bbS^p$, we can find the associated $\hat{G}_\sigma$ by a pre-specified DAG selection method. 
Further, given a stationary Markov chain $(\sigma_t)_{t \geq 1}$  with transition matrix $\mathbf{P}$, $\{\hat{G}_{\sigma_t}\}_{t \geq 1}$ can be seen as correlated samples drawn from the marginal posterior distribution on the DAG space given in~\eqref{eq:marginal.post}, which is just the pushforward of the marginal posterior distribution on $\bbS^p$ under the mapping $\sigma \mapsto \hat{G}_\sigma$. 

The choice of the neighborhood $\cN(\cdot)$ may affect the mixing of the chain significantly. In order to achieve efficient local exploration, the neighborhood size needs to be small. All the three types of proposals considered are desirable in this regard, since the corresponding neighborhood sizes grow at most quadratically in $p$:  $|\cN_{\mathrm{adj}}(\sigma)| = p-1$, and  $|\cN_{\mathrm{rtp}}(\sigma)| = |\cN_{\mathrm{rrs}}(\sigma)| = p(p-1)/2$. 
However, if the neighborhood size is too small, the chain might get stuck at sub-optimal local modes, where a local mode refers to a state with posterior probability larger than that of any neighboring state.    
We will present a simulation study in Section~\ref{subsec:mixing} which confirms that all three proposals yield good mixing of the sampler for moderately large $p$.   

In general, theoretical analysis of the mixing behavior of order-based MCMC methods is very difficult. 
Existing results on the mixing of MCMC for high-dimensional model selection problems suggest that if the posterior distribution is unimodal and tails decay sufficiently fast, an MCMC sampler is expected to mix rapidly~\citep{yang2016computational, zhou2021complexity, chang2022rapidly}; this intuition is highly similar to the rapid mixing of the algorithms with log-concave targets on continuous spaces~\citep{mangoubi2017rapid, dwivedi2018log}.  
However, to rigorously prove a rapid mixing result for our problem seems very difficult.  
One possible strategy is to assume a permutation $\beta$-min condition~\citep{aragam2019globally}, but such a permutation $\beta$-min condition is very restrictive since it requires all nonzero edge weights to be sufficiently large no matter what topological ordering we assume; in our context, this condition means that $\hat{G}_\sigma$ is equal to $G^*_\sigma$ for any $\sigma \in \bbS^p$.  
Here we choose to consider a contrasting setting where all the edge weights of the true DAG $G^*$ are not too large.  
This is probably more realistic and complements the existing theory, though still being moderately restrictive; see Remark~\ref{rmk:disc.conditions} below.  
We are able to prove that the acceptance probability cannot be extremely small for any state proposed from $ \cN_{\mathrm{adj}}(\cdot)$; see Remark~\ref{rmk:swap}. 
That is, by using adjacent transpositions, the chain is able to escape from any sub-optimal local mode, if there is any, in a relatively short amount of time.  
Observe that for any $\sigma \in \bbS^p$,  $\cN_{\mathrm{adj}}(\sigma)$ is a proper subset of both $\cN_{\mathrm{rtp}}(\sigma)$ and $\cN_{\mathrm{rrs}}(\sigma)$.  
Hence, our result partly explains why all the three proposals appear to work well.  
   
\begin{proposition}\label{prop:lower_bound}
Assume (C\ref{A:eigen}) in Proposition~\ref{prop:freq} and the following conditions hold. 
\begin{enumerate}[({C}1')] 
    \item The true covariance matrix $\Omega^* = \diag(\omega^*, \dots, \omega^*)$ for some universal constant $\omega^* > 0$, and  the edge weights of $G^*$ satisfy 
    \begin{align*}
    \max_{i, j \in [p]} |B^*_{ij}|^2 =  O \left(\frac{\vmax^2 \log p}{\vmin^2 n}\right). 
    \end{align*}   \label{c2.beta.min}
    \item The parameter $\din$ satisfies $d^* \leq \din$ and  
    \begin{align*}
        \din^{2} \frac{\vmax^2 \log p}{\vmin^2 n} \rightarrow 0  \text{ as } n \rightarrow \infty. 
    \end{align*} 
    \label{c2.din}  
\end{enumerate} 
\vspace{-0.5cm}
Let $\cN_{\rm{rev}}(G)$ denote the set of all DAGs that can be obtained by applying one edge reversal to $G$, and $c > 0$ be an arbitrary universal constant. Then, for sufficiently large $n$, 
\begin{align*}
    \max_{\sigma \in \bbS^p} \max_{G_1 \in \cG_p^\sigma (\din)} \max_{G_2 \in \cN_{\mathrm{rev}} (G_1) } \frac{\exp(\phi(G_1))}{\exp(\phi(G_2))} \leq  p^{c \vmax^2/\vmin^3},   
\end{align*}
with probability at least $1 - 6p^{-1}$. 
\end{proposition}
\begin{proof}
See Section~\ref{proof:lower_bound} in the supplementary material.
\end{proof}

\begin{remark}\label{rmk:swap}
To see the implication of this result on the mixing of our order MCMC, consider $\sigma = (1, 2, \dots, p)$, and let $\tau =  \sigma \circ (i, i+1)_{\mathrm{c}} $ for some $i$. 
Recall that we use $\hat{G}_\sigma = \hat{G}_\sigma^{\rm{MAP}}$ where $\hat{G}_\sigma^{\rm{MAP}}$ is defined in~\eqref{def.map}. Hence, $\post(\sigma)/\post(\tau) \leq \exp(\phi( \hat{G}_\sigma)) / \exp(\phi(G'))$ where $G'$ is the DAG that results from reversing the edge $i \rightarrow (i+1)$ of $\hat{G}_\sigma$; if the edge does not exist, then $G' = \hat{G}_\sigma$. 
Assuming $\vmax, \vmin$ are bounded, Proposition~\ref{prop:lower_bound} implies that with high probability $\post(\sigma)/\post(\tau)$ is bounded from above by $p^c$ where $c > 0$ is arbitrary, as long as $G' \in \cG_p^\tau(\din)$.  For the schemes we propose on $\bbS^p$, this further implies that an adjacent transposition proposal has acceptance probability greater than $p^{-c}$, and thus the chain cannot get trapped at a local mode for exponentially many iterations in expectation. 
\end{remark}

\begin{remark}\label{rmk:disc.conditions}
The purpose of Proposition~\ref{prop:lower_bound} is to theoretically analyze the posterior landscape when we probably do not have posterior concentration at the true model and Proposition~\ref{prop:freq} no longer holds. In particular, Proposition~\ref{prop:lower_bound} does not require  any assumption on the hyperparameters of our model, so the nonzero entries in $B^*$ may or may not be detected,  depending on the choice of $c_0$. Condition (C1') essentially requires that no signal size has a strictly larger order than the detection threshold given in condition (C3) of Proposition~\ref{prop:freq}. This is restrictive but arguably represents a scenario  of more practical interest than Proposition~\ref{prop:freq}, since in reality signals of small or moderate sizes are common. 
It is possible to construct a scenario where the assumptions of Propositions~\ref{prop:freq} and~\ref{prop:lower_bound} both hold. 
For example, assume $d^* = O(1)$, which is referred to as the ultra-high sparsity regime in the literature~\citep{van2013ell}. 
Then we can set $\din = O(1)$, which implies that  we can choose $c_0 = O(1)$ to satisfy condition  (C2) of Proposition~\ref{prop:freq}. 
Assuming $\vmax, \vmin$ are bounded for convenience, in order to satisfy condition (C3)  of Proposition~\ref{prop:freq} and condition (C1') of Proposition~\ref{prop:lower_bound}, we just need to require that the order of any nonzero entry $B^*_{ij}$ is exactly given by $n^{-1} \log p $. 
\end{remark}

\subsection{Iterative top-down initialization}\label{subsec:ITD}
Standard theory yields that the Markov chain defined in~\eqref{eq:transition} converges to the marginal posterior distribution on $\bbS^p$ in total variation distance regardless of the initial state.  
However, the actual mixing rate of the chain we observe depends on the initial state~\citep[Proposition 1]{sinclair1992improved}, and in general, it is desirable to start the chain at a state with reasonably high posterior probability. 
Since the size of   $\bbS^p$ grows super-exponentially in $p$, choosing a warm start for our sampler can significantly improve the performance of posterior estimation with MCMC samples. 
We propose an initialization method for our order MCMC sampler, called iterative top-down, which aims to quickly find the topological ordering of the true data-generating DAG $G^*$.

\begin{algorithm}[b!]
\caption{Score-based top-down algorithm}
\label{alg:STD}
\KwInput{A positive vector 
$\SSE= (\SSE_1, \dots, \SSE_p)$ (for all displayed algorithms, we assume the data  $X$ and parameters $(c_0, \gamma, \alpha, \kappa, \din)$ are given).}
$\hat{\sigma} \leftarrow\, \argmin_{j \in [p]} \SSE_j$\\ 
\While{$|\hat{\sigma}| < p$}
{
\For{$j \in [p] \backslash \hat{\sigma}$}
{
$S \leftarrow\, \argmax_{S_j \subset \hat{\sigma} \colon  |S_j| \leq \din} \phi_j(S_j, \sum_{i\neq j} \SSE_i )$   \\
\tcp{$\phi_j\left(S, R\right)=-|S| \log \left\{p^{c_0} \sqrt{(1+\alpha / \gamma)}\right\}-\frac{\alpha p n+\kappa}{2} \log \left(R +X_j^{\mathrm{T}} \Phi_S^{\perp} X_j\right)$}
$\SSE_j \leftarrow\, X_j^\T \oproj_S X_j$ 
} 
$j_0 \leftarrow\, \argmin_{j \in [p] \backslash \hat{\sigma}} \SSE_j$ \\
$\hat{\sigma}\leftarrow\, (\hat{\sigma}, j_0)$
}
\KwOutput{An ordering $\hat{\sigma}$, a vector of estimated residual sums of squares $\mathrm{RSS}$.}
\end{algorithm}

\begin{algorithm}[t!]
\caption{Iterative top-down algorithm}
\label{alg:ITD}  
$(\hat{\sigma}^{\mathrm{ITD}}, \SSE) \leftarrow\,  \mathrm{STD}(X_1^\T X_1, \dots, X_p^\T X_p)$ \tcp{STD refers to Algorithm~\ref{alg:STD}} 
\While{$1$}{
$(\tilde{\sigma}, \SSE') \leftarrow\, \mathrm{STD}(\SSE)$ \\
\If{$\hat{\sigma}^{\mathrm{ITD}} \neq \tilde{\sigma}$}{
 $\SSE \leftarrow\, \SSE' $ \\ $\hat{\sigma}^{\mathrm{ITD}}\leftarrow\, \tilde{\sigma}$
}\Else{\textbf{return} $\hat{\sigma}^{\mathrm{ITD}}$}
} 
\KwOutput{An ordering $\hat{\sigma}^{\mathrm{ITD}}$}
\end{algorithm}

Our method is based on the top-down method  proposed by~\citet{chen2019causal}, which we now briefly explain.  
We say a node in a DAG is a source if the node has no parents. 
If the data is generated according to~\eqref{eq:ln.str.eq}, due to the equal variance assumption, a  source node always has the smallest marginal variance, and any node with at least one parent has a strictly larger marginal variance. 
The top-down method first identifies a source node of $G^*$, which always exists, sets it to $\hat{\sigma}(1)$ and then removes it from $G^*$. The resulting subgraph is also a DAG, and thus we can set $\hat{\sigma}(2)$ to a source node of this subDAG; how to identify the source node is explained in the next paragraph. 
Repeating this procedure $p$ times, we obtain $\hat{\sigma}$,  the top-down estimator for the ordering.   

Suppose that in the first $k$ iterations of the top-down method we have identified $\sigma(j) = j$ for $j = 1, \dots, k$. Then in the $(k+1)$-th iteration, we need to estimate the variance of each remaining node that cannot be explained by the first $k$ nodes, and pick the node with the smallest unexplained variance, which we infer as a source node of the subDAG of the remaining $p - k$ nodes. 
\citet{chen2019causal} estimated the unexplained variance of the node $j$ (assuming $j > k$) by $\min_{S \subseteq [k], |S| = \din} X_j^\T \oproj_{S} X_j$, but they noted that a variable selection procedure may be applied as well.  
Since our purpose is to find a warm start for our order MCMC sampler, we estimate the unexplained variance of a node by performing a variable selection procedure that aims to maximize the score~\eqref{eq:post_score}. 
One caveat is that since our score is non-decomposable, when inferring the parent set of node $j$, we need to know the residual sums of squares of all the other $p - 1$ nodes. 
This motivates us to propose the iterative top-down method, detailed in Algorithm~\ref{alg:ITD}, which iteratively applies the top-down procedure and updates all the $p$ residual sums of squares. 
We prove below that under a condition similar to that of~\citet[Theorem 2]{chen2019causal},  the iterative top-down algorithm identifies an ordering consistent with $G^*$ with high probability. In our simulation studies, we observe that the algorithm usually converges within $5$ iterations.   

\begin{theorem}\label{thm:ITD}
Suppose the conditions in Proposition~\ref{prop:freq} hold, and let $\epsilon \in (0, 1)$. 
If 
\begin{align*}
    n > \{\vmax (\din + 1)(\vmin + 3 \omega^*(1+1/\betamin)) / \vmin^2 \}^2 3200 ( \log(p^2 -p) - \log(\epsilon/4)),
\end{align*}
then for sufficiently large $n$, Algorithm~\ref{alg:ITD} returns an ordering in $[\sigma^*]$ with probability at least $1 - \epsilon$.
\end{theorem}
\begin{proof}
See Section~\ref{proof:ITD} in the supplementary material.
\end{proof}
 
\subsection{Reducing variance of edge estimation}\label{subsec:RB} 
One potential limitation of our order MCMC sampler is that it does not take into account the uncertainty in DAG selection with given ordering. So we propose to estimate edge posterior inclusion probabilities using a conditioning scheme. 
Let $\sigma^{(t)}$ denote the $t$-th sample from our order MCMC sampler, and $\Gamma^{(t)}$ denote the adjacency matrix of the DAG $G^{(t)} = \hat{G}_{\sigma^{(t)}}$ such that $\Gamma^{(t)}_{ij} = 1$ if  $\{ i \rightarrow j \} \in G^{(t)}$ and $\Gamma^{(t)}_{ij} = 0$ otherwise. 
The posterior inclusion probability of edge $i \rightarrow j$  can be estimated by $ T^{-1} \sum_{t=1}^T  \Gamma_{ij}^{(t)}$ where $T$ denotes the number of MCMC samples. 
To improve this estimator,  for each pair $(\sigma^{(t)}, G^{(t)})$, we calculate  $\hat{\Gamma}^{(t)} = \hat{\Gamma}(\sigma^{(t)}, G^{(t)})$, where the function $\hat{\Gamma}$ is given by 
\begin{equation}\label{eq:RB} 
    \hat{\Gamma}_{ij}(\sigma, G)  = \frac{ e^{ \phi ( G  \cup \{i \rightarrow j\} ) }  }{e^{ \phi ( G \cup \{i \rightarrow j\} ) } + e^{ \phi (G  \setminus \{i \rightarrow j \} )} } \ind_{P_j^{\sigma }}(i), \quad \forall \, i, j \in [p]. 
\end{equation}   
We can now estimate the posterior inclusion probability of edge $i \rightarrow j$  by $\hat{\Gamma}_{ij}^{\mathrm{RB}} = T^{-1} \sum_{t=1}^T  \hat{\Gamma}_{ij}^{(t)}$. 
The superscript RB indicates that, in a general sense, this can be seen as a Rao-Blackwellized-type estimator~\citep{robert2021rao}. 
In our numerical experiments, we find  this scheme helps reduce the variance of edge posterior inclusion probability estimates.   

\section{Simulation studies}\label{sec:simulation} 

\subsection{Mixing behavior}\label{subsec:mixing}
We first present a numerical example which illustrates how the choice of neighborhood and 
score equivalence property affect the mixing behavior of order MCMC samplers. We generate a 20-node random DAG $G^*$ where any two distinct nodes are connected by an edge with  probability $0.1$, and sample the edge weight $B^*_{ij}$ for each $i \rightarrow j$ in $G^*$ uniformly from $[-1,-0.5] \cup [0.5, 1]$. 
Then, we simulate the data matrix $X$ using the structural equation model in~\eqref{eq:ln.str.eq} with $n = 1,000$ and error variance $\omega^* = 1$. 
 
We implement the order MCMC sampler described in Section~\ref{sec:implement} with $\cN =  \cN_{\mathrm{adj}}, \cN_{\mathrm{rtp}}$ or $\cN_{\mathrm{rrs}}$.  
To impartially compare the three types of proposal, we need to take into account the computational complexity of sampling from each type of neighborhood.   
Consider a proposal move from $\sigma$ to $\sigma' = \sigma \circ (i, j)_{\mathrm{c}}$ for some $i < j$. 
In Section~\ref{subsubsec:effective} of the supplementary material,  we present a stepwise procedure for selecting the parent set of a given node in Algorithm~\ref{alg:nodewise_step}, and 
describe how to efficiently obtain $\hat{G}_{\sigma'}$ from $\hat{G}_\sigma$ by applying Algorithm~\ref{alg:nodewise_step} at nodes $\sigma(i), \sigma(i+1), \dots, \sigma(j)$. 
Hence, an adjacent transposition always requires performing Algorithm~\ref{alg:nodewise_step} at two nodes, while for a random transposition, which randomly samples $\sigma'$ from $\cN_{\mathrm{rtp}}(\sigma)$ with equal probability, on average we need to perform Algorithm~\ref{alg:nodewise_step}  at $(p+4)/3 \approx p / 3$ nodes, and the same holds true for a random-to-random shuffle.  
So, when we run the sampler defined in~\eqref{eq:transition} for $T$ iterations, we say the effective number of iterations is $2 T$ if $\cN = \cN_{\mathrm{adj}}$, and $pT/3$ if $\cN = \cN_{\mathrm{rtp}}$ or $\cN = \cN_{\mathrm{rrs}}$. 
We let the effective number of iterations be $10,000$ for all three samplers in our simulation; that is,  we run our sampler with $\cN = \cN_{\mathrm{adj}}$ for $5,000$ iterations, and the samplers with $\cN = \cN_{\mathrm{rtp}}$ and $\cN = \cN_{\mathrm{rrs}}$ for $1,500$ iterations.   
We plot the trajectories for 30 runs with random initialization in the panels (a), (b), (c) of Fig.~\ref{fig:mixing},  from which we see that all three proposals work well. 
We have also tried $n = 100$ and observed good mixing performance, probably because with a smaller sample size the posterior distribution tends to be flatter~\citep{agrawal2018minimal}; we display the result in Section~\ref{subsec:mixing_supp} of the supplementary material. 
Given that adjacent transposition appears to yield the best mixing, it will be used for all the remaining numerical studies. 
 
To compare our method with a score equivalent procedure, we consider the following posterior score, which is decomposable and yields the same value for Markov equivalent DAGs, 
\begin{align}\label{eq:equivalent_score}
    \phi_{\mathrm{eq}} (G) =  -|G| c_0 \log p  - \frac{|G|}{2} \log [ (1 + \alpha/\gamma) ] -  \frac{\alpha n + \kappa}{2} \sum_{j=1}^p \log  \left(\SSE_j(G) \right). 
\end{align}
This score can be derived by a slight modification of our model: instead of assuming equal error variances, use an error variance parameter $\omega_j$ for each $\se_j$ in~\eqref{eq:ln.str.eq}  and put an inverse-gamma prior on $\omega_j$~\citep{zhou2021complexity}.   
To sample from the corresponding posterior distribution, we use the minimal I-MAP MCMC sampler of~\citet{agrawal2018minimal}, which is also a Metropolis-Hastings algorithm defined on the order space and proposes moves from the adjacent transposition neighborhood $\cN_{\mathrm{adj}}(\cdot)$; compared with our method, the main difference is that the minimal I-MAP MCMC uses conditional independence tests to find $\hat{G}_\sigma$. 
We run the minimal I-MAP MCMC  for 10,000 iterations, and plot 30 trajectories with random initialization in Fig.~\ref{fig:mixing}(d). 
Comparing it with Fig.~\ref{fig:mixing}(a), we see that our sampler with non-decomposable score mixes better in the sense that all 30  trajectories are able to visit some $\sigma \in [\sigma^*]$, while the minimal I-MAP MCMC may get stuck at local modes depending on the initialization. 
As we have explained in Section~\ref{sec:intro}, score equivalence is likely to make the posterior distribution on the order space (or the DAG space) difficult to explore due to the existence of large equivalence classes. This simple numerical study verifies that the use of identifiability conditions does simplify the posterior distribution so that MCMC samplers tend to mix faster. 
In Section~\ref{subsec:mixing_supp} of the supplementary material, we show that the same observation can still be made if we simulate $X$ using unequal error variances.

\begin{figure}[t!]
    \centering 
    \subfigure[]{\includegraphics[width=0.24\textwidth]{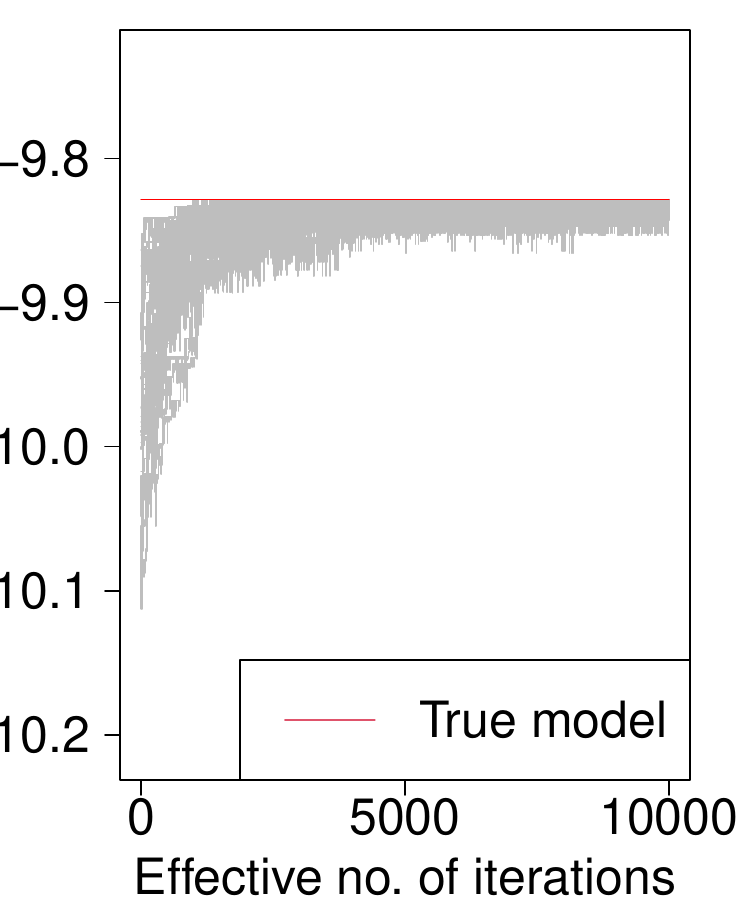}}
    \subfigure[]{\includegraphics[width=0.24\textwidth]{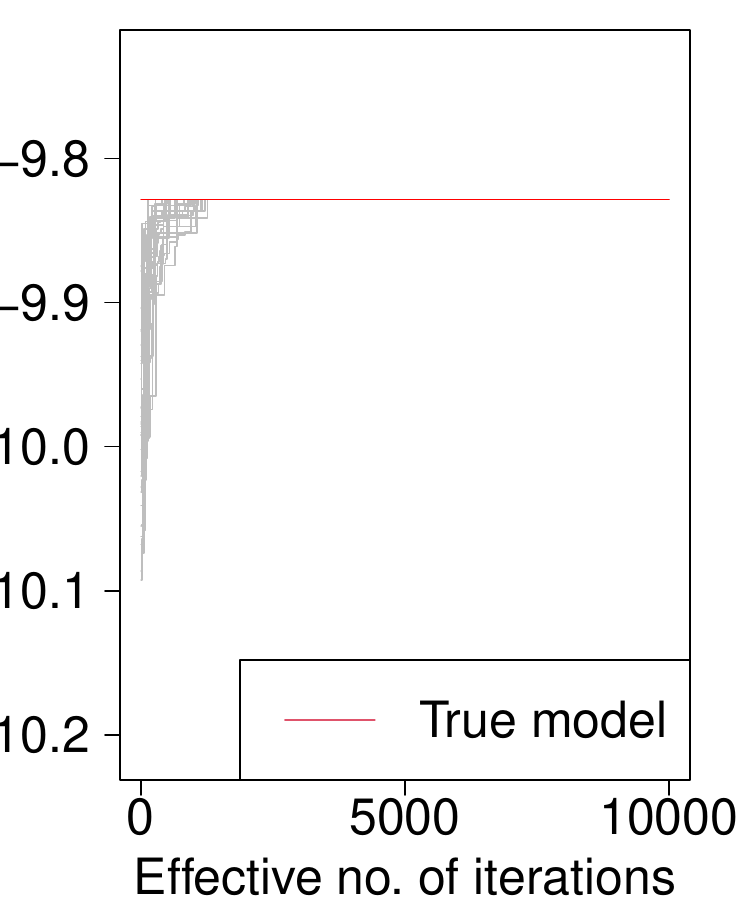}}
    \subfigure[]{\includegraphics[width=0.24\textwidth]{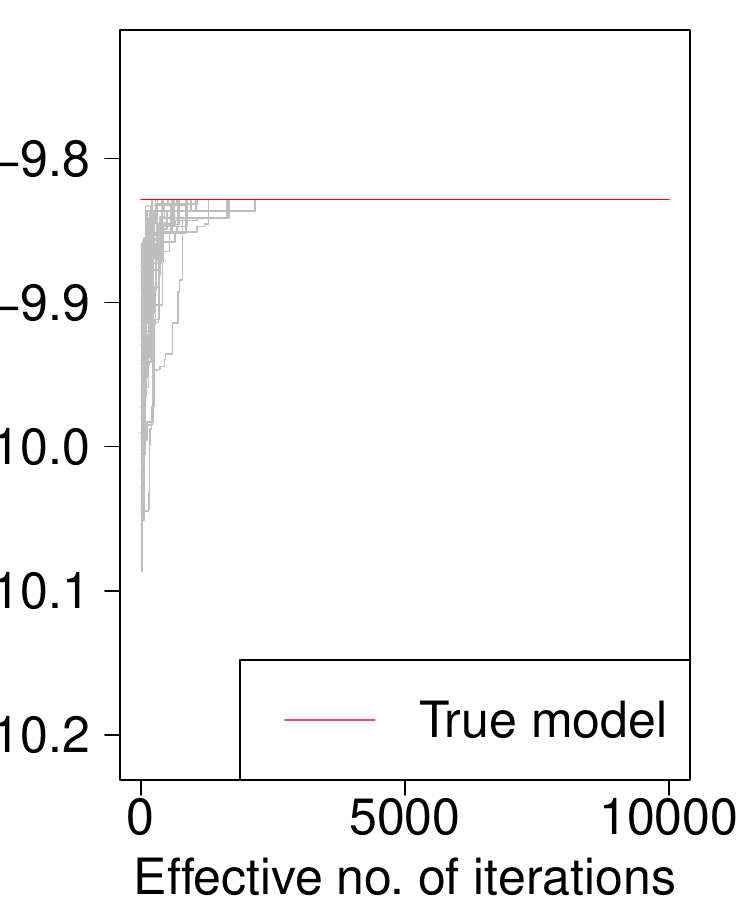}}
    \subfigure[]{\includegraphics[width=0.24\textwidth]{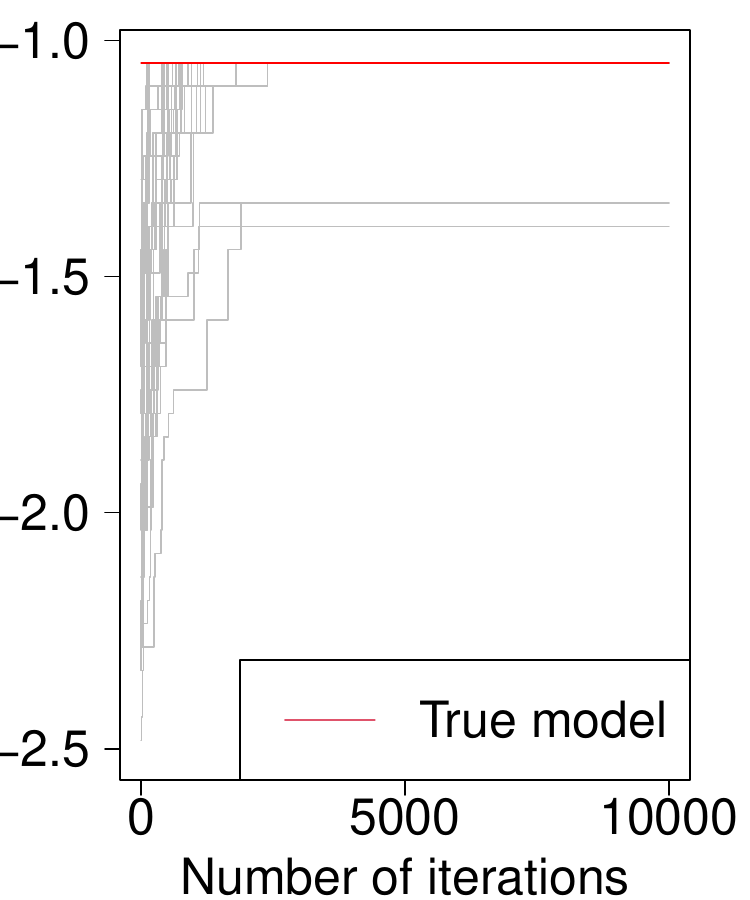}}
    \caption{Log posterior probability $\times 10^{-4}$ versus the  effective number of iterations in 30 MCMC runs with random initialization. The red line gives the log posterior probability of the true ordering $\sigma^*$. Panel (d) is for the minimal I-MAP MCMC with decomposable score. Panels (a), (b), (c) correspond to our method with three types of proposals: (a) adjacent transposition, (b) random transposition, (c) random-to-random shuffle. We have checked that, for our method, all $30 \times 3 = 90$ runs have successfully reached the red line. } 
    \label{fig:mixing}
\end{figure}

\subsection{Performance evaluation}\label{subsec:compare} 
We conduct simulation studies to empirically evaluate the performance of the proposed order MCMC sampler.  
We still use $G^*$ to denote the true $p$-node DAG that governs the data generating process described in~\eqref{eq:ln.str.eq} and let $\Gamma^*$ be its adjacency matrix. 
Let $\hat{G}$ and $\hat{\Gamma}$ denote the corresponding estimators, and for our method, we always use $\hat{\Gamma} = \hat{\Gamma}^{\mathrm{RB}}$ where $\hat{\Gamma}^{\mathrm{RB}}$ is defined in Section~\ref{subsec:RB}. 
Entries of $\hat{\Gamma}$ are edge posterior inclusion probability estimates and thus take value in $[0, 1]$, while  $\Gamma^* \in \{0, 1\}^{p \times p}$.  
We use four performance metrics to evaluate an estimator. The structural Hamming distance (HD) between $G^*$ and $\hat{G}$ is the number of different edges between $G^*$ and $\hat{G}$, which equals $\sum_{i,j} |\Gamma^*_{ij} - \hat{\Gamma}_{ij}|$. 
False negative rate (FNR) and false discovery rate (FDR) are defined as ($\sum_{i,j} \Gamma^*_{ij} (1- \hat{\Gamma}_{ij})) / |G^*| \times 100 \% $ and $(\sum_{i,j} (1 - \Gamma^*_{ij}) \hat{\Gamma}_{ij}) / |\hat{G}| \times 100  \% $, respectively.  
The fourth metric, percentage of flipped edges, is calculated as $(\sum_{i,j} \Gamma^*_{ji} \hat{\Gamma}_{ij}) / |G^*| \times 100\%$.  
We compare our method with two competing algorithms, the top-down method~\citep{chen2019causal} and the algorithm of~\citet{ghoshal2018learning}, and we follow the suggestions given in the two papers to choose the tuning parameters. These two algorithms are reported to have better performance than others. 
For our method, we fix $\alpha = 0.99, \gamma = 0.01, \kappa = 0$, $c_0 = 3$ and run MCMC for $3,000$ iterations for each simulated data set and discard the first 1,500 samples as burn-in.   
We always use the following procedure to generate the true DAG $G^*$. 
We fix the true ordering to be $\sigma^* = (1, \dots, p)$, and for each pair $(i, j)$ such that $i < j$,  we add edge $i \rightarrow j$ to $G^*$ with probability $p_{\mathrm{edge}} = 3/(2p-2)$. Hence, the expected number of edges of $G^*$ is $3p/4$. The DAG $G^*$ is resampled for each simulated data set.

\begin{table}[t!]
\begin{adjustbox}{width=1\textwidth}
\small
\begin{tabular}{cccccccc}
Method   & Signal & \multicolumn{3}{c}{$\mathrm{Uniform}([-1, -0.3] \cup [0.3, 1])$} & \multicolumn{3}{c}{$\mathrm{Uniform}([-1, -0.1] \cup [0.1, 1])$} \\ 
   & $n$      & 100               & 500             & 1000            & 100               & 500             & 1000            \\ 
Proposed 
         & HD     & \text{10.0$\pm$0.5}     & \text{0.8$\pm$0.2}    & 0.1$\pm$0.1      & \text{13.9$\pm$0.7}     & 5.2$\pm$0.3      & \text{3.0$\pm$0.3}   \vspace{-1mm}  \\
         & FNR & \text{33.3$\pm$1.5}     & 1.6$\pm$0.4     & 0.2$\pm$0.1     & \text{47.6$\pm$1.7}     & 16.4$\pm$1.1     & 8.4$\pm$1.0    \vspace{-1mm}  \\
         & FDR    & \text{3.2$\pm$0.8}      & \text{1.4$\pm$0.4}    & \text{0.2$\pm$0.1}    & \text{2.4$\pm$0.6}      & \text{2.6$\pm$0.5}    & \text{2.5$\pm$0.4}   \vspace{-1mm}  \\
         & Flip   & \text{1.9$\pm$0.5}      & \text{1.2$\pm$0.3}    & 0.2$\pm$0.1      & \text{1.1$\pm$0.3}      & \text{2.3$\pm$0.5}    & \text{2.2$\pm$0.4}   \vspace{-1mm}  \\
         & Time   & 13.3$\pm$0.2       & 13.6$\pm$0.2     & 13.3$\pm$0.2     & 12.3$\pm$0.2       & 13.2$\pm$0.2     & 13.4$\pm$0.2      \\
TD      
         & HD     & 11.9$\pm$0.8       & 1.5$\pm$0.4      & 0.3$\pm$0.2      & 16.0$\pm$0.9       & 5.9$\pm$0.6      & 4.1$\pm$0.6     \vspace{-1mm}  \\
         & FNR & 37.8$\pm$1.8       & 2.3$\pm$0.5     & 0.3$\pm$0.2     & 52.5$\pm$1.7       & 15.0$\pm$1.3     & 8.2$\pm$0.9   \vspace{-1mm}   \\
         & FDR    & 6.4$\pm$1.3        & 2.8$\pm$0.7      & 0.7$\pm$0.4      & 7.0$\pm$1.4        & 5.9$\pm$1.0      & 5.8$\pm$1.1   \vspace{-1mm}    \\
         & Flip   & 3.6$\pm$0.8        & 1.8$\pm$0.5      & 0.3$\pm$0.2      & 2.9$\pm$0.7        & 4.1$\pm$0.7      & 4.1$\pm$0.6   \vspace{-1mm}    \\
         & Time   & 0.6$\pm$0.0        & 0.5$\pm$0.0      & 0.5$\pm$0.0      & 0.5$\pm$0.0        & 0.6$\pm$0.0      & 0.5$\pm$0.0      \\
LISTEN  
         & HD     & 12.6$\pm$0.7       & 2.1$\pm$0.5      & 0.9$\pm$0.4      & 16.3$\pm$0.9       & 6.5$\pm$0.6      & 4.2$\pm$0.6    \vspace{-1mm}   \\
         & FNR & 39.5$\pm$1.7       & 3.1$\pm$0.7     & 1.0$\pm$0.4     & 52.2$\pm$1.7       & 15.9$\pm$1.1     & 8.9$\pm$1.0   \vspace{-1mm}   \\
         & FDR    & 7.6$\pm$1.5        & 3.9$\pm$1.1      & 1.8$\pm$0.8      & 8.7$\pm$1.7        & 7.0$\pm$1.1      & 5.5$\pm$1.0   \vspace{-1mm}    \\
         & Flip   & 3.6$\pm$0.7        & 2.6$\pm$0.7      & 1.0$\pm$0.4      & 3.3$\pm$0.7        & 4.6$\pm$0.7      & 4.3$\pm$0.7   \vspace{-1mm}    \\
         & Time   & 0.5$\pm$0.0        & 0.5$\pm$0.0      & 0.6$\pm$0.0      & 0.6$\pm$0.0        & 0.6$\pm$0.0      & 0.5$\pm$0.0      \\ 
\end{tabular}
\end{adjustbox}
\caption{Uniform signal case with $p = 40$. TD and LISTEN refer to the top-down algorithm and the algorithm of~\citet{ghoshal2018learning}, respectively. Each entry gives mean $\pm$ 1 standard error. Time is measured in seconds.}
\label{table:uniform}
\end{table}

We first generate the data from the structural equation models given in~\eqref{eq:ln.str.eq}. We fix $p = 40$, set $\omega^* = 1$, and draw the edge weight  $B^*_{ij}$ for each edge $i \rightarrow j$ in $G^*$ independently from some distribution $F$.   
We let sample size $n$ be $100, 500$ or $1,000$, and repeat 30 times for each choice. 
In Table~\ref{table:uniform},  we present the result for $F$ being the uniform distribution on $[-1, -0.3]\cup [0.3, 1]$ and that for $F$ being the uniform distribution on $[-1, -0.1]\cup [0.1, 1]$. 
The result for $F$ being the standard Gaussian distribution is displayed in Section~\ref{subsec:compare_supp} in the supplementary material. 
Table~\ref{table:uniform} shows that our method outperforms the other two methods in all settings by any of the four performance metrics, and in most cases, our method is better by a margin of at least one standard error.

\begin{table}[!h]
\centering
\begin{tabular}{ccclcccc}
$p$   & $n$ & $d$     &          & FNR         & FDR        & Flip       & Time          \vspace{1mm}  \\
7   & 60                    & 1.549 &&  22.9$\pm$3.8 & 8.9$\pm$2.5 & 4.8$\pm$1.4 & 2.3$\pm$0.1     \\
14  & 90                    & 1.897 && 11.3$\pm$1.8 & 2.7$\pm$0.8 & 2.2$\pm$0.7 & 3.7$\pm$0.1     \\
28  & 120                   & 2.191 &&  4.6$\pm$0.7  & 0.6$\pm$0.2 & 0.4$\pm$0.2 & 6.2$\pm$0.1     \\
56  & 150                   & 2.449 && 2.7$\pm$0.3  & 0.5$\pm$0.2 & 0.4$\pm$0.2 & 15.1$\pm$0.2    \\
112 & 180                   & 2.683 &&  1.2$\pm$0.2  & 0.2$\pm$0.1 & 0.1$\pm$0.0 & 64.4$\pm$1.2    \\
224 & 210                   & 2.898 &&  0.8$\pm$0.1  & 0.2$\pm$0.1 & 0.1$\pm$0.0 & 377.2$\pm$5.5   \\
448 & 240                   & 3.098 &&  0.5$\pm$0.1  & 0.1$\pm$0.1 & 0.1$\pm$0.0 & 2896.4$\pm$48.1
\end{tabular}
\caption{Simulation under a high-dimensional regime. Each entry gives mean $\pm$ 1 standard error.  Time is measured in seconds.} 
\label{table:highdim}
\end{table}

Next, we examine the performance of our method with varying $n$ and $p$. To emulate a high-dimensional asymptotic regime where $n$ grows linearly and $p$ increases exponentially, we consider $7$ settings where $n = 30(k+1)$ and $p =7 \cdot 2^{k-1}$ in the $k$-th setting. When $k=6$ or $7$, we have $p > n$. We generate $G^*$ with $p_{\mathrm{edge}} =  d / (p-1)$, where $d = 0.2 \sqrt{n}$ is the expected number of neighbors for each node. We sample the edge weight $B^*_{ij}$ for each $i \rightarrow j$ in $G^*$ uniformly from $[-1,-0.5] \cup [0.5, 1]$ and set the error variance $\omega^* = 1$. We use the same values for $\alpha, \gamma, \kappa, c_0$ and run 3,000 MCMC iterations with 1,500 discarded samples as burn-in. The result of 30 replicates is summarized in Table~\ref{table:highdim}, from which we see that FNR, FDR and flip rates all decrease as $p$ increases.   
Further, the method is considerably scalable as it completes 3,000 iterations within an hour even when $p = 448$.

Lastly, we generate $X$ by assuming each $\se_j$ in the structural equation models~\eqref{eq:ln.str.eq} has variance $\omega_j$; thus, the equal variance assumption is violated.    
We repeat the simulation study presented in the left column of Table~\ref{table:uniform} by sampling $\omega_j$ independently from the uniform distribution on $[0.7, 1.3]$ for each $j$, and we observe that the advantage of the proposed method is more significant; see Section~\ref{subsec:compare_supp} in the supplementary material for the result. 
To further examine how the heterogeneity of error variances affects the performance of our method, we fix  $n = 500$ and $p = 40$, and sample $\omega_j$ from $\mathrm{Uniform}([1-b, 1+b])$ for $b = 0, 0.1, \dots, 0.9$. We plot the distribution of the HD metric over 30 replicates against $b$ in Fig.~\ref{fig:misspecification}.  
The proposed order MCMC sampler again performs uniformly better than competing algorithms. Besides, our method appears to be more robust, especially when $b$ is not too large,  which is probably due to the use of model averaging in Bayesian posterior inference.  

\begin{figure}[t!]
    \centering
    \includegraphics[width=1\textwidth]{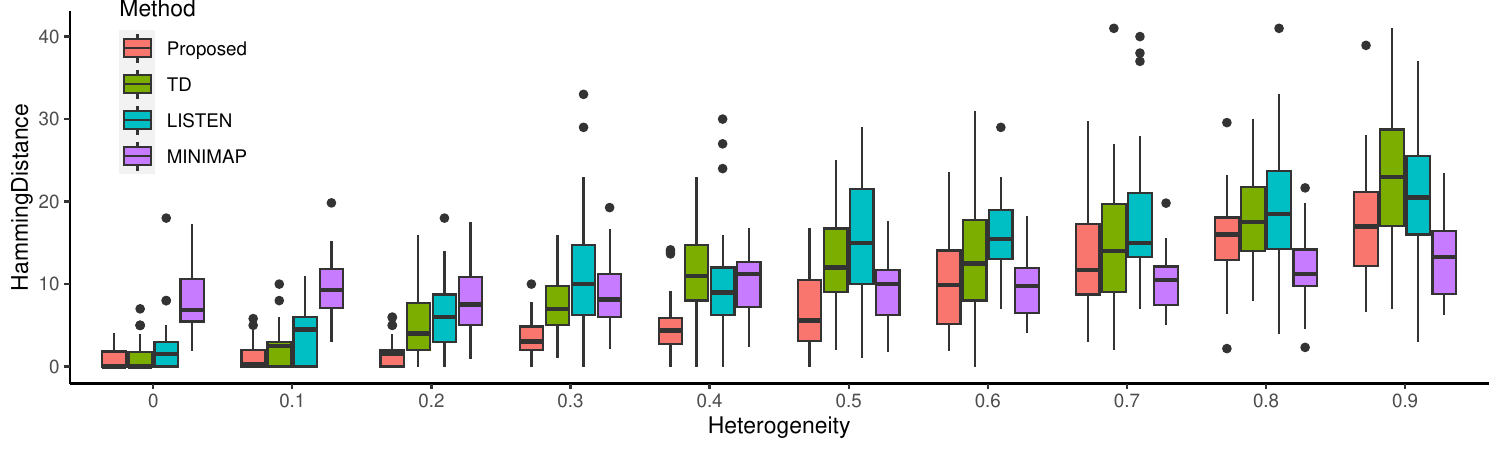}
    \caption{Boxplots for heterogeneous error variance case with $n = 500, p = 40$. We sample error variances from $\mathrm{Uniform}([1 -b, 1+b])$ for $b = 0, 0.1, \dots, 0.9$ and nonzero edge weights from $\mathrm{Uniform}([-1, -0.3] \cup [0.3, 1])$. The $x$-axis indicates the heterogeneity parameter $b$, and the $y$-axis represents the Hamming distance between the estimated DAG and $G^*$. 
    MINIMAP is the minimal I-MAP MCMC that uses score~\eqref{eq:equivalent_score}, and thus it is score equivalent.}
    \label{fig:misspecification} 
\end{figure}

\subsection{Quantification of the bias caused by the equal variance assumption}\label{subsec:bias} 

When the true data generating process does not satisfy the equal variance assumption, our method is expected to have some bias. This is confirmed in Fig.~\ref{fig:misspecification}, from which we see that HD increases with the heterogeneity of error variances. 
For comparison, we have also included in Fig.~\ref{fig:misspecification} the score-equivalent minimal I-MAP MCMC with  score given by~\eqref{eq:equivalent_score}. 
Since this score does not encode the equal variance assumption, the minimal I-MAP MCMC sampler  cannot determine the direction of an edge if reversing it yields another Markov equivalent DAG. 
This can be clearly seen from Fig.~\ref{fig:misspecification}: the performance of the minimal I-MAP MCMC does not change significantly with the heterogeneity level $b$, and it always has HD away from zero.    
When the heterogeneity level $b = 0.6$, which implies that the ratio between the maximum and minimum error variances can be as large as $4$, the minimal I-MAP MCMC has a comparable performance to our method, and when $b \geq 0.7$, the minimal I-MAP MCMC performs better.

\begin{table}[t]
\centering
\begin{tabular}{cccccccc}
Method         &      & $b =$ 0                   & $b =$ 0.3         & $b =$ 0.5         & $b =$ 0.7         & $b =$ 0.9   &  $\mathrm{IG}(3, 2)$    \\\vspace{-1mm} 
Proposed & HD   & 0.1$\pm$0.0   & 0.5$\pm$0.2  & 1.6$\pm$0.4  & 2.1$\pm$0.5  & 2.6$\pm$0.5 &  3.3$\pm$0.8\\\vspace{-1mm} 
         & SHD  & 0.0$\pm$0.0   & 0.1$\pm$0.0  & 0.3$\pm$0.1  & 0.4$\pm$0.1  & 0.4$\pm$0.1 & 0.5$\pm$0.2 \\
         & Flip & 1.1$\pm$0.7    & 4.0$\pm$1.5  & 10.0$\pm$2.4 & 13.4$\pm$3.0 & 18.5$\pm$3.9 & 21.1$\pm$4.1 \\\vspace{-1mm} 
MINIMAP  & HD   & 3.0$\pm$0.3    & 2.5$\pm$0.2  & 2.6$\pm$0.3  & 2.6$\pm$0.2  & 2.7$\pm$0.2 & 2.6$\pm$0.2 \\\vspace{-1mm} 
         & SHD  & 0.5$\pm$0.1    & 0.3$\pm$0.1  & 0.4$\pm$0.1  & 0.4$\pm$0.1  & 0.4$\pm$0.1 & 0.3$\pm$0.1\\\vspace{-1mm} 
         & Flip & 23.0$\pm$2.9  & 22.3$\pm$3.1 & 23.4$\pm$3.2 & 23.7$\pm$3.2 & 24.7$\pm$3.1 & 23.7$\pm$3.0
\end{tabular}
\caption{Analysis of the posterior distributions for $p=7$.  MINIMAP uses score~\eqref{eq:equivalent_score}, and thus it is score equivalent.
The posterior inclusion probabilities of all edges are calculated exactly for both methods. The error variances are sampled from  $\mathrm{Uniform}([1-b, 1+b])$ or $\text{inverse-Gamma}(3,2)$.  
Each entry gives mean $\pm$ 1 standard error.}
\label{table:bias}
\end{table}

In order to better quantify the bias of our method,   we exactly calculate the matrix $\Gamma$ whose $(i,j)$-th element gives the  posterior inclusion probability of the edge $i\rightarrow j$. We fix $p = 7$ so that we can enumerate all possible orderings, and the exact posterior inclusion probabilities corresponding to scores \eqref{eq:post_score} and \eqref{eq:equivalent_score} can be calculated as \vspace{-3mm}
\begin{align*}
    \Gamma_{ij} = \sum_{\sigma \in \bbS^p}  \frac{e^{\phi(\hat{G}_\sigma)}}{\sum_{\sigma \in \bbS^p} e^{\phi(\hat{G}_\sigma)} }\ind(\{i \rightarrow j\} \in \hat{G}_\sigma),  \quad
    \Gamma_{ij}^{\mathrm{eq}} = \sum_{\sigma \in \bbS^p}  \frac{e^{\phi_\mathrm{eq}(\hat{G}_\sigma^{\mathrm{M}})}}{\sum_{\sigma \in \bbS^p} e^{\phi_\mathrm{eq}(\hat{G}_\sigma^{\mathrm{M}})} }\ind(\{i \rightarrow j\} \in \hat{G}_\sigma^{\mathrm{M}}),
\end{align*}
where $\hat{G}_\sigma$ and $\hat{G}_\sigma^{\mathrm{M}}$ are the estimated DAGs given an ordering $\sigma$ by our method and the minimal I-MAP method, respectively. We set $n = 100 \, p$ and $p_\mathrm{edge} = 3/(2p-2)$, sample nonzero edge weights from $\mathrm{Uniform}([-1, -0.3] \cup [0.3, 1])$,  and sample error variances from $\mathrm{Uniform}([1-b, 1+b])$ 
and the inverse gamma distribution $\mathrm{IG}(a_1, a_2)$. 
We set $a_1 = 3$, which is the smallest integer that yields a finite variance, and set $a_2 = 2$ so that  the expected value   equals 1. 
We generate 30 replicates for each simulation setting. 
In Table~\ref{table:bias}, we report  three metrics, HD, Flip, and the Hamming distance for skeletons (SHD); recall that the skeleton of a DAG is the undirected graph obtained by undirecting all edges. SHD is consistently close to zero throughout the simulation settings, which implies that the true skeleton is correctly identified by both methods regardless of the heterogeneity level $b$. Notably, in all the settings considered, even when $b = 0.9$ or in the inverse-gamma case, our method has a smaller flip rate than the minimal I-MAP method. That is, imposing the equal variance assumption does not increase the flip rate compared to a score-equivalent approach, which suggests that the computational gain resulting from this assumption is essentially obtained for free in this example.

\begin{figure}[t]
    \centering
    \includegraphics[width=0.60\textwidth]{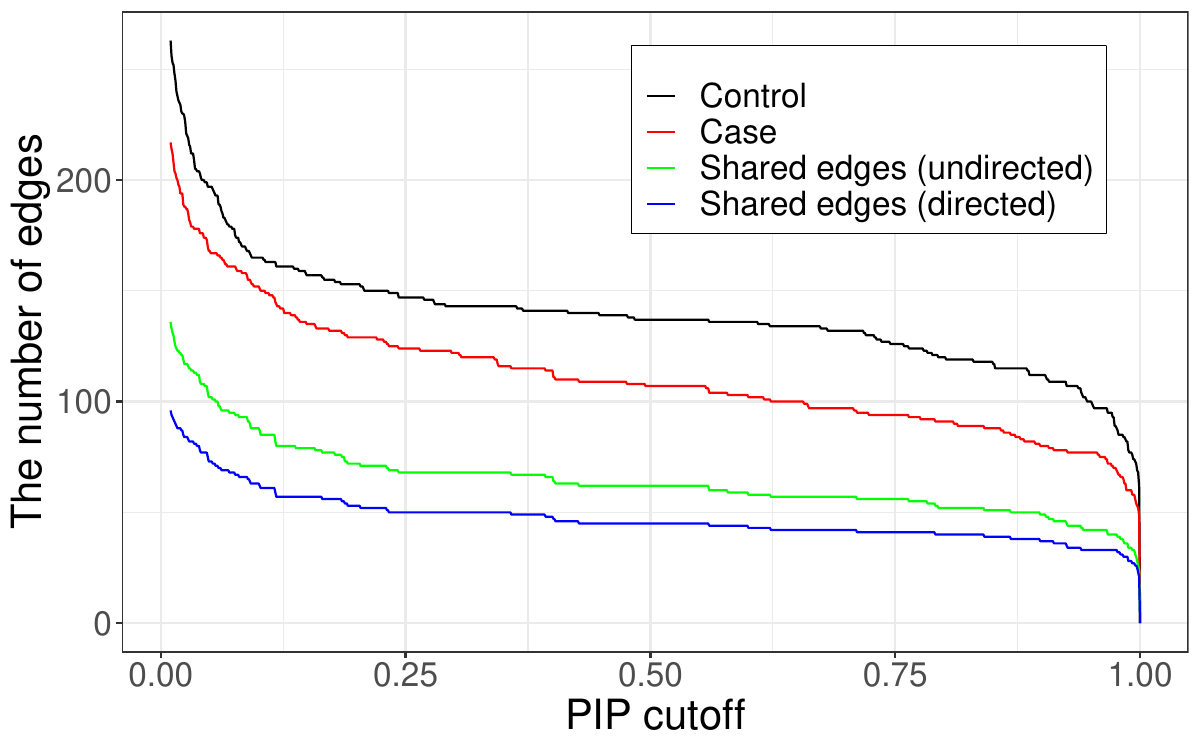}
    \caption{Result of the proposed method for the real data analysis. Given  $\hat{\Gamma}^{\rm{RB}}_{ij}$, we infer the edge $i \rightarrow j$ exists in the DAG if $\hat{\Gamma}^{\rm{RB}}_{ij} > c$ where $c$ is the cutoff of posterior inclusion probability. For each $c$, we count the number of edges occurring in the DAG for control samples (black), the number of edges in the DAG for case samples (red), the number of edges with edge direction ignored in both DAGs (green), and the number of directed edges in both DAGs (blue).}
    \label{fig:real_data}
\end{figure}
 
\section{Single-cell real data analysis}\label{sec:real}
We use a real data set from the single-cell RNA database for Alzheimer's disease, known as scREAD~\citep{jiang2020scread}, to illustrate the advantages of the proposed algorithm. We only consider genes involved in the brain-derived neurotrophic factor signaling pathway and expressed in the layer 2--3 glutamatergic neurons. 
The goal is to learn two DAG models, one from case samples and the other from control samples, and then inspect how different the two DAGs are.  
To mitigate potential batch effects, we only use samples that are generated at similar sequencing depths by checking the total and median expression level across all genes for each sample cell, which results in $n_0 = 2300$ control samples and $n_1 = 1666$ case samples. 
Next, we select the genes in this pathway expressed in at least half of the samples in both data sets, which yields $p = 73$. The data matrices for both case and control samples are obtained by performing normalization of log-transformed expression levels~\citep[Chapter 6]{lee2007analysis}.  

For each of the two data sets, we run the proposed order MCMC sampler with iterative top-down initialization for $2\times 10^5$ MCMC iterations, and then discard the first $10^5$ iterations as burn-in. It only takes about 480 seconds for each data set. 
To infer the edge posterior inclusion probabilities, we use the conditioning scheme described in Section~\ref{subsec:RB}, and the result is presented in Fig.~\ref{fig:real_data}. 
The two DAGs learned from the data share a significant proportion of  undirected edges, and more importantly, most of these edges have the same direction in both data sets: the gap between the blue and green lines in Fig.~\ref{fig:real_data} is narrow. In other words, the orderings of the variables learned from the two data sets are very similar.  
The true ordering of the variables is hard to determine as there may even exist feedback loops among the selected genes, and we do not know to what extent the true model satisfies the equal variance assumption.   But Fig.~\ref{fig:real_data}  suggests that the use of this score is very reasonable from a pragmatic perspective. 
For comparison, we have also tried the minimal I-MAP MCMC with the decomposable score given in~\eqref{eq:equivalent_score}, which represents a state-of-the-art score equivalent Bayesian structure learning procedure, and the result is shown in Section~\ref{subsec:realdata_supp} of the supplementary material. Given the same initialization and same number of MCMC and burn-in iterations, our method yields a higher proportion of shared directed edges than the minimal I-MAP MCMC. For example, with the posterior inclusion probability cutoff being 0.5, for our method 41\% of the edges in the inferred DAG for case samples also occur in the same direction in the DAG for control samples, while this ratio drops to 26\% for the minimal I-MAP MCMC.  

To provide further evidence for the advantage of the proposed structure learning method, we repeat the above analysis 30 times using both our sampler with non-decomposable score and the minimal I-MAP MCMC with decomposable score. Then, for each pair $(i, j)$ with $i \neq j$, we calculate the Gelman-Rubin scale factor~\citep{gelman1992inference} using $\Gamma_{ij}$, which is equal to $1$ if $i \rightarrow j$ is in the sampled DAG and $0$ otherwise. Thus, we get $p(p - 1)$ Gelman-Rubin statistics for each data set, one for each directed edge.  We find that 99.7\% of the directed edges in the two DAGs have Gelman-Rubin statistics lower than $1.1$ for our method, and 93.7\%  for the minimal I-MAP MCMC; we use the threshold $1.1$ since this is the most common choice according to~\citet{vats2021revisiting}. Moreover, for the minimal I-MAP MCMC,  
Gelman-Rubin statistics of 90 directed edges yield infinity, which means that the within-chain variance of $\Gamma_{ij}$ is zero for all 30 runs, but the between-chain variance is nonzero; that is, in some runs the edge $i \rightarrow j$ is selected in every iteration excluding burn-in, while in the other runs the edge $i \rightarrow j$ is never selected.  
This observation again illustrates that for a score equivalent procedure, traversing equivalence classes can sometimes be very difficult and cause slow mixing of MCMC samplers.  
In contrast,  the maximum Gelman-Rubin statistic for our method is  $2.56$ for the control data set and $1.26$ for the case data set.

\section*{Acknowledgement}
The authors would like to thank the anonymous reviewers for their comments which helped improve the paper, and thank Prof. Mohsen Pourahmadi and Yongjian Yang for helpful discussions. 
HC and QZ were supported in part by NSF grant DMS-2245591.   
All authors were supported by the Triads for Transformation Grant of Texas A\&M University.

\bibliographystyle{plainnat}
\bibliography{paper-ref}

\begin{thebibliography}{52}
\providecommand{\natexlab}[1]{#1}
\providecommand{\url}[1]{\texttt{#1}}
\expandafter\ifx\csname urlstyle\endcsname\relax
  \providecommand{\doi}[1]{doi: #1}\else
  \providecommand{\doi}{doi: \begingroup \urlstyle{rm}\Url}\fi

\bibitem[Agrawal et~al.(2018)Agrawal, Uhler, and Broderick]{agrawal2018minimal}
Raj Agrawal, Caroline Uhler, and Tamara Broderick.
\newblock Minimal {I-MAP} {MCMC} for scalable structure discovery in causal
  {DAG} models.
\newblock In \emph{International Conference on Machine Learning}, pages 89--98,
  2018.

\bibitem[An et~al.(2008)An, Huang, Yao, and Zhang]{an2008stepwise}
Hongzhi An, Da~Huang, Qiwei Yao, and Cun-Hui Zhang.
\newblock Stepwise searching for feature variables in high-dimensional linear
  regression.
\newblock \emph{Technical report}, 2008.

\bibitem[Andersson et~al.(1997)Andersson, Madigan, and
  Perlman]{andersson1997characterization}
Steen~A Andersson, David Madigan, and Michael~D Perlman.
\newblock A characterization of {M}arkov equivalence classes for acyclic
  digraphs.
\newblock \emph{The Annals of Statistics}, 25\penalty0 (2):\penalty0 505--541,
  1997.

\bibitem[Aragam et~al.(2019)Aragam, Amini, and Zhou]{aragam2019globally}
Bryon Aragam, Arash Amini, and Qing Zhou.
\newblock Globally optimal score-based learning of directed acyclic graphs in
  high-dimensions.
\newblock \emph{Advances in Neural Information Processing Systems},
  32:\penalty0 4450--4462, 2019.

\bibitem[Ben-David et~al.(2011)Ben-David, Li, Massam, and
  Rajaratnam]{ben2011high}
Emanuel Ben-David, Tianxi Li, H{\'e}lene Massam, and Bala Rajaratnam.
\newblock High dimensional {B}ayesian inference for {G}aussian directed acyclic
  graph models.
\newblock \emph{arXiv preprint arXiv:1109.4371}, 2011.

\bibitem[Bernstein and Nestoridi(2019)]{bernstein2019cutoff}
Megan Bernstein and Evita Nestoridi.
\newblock Cutoff for random to random card shuffle.
\newblock \emph{The Annals of Probability}, 47\penalty0 (5):\penalty0
  3303--3320, 2019.

\bibitem[Cao et~al.(2019)Cao, Khare, and Ghosh]{cao2019posterior}
Xuan Cao, Kshitij Khare, and Malay Ghosh.
\newblock Posterior graph selection and estimation consistency for
  high-dimensional {B}ayesian {DAG} models.
\newblock \emph{The Annals of Statistics}, 47\penalty0 (1):\penalty0 319--348,
  2019.

\bibitem[Carvalho and Scott(2009)]{carvalho2009objective}
Carlos~M Carvalho and James~G Scott.
\newblock Objective {B}ayesian model selection in {G}aussian graphical models.
\newblock \emph{Biometrika}, 96\penalty0 (3):\penalty0 497--512, 2009.

\bibitem[Castelletti and Consonni(2021)]{castelletti2021bayesian}
Federico Castelletti and Guido Consonni.
\newblock Bayesian inference of causal effects from observational data in
  {G}aussian graphical models.
\newblock \emph{Biometrics}, 77\penalty0 (1):\penalty0 136--149, 2021.

\bibitem[Castelletti et~al.(2018)Castelletti, Consonni, Della~Vedova, and
  Peluso]{castelletti2018learning}
Federico Castelletti, Guido Consonni, Marco~L Della~Vedova, and Stefano Peluso.
\newblock Learning {M}arkov equivalence classes of directed acyclic graphs: an
  objective {B}ayes approach.
\newblock \emph{Bayesian Analysis}, 13\penalty0 (4):\penalty0 1235--1260, 2018.

\bibitem[Chang et~al.(2022)Chang, Lee, Luo, Sang, and Zhou]{chang2022rapidly}
Hyunwoong Chang, Changwoo Lee, Zhao~Tang Luo, Huiyan Sang, and Quan Zhou.
\newblock Rapidly mixing multiple-try {M}etropolis algorithms for model
  selection problems.
\newblock \emph{Advances in Neural Information Processing Systems},
  35:\penalty0 25842--25855, 2022.

\bibitem[Chen et~al.(2019)Chen, Drton, and Wang]{chen2019causal}
Wenyu Chen, Mathias Drton, and Y~Samuel Wang.
\newblock On causal discovery with an equal-variance assumption.
\newblock \emph{Biometrika}, 106\penalty0 (4):\penalty0 973--980, 2019.

\bibitem[Chickering(2002)]{chickering2002learning}
David~Maxwell Chickering.
\newblock Learning equivalence classes of {Bayesian}-network structures.
\newblock \emph{Journal of machine learning research}, 2\penalty0
  (Feb):\penalty0 445--498, 2002.

\bibitem[Drton and Maathuis(2017)]{drton2017structure}
Mathias Drton and Marloes~H Maathuis.
\newblock Structure learning in graphical modeling.
\newblock \emph{Annual Review of Statistics and Its Application}, 4:\penalty0
  365--393, 2017.

\bibitem[Dwivedi et~al.(2018)Dwivedi, Chen, Wainwright, and Yu]{dwivedi2018log}
Raaz Dwivedi, Yuansi Chen, Martin~J Wainwright, and Bin Yu.
\newblock Log-concave sampling: {M}etropolis-{H}astings algorithms are fast!
\newblock In \emph{Conference on learning theory}, pages 793--797. PMLR, 2018.

\bibitem[Friedman and Koller(2003)]{friedman2003being}
Nir Friedman and Daphne Koller.
\newblock Being bayesian about network structure. {A} {B}ayesian approach to
  structure discovery in {B}ayesian networks.
\newblock \emph{Machine learning}, 50\penalty0 (1):\penalty0 95--125, 2003.

\bibitem[Geiger and Heckerman(2002)]{geiger2002parameter}
Dan Geiger and David Heckerman.
\newblock Parameter priors for directed acyclic graphical models and the
  characterization of several probability distributions.
\newblock \emph{The Annals of Statistics}, 30\penalty0 (5):\penalty0
  1412--1440, 2002.

\bibitem[Gelman and Rubin(1992)]{gelman1992inference}
Andrew Gelman and Donald~B Rubin.
\newblock Inference from iterative simulation using multiple sequences.
\newblock \emph{Statistical Science}, 7\penalty0 (4):\penalty0 457--472, 1992.

\bibitem[Ghoshal and Honorio(2018)]{ghoshal2018learning}
Asish Ghoshal and Jean Honorio.
\newblock Learning linear structural equation models in polynomial time and
  sample complexity.
\newblock In \emph{International Conference on Artificial Intelligence and
  Statistics}, pages 1466--1475. PMLR, 2018.

\bibitem[Glymour et~al.(2019)Glymour, Zhang, and Spirtes]{glymour2019review}
Clark Glymour, Kun Zhang, and Peter Spirtes.
\newblock Review of causal discovery methods based on graphical models.
\newblock \emph{Frontiers in genetics}, 10:\penalty0 524, 2019.

\bibitem[Grzegorczyk and Husmeier(2008)]{grzegorczyk2008improving}
Marco Grzegorczyk and Dirk Husmeier.
\newblock Improving the structure {MCMC} sampler for {B}ayesian networks by
  introducing a new edge reversal move.
\newblock \emph{Machine Learning}, 71\penalty0 (2-3):\penalty0 265, 2008.

\bibitem[Hoyer et~al.(2008)Hoyer, Janzing, Mooij, Peters, and
  Sch{\"o}lkopf]{hoyer2008nonlinear}
Patrik~O Hoyer, Dominik Janzing, Joris~M Mooij, Jonas Peters, and Bernhard
  Sch{\"o}lkopf.
\newblock Nonlinear causal discovery with additive noise models.
\newblock In \emph{NIPS}, volume~21, pages 689--696. Citeseer, 2008.

\bibitem[Jiang et~al.(2020)Jiang, Wang, Qi, Fu, and Ma]{jiang2020scread}
Jing Jiang, Cankun Wang, Ren Qi, Hongjun Fu, and Qin Ma.
\newblock scread: {A} single-cell {RNA}-{S}eq database for {A}lzheimer's
  disease.
\newblock \emph{Iscience}, 23\penalty0 (11):\penalty0 101769, 2020.

\bibitem[Koller and Friedman(2009)]{koller2009probabilistic}
Daphne Koller and Nir Friedman.
\newblock \emph{Probabilistic graphical models: principles and techniques}.
\newblock MIT press, 2009.

\bibitem[Kuipers et~al.(2022)Kuipers, Suter, and Moffa]{kuipers2022efficient}
Jack Kuipers, Polina Suter, and Giusi Moffa.
\newblock Efficient sampling and structure learning of {B}ayesian networks.
\newblock \emph{Journal of Computational and Graphical Statistics}, 31\penalty0
  (3):\penalty0 639--650, 2022.

\bibitem[Laurent and Massart(2000)]{laurent2000adaptive}
Beatrice Laurent and Pascal Massart.
\newblock Adaptive estimation of a quadratic functional by model selection.
\newblock \emph{Annals of Statistics}, pages 1302--1338, 2000.

\bibitem[Lauritzen(1992)]{lauritzen1992propagation}
Steffen~L Lauritzen.
\newblock Propagation of probabilities, means, and variances in mixed graphical
  association models.
\newblock \emph{Journal of the American Statistical Association}, 87\penalty0
  (420):\penalty0 1098--1108, 1992.

\bibitem[Lee et~al.(2019)Lee, Lee, and Lin]{lee2019minimax}
Kyoungjae Lee, Jaeyong Lee, and Lizhen Lin.
\newblock Minimax posterior convergence rates and model selection consistency
  in high-dimensional {DAG} models based on sparse {C}holesky factors.
\newblock \emph{The Annals of Statistics}, 47\penalty0 (6):\penalty0
  3413--3437, 2019.

\bibitem[Lee(2007)]{lee2007analysis}
Mei-Ling~Ting Lee.
\newblock \emph{Analysis of microarray gene expression data}.
\newblock Springer Science \& Business Media, 2007.

\bibitem[Levin and Peres(2017)]{levin2017markov}
David~A Levin and Yuval Peres.
\newblock \emph{Markov chains and mixing times}, volume 107.
\newblock American Mathematical Soc., 2017.

\bibitem[Madigan et~al.(1995)Madigan, York, and Allard]{madigan1995bayesian}
David Madigan, Jeremy York, and Denis Allard.
\newblock Bayesian graphical models for discrete data.
\newblock \emph{International Statistical Review/Revue Internationale de
  Statistique}, pages 215--232, 1995.

\bibitem[Mangoubi and Smith(2017)]{mangoubi2017rapid}
Oren Mangoubi and Aaron Smith.
\newblock Rapid mixing of {H}amiltonian {M}onte {C}arlo on strongly log-concave
  distributions.
\newblock \emph{arXiv preprint arXiv:1708.07114}, 2017.

\bibitem[Martin et~al.(2017)Martin, Mess, and Walker]{martin2017empirical}
Ryan Martin, Raymond Mess, and Stephen~G Walker.
\newblock Empirical {B}ayes posterior concentration in sparse high-dimensional
  linear models.
\newblock \emph{Bernoulli}, 23\penalty0 (3):\penalty0 1822--1847, 2017.

\bibitem[Park(2020)]{park2020identifiability}
Gunwoong Park.
\newblock Identifiability of additive noise models using conditional variances.
\newblock \emph{The Journal of Machine Learning Research}, 21\penalty0
  (1):\penalty0 2896--2929, 2020.

\bibitem[Peters et~al.(2011)Peters, Mooij, Janzing, and
  Sch{\"o}lkopf]{peters2011identifiability}
J~Peters, J~Mooij, D~Janzing, and B~Sch{\"o}lkopf.
\newblock Identifiability of causal graphs using functional models.
\newblock In \emph{27th Conference on Uncertainty in Artificial Intelligence
  (UAI 2011)}, pages 589--598. AUAI Press, 2011.

\bibitem[Peters and B{\"u}hlmann(2014)]{peters2014identifiability}
Jonas Peters and Peter B{\"u}hlmann.
\newblock Identifiability of {G}aussian structural equation models with equal
  error variances.
\newblock \emph{Biometrika}, 101\penalty0 (1):\penalty0 219--228, 2014.

\bibitem[Ravikumar et~al.(2011)Ravikumar, Wainwright, Raskutti, and
  Yu]{ravikumar2011high}
Pradeep Ravikumar, Martin~J Wainwright, Garvesh Raskutti, and Bin Yu.
\newblock High-dimensional covariance estimation by minimizing
  $\ell^1$-penalized log-determinant divergence.
\newblock \emph{Electronic Journal of Statistics}, 5:\penalty0 935--980, 2011.

\bibitem[Robert and Roberts(2021)]{robert2021rao}
Christian~P Robert and Gareth~O Roberts.
\newblock {R}ao-{B}lackwellization in the {MCMC} era.
\newblock \emph{arXiv preprint arXiv:2101.01011}, 2021.

\bibitem[Shimizu et~al.(2006)Shimizu, Hoyer, Hyv{\"a}rinen, Kerminen, and
  Jordan]{shimizu2006linear}
Shohei Shimizu, Patrik~O Hoyer, Aapo Hyv{\"a}rinen, Antti Kerminen, and Michael
  Jordan.
\newblock A linear non-{G}aussian acyclic model for causal discovery.
\newblock \emph{Journal of Machine Learning Research}, 7\penalty0 (10), 2006.

\bibitem[Shojaie and Michailidis(2010)]{shojaie2010penalized}
Ali Shojaie and George Michailidis.
\newblock Penalized likelihood methods for estimation of sparse
  high-dimensional directed acyclic graphs.
\newblock \emph{Biometrika}, 97\penalty0 (3):\penalty0 519--538, 2010.

\bibitem[Sinclair(1992)]{sinclair1992improved}
Alistair Sinclair.
\newblock Improved bounds for mixing rates of {M}arkov chains and
  multicommodity flow.
\newblock \emph{Combinatorics, probability and Computing}, 1\penalty0
  (4):\penalty0 351--370, 1992.

\bibitem[Strieder et~al.(2021)Strieder, Freidling, Haffner, and
  Drton]{strieder2021confidence}
David Strieder, Tobias Freidling, Stefan Haffner, and Mathias Drton.
\newblock Confidence in causal discovery with linear causal models.
\newblock In \emph{Uncertainty in Artificial Intelligence}, pages 1217--1226.
  PMLR, 2021.

\bibitem[Su and Borsuk(2016)]{su2016improving}
Chengwei Su and Mark~E Borsuk.
\newblock Improving structure {MCMC} for {B}ayesian networks through {M}arkov
  blanket resampling.
\newblock \emph{The Journal of Machine Learning Research}, 17\penalty0
  (1):\penalty0 4042--4061, 2016.

\bibitem[Sullivant et~al.(2010)Sullivant, Talaska, and
  Draisma]{sullivant2010trek}
Seth Sullivant, Kelli Talaska, and Jan Draisma.
\newblock Trek separation for {G}aussian graphical models.
\newblock \emph{The Annals of Statistics}, 38\penalty0 (3):\penalty0
  1665--1685, 2010.

\bibitem[Tadesse and Vannucci(2021)]{tadesse2021handbook}
Mahlet~G Tadesse and Marina Vannucci.
\newblock Handbook of {B}ayesian variable selection.
\newblock 2021.

\bibitem[Uhler et~al.(2013)Uhler, Raskutti, B{\"u}hlmann, and
  Yu]{uhler2013geometry}
Caroline Uhler, Garvesh Raskutti, Peter B{\"u}hlmann, and Bin Yu.
\newblock Geometry of the faithfulness assumption in causal inference.
\newblock \emph{The Annals of Statistics}, pages 436--463, 2013.

\bibitem[Van~de Geer and B{\"u}hlmann(2013)]{van2013ell}
Sara Van~de Geer and Peter B{\"u}hlmann.
\newblock $\ell^0$-penalized maximum likelihood for sparse directed acyclic
  graphs.
\newblock \emph{The Annals of Statistics}, 41\penalty0 (2):\penalty0 536--567,
  2013.

\bibitem[Vats and Knudson(2021)]{vats2021revisiting}
Dootika Vats and Christina Knudson.
\newblock Revisiting the {G}elman--{R}ubin diagnostic.
\newblock \emph{Statistical Science}, 36\penalty0 (4):\penalty0 518--529, 2021.

\bibitem[Yang et~al.(2016)Yang, Wainwright, and Jordan]{yang2016computational}
Yun Yang, Martin~J Wainwright, and Michael~I Jordan.
\newblock On the computational complexity of high-dimensional {B}ayesian
  variable selection.
\newblock \emph{The Annals of Statistics}, 44\penalty0 (6):\penalty0
  2497--2532, 2016.

\bibitem[Yu and Bien(2017)]{yu2017learning}
Guo Yu and Jacob Bien.
\newblock Learning local dependence in ordered data.
\newblock \emph{The Journal of Machine Learning Research}, 18\penalty0
  (1):\penalty0 1354--1413, 2017.

\bibitem[Zhou and Chang(2021)]{zhou2021complexity}
Quan Zhou and Hyunwoong Chang.
\newblock Complexity analysis of {B}ayesian learning of high-dimensional {DAG}
  models and their equivalence classes.
\newblock \emph{arXiv preprint arXiv:2101.04084}, 2021.

\bibitem[Zhou(2010)]{zhou2010thresholded}
Shuheng Zhou.
\newblock Thresholded {L}asso for high dimensional variable selection and
  statistical estimation.
\newblock \emph{arXiv preprint arXiv:1002.1583}, 2010.

\end{thebibliography}

\setcounter{section}{0}
\renewcommand\thesection{\Alph{section}}
\renewcommand\thesubsection{\Alph{section}.\arabic{subsection}}

\setcounter{equation}{16}
\setcounter{algocf}{2}

\newpage
\begin{center}
\LARGE{Supplementary material}
\end{center}

\section{Algorithms}\label{subsec:alg}

\subsection{Overview of the proposed method}
We outline the proposed order MCMC algorithm in Algorithm~\ref{alg:full}. 
For all displayed algorithms, we assume the data matrix $X$ and model parameters $(c_0, \gamma, \alpha, \kappa, \din)$ are given.
The \texttt{R} code for the proposed method and simulation studies can be found at \url{https://github.com/hwchang1201/bayes.eqvar}.

\begin{algorithm}[h!]
\caption{Bayesian order-based structure learning} \label{alg:full}
\KwInput{Number of MCMC iterations $T$, neighborhood function $\cN = \cN_{\mathrm{adj}}, \cN_{\rm{rtr}}$ or $\cN_{\mathrm{rrs}}$, a DAG selection procedure $\hat{G} \colon \bbS^{p} \rightarrow \cG_p$ (e.g. Algorithm~\ref{alg:dagwise})}
$\sigma^{(0)} \leftarrow\, \hat{\sigma}^{\mathrm{ITD}}$  
\tcp{$\hat{\sigma}^{\mathrm{ITD}}$ is the output of Algorithm~\ref{alg:ITD}}
$G^{(0)} \leftarrow \, \hat{G} ( \sigma^{(0)})$  \\ 
\For{$t = 1, \dots, T$}
{ 
Draw $\sigma$ uniformly from $\cN(\sigma^{(t-1)}))$\\
Draw $u \sim \mathrm{Uniform}(0,1)$ \\ $a \leftarrow\, \min(\post(\sigma)/\post(\sigma^{(t-1)}),1)$   \\
\If{$u \leq a $}{
$ \sigma^{(t)} \leftarrow\, \sigma $ \\ 
$ G^{(t)} \leftarrow\, \hat{G}( \sigma )$ 
}\Else{
$ \sigma^{(t)} \leftarrow \, \sigma^{(t - 1)}$ \\
$ G^{(t)} \leftarrow\, G^{(t - 1)} $
}
$\hat{\Gamma}^{(t)} = \hat{\Gamma}( \sigma^{(t)}, G^{(t)})$ \tcp{See~\eqref{eq:RB} for the definition of $\hat{\Gamma}$}
}
\KwOutput{``Rao-Blackwellized'' adjacency matrices $\{ \hat{\Gamma}^{(t)} \}_{t=1}^T$ }
\end{algorithm}

\subsection{Forward-backward algorithms  with non-decomposable scores}\label{subsec:stepwise}
Recall the posterior score of a DAG  given in~\eqref{eq:post_score}. Define the nodewise score at node $j$ by 
\begin{align}\label{eq:nodewise_step}
    \phi_j(S, \SSE_{\minus j}) = -|S|  \log \left\{ p^{c_0}\sqrt{ (1 + \alpha/\gamma)} \right\}  -  \frac{\alpha p n + \kappa}{2} \log  \left(  \SSE_{\minus j}  +  X_j^\T \oproj_S X_j \right), 
\end{align}
for $S \subseteq P_j$, where $\SSE_{\minus j}$ denotes the total residual sum of squares of nodes other than $j$, and $P_j$ is the potential parent set defined in \eqref{eq:potential}. 
Hence, given $\SSE_{\minus j}$, we can use the standard forward-backward stepwise algorithm to select the parent set of node $j$; this is described in Algorithm~\ref{alg:nodewise_step}. 
We allow using two different estimates for $\SSE_{\minus j}$, one for the forward phase and the other for the backward phase; the reason will become clear in the next subsection.  

\begin{algorithm}[t!]
\caption{Nodewise forward-backward selection} \label{alg:nodewise_step}
\DontPrintSemicolon
\KwInput{Node index $j \in [p]$, a set of potential parent nodes $P_j \subset [p]$, two estimates for the total residual of sum of squares of other nodes $\SSE_{\minus j}$, $\SSE'_{\minus j}$}
\textbf{Forward phase:} $S_\mathrm{f} \leftarrow\, \emptyset $ \\ 
\For{$k = 1, \dots, |P_j|$}{
$\ell_0 \leftarrow\, \argmax_{ \ell \in P_j \setminus S_\mathrm{f} } \phi_j( S_\mathrm{f} \cup \{\ell \},  \SSE_{\minus j} )$  \\
$\tilde{S}_\mathrm{f} \leftarrow\, S_\mathrm{f} \cup \{ \ell_0 \}$

\If{ $\phi_j ( \tilde{S}_\mathrm{f}, \SSE_{\minus j }) \geq \phi_j (S_\mathrm{f}, \SSE_{\minus j} )$}{
$S_\mathrm{f} \leftarrow\, \tilde{S}_\mathrm{f}$
}\Else{\textbf{break}}
}
\textbf{Backward phase:} $S_\mathrm{b} \leftarrow\, S_\mathrm{f} $ \\ 
\For{$k = 1, \dots, |S_\mathrm{f}|$}{
$\ell_1 \leftarrow\, \argmax_{ \ell \in S_\mathrm{b} } \phi_j( S_\mathrm{b} \setminus \{\ell \},  \SSE'_{\minus j} )$  \\
$\tilde{S}_\mathrm{b} \leftarrow\, S_\mathrm{b} \setminus \{ \ell_1 \}$\\
\If{ $\phi_j ( \tilde{S}_\mathrm{b}, \SSE'_{\minus j }) \geq \phi_j (S_\mathrm{b}, \SSE'_{\minus j} )$}{
$S_\mathrm{b} \leftarrow\, \tilde{S}_\mathrm{b}$
}\Else{\textbf{break}}
}
\KwOutput{A parent set $S_\mathrm{b}$ of node $j$}
\end{algorithm}

\begin{algorithm}[t!]
\caption{Forward-backward DAG selection} 
\label{alg:dagwise}
\KwInput{$\sigma \in \bbS^p$} 
$G \leftarrow\, \text{empty DAG} $\\
\tcp{Forward phase}
\While{$1$}{
$(i_0, j_0) \leftarrow\, \argmax_{i, j \colon \sigma^{-1}(i) < \sigma^{-1}(j), \{i \rightarrow j \}  \notin G} \phi ( G \cup \{ i \rightarrow j\}  )$ \\
$\tilde{G} \leftarrow\, G \cup \{ i_0 \rightarrow j_0 \}$\\
\If{$\phi(\tilde{G}) \geq \phi(G)$}{
$G \leftarrow\, \tilde{G}$
}\Else{\textbf{break}}
}
\tcp{Backward phase}
\While{$1$}{
$(i_1, j_1) \leftarrow\, \argmax_{i, j \colon \{i \rightarrow j \}  \in G} \phi ( G \setminus \{ i \rightarrow j\}  )$ \\
$\tilde{G} \leftarrow\, G \setminus \{ i_1 \rightarrow j_1 \}$\\
\If{$\phi(\tilde{G}) \geq \phi(G)$}{
$G \leftarrow\, \tilde{G}$
}\Else{\textbf{break}}
}
\KwOutput{DAG $G$}
\end{algorithm}

\subsection{Implementation of order MCMC with non-decomposable scores}\label{subsubsec:effective} 
For our model, the main computational challenge is that a local change to the ordering $\sigma$ can cause some global changes to the maximum a posteriori DAG estimator $\hat{G}_\sigma^{\MAP}$, due to the use of the non-decomposable posterior score.  
Were the posterior score decomposable,  whenever we use an adjacent transposition to move from $\sigma$ to $\sigma' =  \sigma \circ (i, i+1)_{\mathrm{c}}$, we know that $\Pa_j(\hat{G}_\sigma^{\MAP}) = \Pa_j(\hat{G}_{\sigma'}^{\MAP})$ for any $j \notin \{ \sigma(i), \sigma(i + 1) \}$,  since maximizing the score of the entire DAG is equivalent to maximizing the local score at each node separately. 

We describe a strategy for implementing local moves on $\bbS^p$ for our model, which is as efficient as with a decomposable posterior score. 
We start by proving two monotone properties of the nodewise score defined in~\eqref{eq:nodewise_step}.   
\begin{lemma}\label{lm:nodewise.range}
Let $\phi_j$ be as given in~\eqref{eq:nodewise_step}, $S \subset [p] \setminus \{j\}$, $k \notin S \cup \{j\}$ and $a > 0$. 
\begin{enumerate}[(i)]
    \item If $\phi_j(S \cup \{k\}, a) > \phi_j(S, a)$, then $\phi_j(S \cup \{k\}, b) > \phi_j(S, b)$ for any $0 < b < a$. 
    \item If $\phi_j(S \cup \{k\}, a) < \phi_j(S, a)$, then $\phi_j(S \cup \{k\}, b) <  \phi_j(S, b)$ for any $b > a$. 
\end{enumerate} 
\end{lemma}
\begin{proof}
To simplify the notation, let $K_0 = \log \{ p^{c_0}\sqrt{ (1 + \alpha/\gamma)} \}$ and $K_1 = (\alpha p n + \kappa) / 2$. A routine calculation shows that $\phi_j(S \cup \{k\}, a) > \phi_j(S, a)$ if and only if 
\begin{align*}
    \log \frac{a +  X_j^\T \oproj_{S} X_j }{a +   X_j^\T \oproj_{S  \cup \{k\}  } X_j} > \frac{K_0}{K_1}. 
\end{align*}
The claim follows by observing that the left-hand side is monotonically decreasing in $a$. 
\end{proof}

Motivated by Lemma~\ref{lm:nodewise.range}, we use the following procedure to find $\hat{G}_\sigma^{\MAP}$ for a given $\sigma \in \bbS^p$. First, for $j = 1, \dots, p$, we find a lower bound and an upper bound on $\SSE_j$ such that both bounds do not depend on $\sigma$. An obvious choice for the upper bound on $\SSE_j$ is given by  $\umax_{j} = X_j^\T X_j$, and if $p < n$, a lower bound is given by $\umin_{j} = X_j^\T \oproj_{[p] \setminus \{j\}} X_j$ (we assume $\umin_{j}$ is strictly positive). 
Next, for $j=1, \dots, p$, we apply Algorithm~\ref{alg:nodewise_step} with input $(j, P_j, \sum_{k \neq j} \umin_{k},  \sum_{k \neq j} \umax_{k})$; that is, in the forward stage, we let the algorithm  select as many parent nodes as possible by using minimum estimates for the residual sum of squares of other nodes, and in the backward stage, we let the algorithm remove as many nodes as possible. 
For all nodes, save the search paths of Algorithm~\ref{alg:nodewise_step}, including  the changes in residual sum of squares in each step, in the internal memory, and let $\overline{S}_j^\sigma$ denote the parent set of node $j$ at the end of the forward stage. 
Denote by $\overline{G}_\sigma$ the DAG such that $\Pa_j(\overline{G}_\sigma) = \overline{S}_j^\sigma$ for each $j$. Now to find $\hat{G}_\sigma^{\MAP}$, we simply apply the backward stage of Algorithm~\ref{alg:dagwise} by initializing the DAG to $\overline{G}_\sigma$. This  can be done very efficiently by using the search paths of  Algorithm~\ref{alg:nodewise_step}; no calculation of residual sum of squares is needed.  

The above procedure enables an efficient updating algorithm for finding $\hat{G}_\sigma^{\MAP}$ when we move locally on the ordering space $\bbS^p$.  For example, consider moving from $\sigma$ to $\sigma' =   \sigma \circ (i, i+1)_{\mathrm{c}}$. We only need to apply Algorithm~\ref{alg:nodewise_step} at nodes $\sigma(i)$ and $\sigma(i+1)$, and then perform backward DAG selection using the saved search paths of nodewise forward-backward selection. The computational time of the DAG selection step is negligible compared to that of Algorithm~\ref{alg:nodewise_step}. 
Note that the parent sets of nodes other than $\sigma(i)$ and $\sigma(i+1)$ may change.  

\subsection{Three random walk proposals}\label{subsubsec:proposal} 

Figure~\ref{fig:proposal} describes  (1) adjacent transposition, (2) random transposition, and (3) random-to-random shuffle, given the current topological ordering $\sigma$. The random transposition $\sigma \circ (i, j)_{\mathrm{c}}$ interchanges the $i$-th and the $j$-th elements of $\sigma$ while keeping the others unchanged.  The adjacent transposition is a special case of random transposition where $i$ and $j$ are adjacent, i.e., $|i-j| = 1$. The random-to-random shuffle $\sigma \circ \xi(i, j)$ inserts the $i$-th element of $\sigma$ to the $j$-th position.

\begin{figure}[h!]
    \centering
    \includegraphics[width=0.96\textwidth]{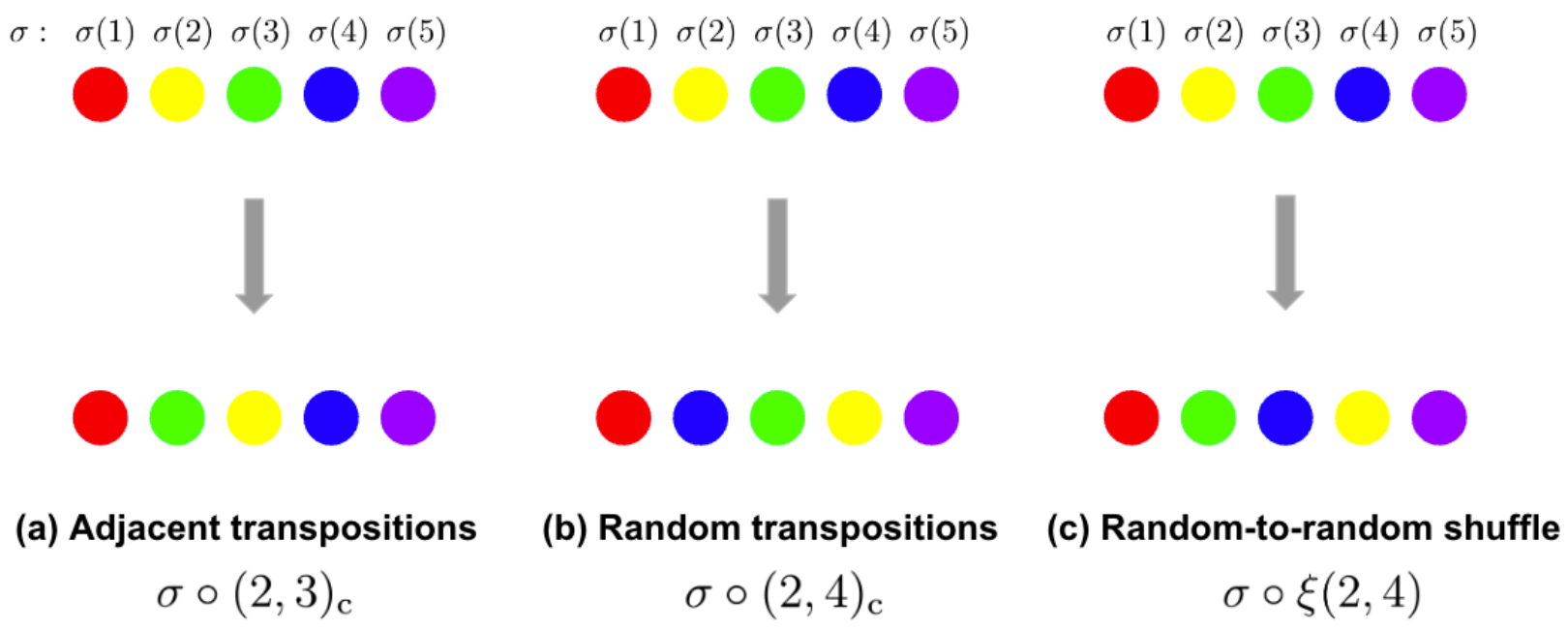}
    \caption{Illustration of the three proposals introduced in Section~\ref{subsec:MH}: adjacent transposition, the random transposition and the random-to-random shuffle.}
    \label{fig:proposal}
\end{figure}

\section{Proofs} \label{subsec:proof}

\subsection{High-probability events}\label{subsec:event}

Recall that we assume the data is generated according to the linear structural equation model (SEM) given in~\eqref{eq:true.sem}. 
Since the rows of $X$ are assumed to be i.i.d. copies of $\sX$, we have 
\begin{align}\label{st.eq:true}
    X_j = \sum_{i = 1}^p (B^*)_{ij} X_i + \epsilon_j, \text{ where }\epsilon_j \sim N_n(0, \omega^*_{j} I), \text{ for all } j \in [p].
\end{align} 
By Remark~\ref{remark:minimap}, for each  $\sigma \in \bbS^p$, we can derive a linear SEM equivalent to~\eqref{st.eq:true}, which is given by 
\begin{align}\label{st.eq:minimap}
    X_j = \sum_{i = 1}^p (B_{\sigma}^*)_{ij} X_i + \epsilon_j^{\sigma}, \text{ where }\epsilon_j^\sigma \sim N_n(0, \omega^{\sigma}_j I), \text{ for all } j \in [p].
\end{align}   
We define the normalized error vectors by
\begin{gather*}
    z_j = (\omega^*_{j})^{-\frac{1}{2}} \epsilon_j \text{ for } j \in [p], \qquad \qquad 
    z_j^\sigma = (\omega_{j}^{\sigma})^{-\frac{1}{2}} \epsilon_j^\sigma \text{ for }  \sigma \in \bbS^p, j \in [p],
\end{gather*}
where $z_j$ and $z_j^\sigma$ are associated with the true model given in~\eqref{st.eq:true} and the linear SEM in~\eqref{st.eq:minimap}, respectively. The sets of the corresponding normalized errors are defined by $\cZ_0 = \{ z_j \colon  j \in [p]\}$ and $\cZ_1 = \{ z_j^\sigma \colon \sigma \in \bbS^p, j \in [p]\}$. Clearly, $\cZ_0 \subseteq \cZ_1$ and $|\cZ_0| = p$. Further, one can show that 
\begin{align*}
    |\cZ_1| \leq p \cdot \binom{p}{d^*}, 
\end{align*} 
where $d^*$ is defined in~\eqref{eq:def.dstar}.

Before we prove the results given in the main text, we first define some event sets on which the random components of our generating SEM  behaves as desired, and use concentration inequalities to show that they happen with high probability. We will then prove the main results of the paper by conditioning on these high-probability events. 
Recall $P_j^\sigma$  defined in \eqref{eq:potential} and let $\model(d, P) = \{S \subseteq P \colon  |S| \leq d\}$.  Define 
\begin{gather*}
    \cA  = \Big\{   n \vmin  \leq   \min_{\substack{S \subseteq \model(2\din, [p])}} \lmin( \X_S^\T \X_S ) \leq  \max_{\substack{S \subseteq \model(2\din, [p])}} \lmax( \X_S^\T \X_S ) \leq n \vmax  \Big\},  \\
    \cB   = \left\{\min_{j \in [p]}\min_{ S \subseteq \model(2\din, P_j^{\sigma^*})} (z_j)^\T \oproj_S z_j \geq \frac{1}{2}n  \right\}, \\
    \cB'   = \left\{\min_{j \in [p] , \sigma \in \bbS^p}\min_{ S \subseteq \model(2\din, P_j^\sigma)} (z_j^\sigma)^\T \oproj_S z_j^\sigma \geq \frac{1}{2}n  \right\}, \\
    \cC  = \left\{ \max_{j \in [p]} \max_{  \substack{k \not\in S \\ S \cup \{k\} \subseteq \model(2\din, P_j^{\sigma^*}) }} z_j^\T (\proj_{S\cup\{k\}} - \proj_S )z_j \leq \rho \log p \right\}, 
\end{gather*}
\begin{gather*}
    \cD  = \left\{\min_{j \in [p] , \sigma \in \bbS^p}\min_{ S \subseteq \model(2\din, P_j^\sigma)}  (z_j^\sigma)^\T \oproj_S z_j^\sigma > (1 - \frac{1}{  2\eta})n  \right\},   \\
    \cE  = \left\{\max_{j \in [p]} \max_{  S \subseteq \model(2\din, P_j^{\sigma^*}) } z_j^\T \oproj_S z_j < (1 + \frac{1}{4 \eta})n \right\}, \\
    \cJ = \bigcap_{i, j \in [p]} \left\{ \left|\frac{X_i^\T X_j}{n} - \Sigma_{ij}^* \right| \leq 160 \vmax \sqrt{\frac{\log p}{n}}\right\},
\end{gather*}
where $\eta, \rho > 0$ are universal constants. 

\begin{lemma}\label{L:set_ordering}
Under the conditions of Proposition~\ref{prop:freq}, we have $\bbP^*(\cA \cap \cB \cap \cC) \geq 1 - 4p^{-1}$ for sufficiently large $n$.
\end{lemma}

\begin{proof}
From Lemma F1 of~\citet{zhou2021complexity}, we have $\bbP^*(\cA) \geq 1 - p^{-1}$ for sufficiently large $n$. The proof for the bounds of $\bbP^*(\cB)$ and $\bbP^*(\cC)$ is analogous to that of Lemma F2 of~\citet{zhou2021complexity}. A standard calculation using the tail bounds for chi-squared distributions~\citep{laurent2000adaptive}[Lemma 1] yields 
\begin{gather*}
    \bbP^* \left\{z_j^\T \oproj_S z_j \leq \frac{1}{2}n\right\} \leq e^{-n/48},\\
    \bbP^* \left\{z_j^\T (\proj_{T\cup\{k\}} - \proj_T )z_j \geq \rho \log p\right\} \leq 2e^{-\rho \log p /2},
\end{gather*}
for any $j \in [p]$, $S \subseteq \model(2\din, P_j^{\sigma^*})$ and $T \cup \{k\} \subseteq \model(2\din, P_j^{\sigma^*})$. 
To conclude the proof,  apply union bounds with the observations $|Z_0| \leq p$ and $|\model(2\din, P_j^{\sigma^*}) \}| \leq p^{2\din + 1}$ and  the assumptions $\din \log p = o(n)$ and $\rho > 4\din + 6$.  
\end{proof}

\begin{lemma}\label{L:event_theorem}
Assume $\din \log p = o(n)$ and $d^* \leq \din$. There exists some universal constant $c' = c'(\eta) > 0$ such that $\bbP^*(\cD \cap \cE) \geq 1 - 2e^{-c' n}$ for all sufficiently large $n$. 
\end{lemma}
\begin{proof}
By Lemma 1 of~\citet{laurent2000adaptive}, 
\begin{equation}\label{laurent1}
        \bbP^* \left\{ \frac{\chi^2_{d}}{d} \leq 1- a \right\} \leq e^{- a^2 d / 4},  \quad 
    \bbP^* \left\{ \frac{\chi^2_{d}}{d} \geq 1 + a + \frac{a^2}{2} \right\} \leq e^{-a^2d/4},  
\end{equation}
where $\chi^2_d$ denotes a chi-squared random variable with  $d$  degrees of freedom and  $a>0$ is arbitrary. 
Consider $\bbP^*(\cD)$ first. 
For any $j \in [p], \sigma \in \bbS^p$, and $S \in \model(2\din, P_j^{\sigma})$, by~\eqref{laurent1}, 
\begin{align*}
    \bbP^* \left\{ \frac{(z_j^\sigma)^\T \oproj_{S} z_j^\sigma}{n - |S|} \leq 1 - \frac{1}{4\eta} \right\} \leq \exp \left(-\frac{ n - |S|}{64 \eta^2} \right).
\end{align*}
Since  $|S| \leq 2 \din  = o(n / \log p)$,  $n (n-|S|)^{-1} (1 - (2\eta)^{-1}) \leq 1 - (4\eta)^{-1} $ for sufficiently large $n$. 
Applying the union bound with $|\cZ_1| \leq p^{\din + 1}$ and $|\model(2\din, P_j^\sigma)| \leq p^{2 \din +1}$, we obtain   
\begin{align*}
    \bbP^* (\cD^\c) \leq p^{3\din + 2} \exp \left(-\frac{n}{128 \eta^2} \right) \leq e^{-c' n}, 
\end{align*}
for sufficiently large $n$. 
Next, consider $\bbP^*(\cE)$. For any $j \in [p]$ and $S \in \model(2\din, P_j^{\sigma^*})$, we have 
\begin{align*}
    \bbP^* \left\{ \frac{z_j^\T \oproj_{S} z_j}{n - |S|} \geq 1 + \frac{1}{8 \eta} + \frac{1}{128 \eta^2} \right\} \leq \exp \left(-\frac{n - |S|}{256 \eta^2} \right),
\end{align*}
by~\eqref{laurent1}. Since $|\cZ_0| = p$ and $|\model(2\din, P_j^{\sigma^*})| \leq p^{2 \din +1}$, the union bound gives 
\begin{align*}
    \bbP^* (\cE^\c) \leq p^{2\din + 2} \exp \left(-\frac{n}{512 \eta^2} \right) \leq e^{-c'n}. 
\end{align*}
Another application of the union bound yields the conclusion.  
\end{proof}

\begin{lemma}\label{L:event}
Under the conditions of Proposition~\ref{prop:lower_bound}, 
we have $\bbP^*(\cA \cap \cB' \cap \cJ) \geq 1 - 6p^{-1}$ for all sufficiently large $n$.
\end{lemma}

\begin{proof}
We have obtained the bound $\bbP^*(\cA) \geq 1 - p^{-1}$ from Lemma~\ref{L:set_ordering}, and the bound on $\bbP^*(\cB')$ is proved in Lemma F2 of~\citet{zhou2021complexity}. 
Consider $\bbP^*(\cJ^\c)$. Let 
\begin{align*}
\cJ_{ij}^\c = \left\{ \left|\frac{X_i^\T X_j}{n} - \Sigma_{ij}^* \right| > 160 \vmax \sqrt{\frac{\log p}{n}}\right\}.
\end{align*}
By~\citet[Lemma 1]{ravikumar2011high}, 
\begin{align*}
   \bbP^*(\cJ_{ij}^\c) \leq 4 \exp(-3 \vmax^2 \log p / (\max_{i}\Sigma_{ii}^*)^2 )\leq 4 p^{-3},
\end{align*}
from which we obtain $\bbP^*(\cJ^\c) = \bbP^*(\cup_{i,j \in [p]} \cJ_{ij}^\c) \leq 4 p^{-1}$ by the union bound.
\end{proof}

\subsection{Proof of Proposition~\ref{prop:freq}}\label{proof:freq}
We consider the proof of consistency for the estimator $\hat{G}_\sigma^{\MAP}$ defined in~\eqref{def.map};  that is, we show that the scoring criterion $\phi$ is consistent when the ordering $\sigma$ is known. 
We first prove a technical lemma, which bounds the residual sum of squares $\SSE_j(G)$ when the node $j$ is underfitted (i.e., $\Pa_j(G^*) \not\subseteq \Pa_j(G)$). 

\begin{lemma}\label{L:lower_bound_betamin}
Fix some $S \subseteq [p]$ such that $|S| \leq \din$ and $S \neq S^* = \Pa_j(G^*)$. 
Suppose we are on the event $\cA \cap \cB \cap \cC$ and the conditions of Proposition~\ref{prop:freq} hold.   Then
\begin{align*}
    X_j^\T (\proj_{S \cup \{k_0\}} - \proj_{S} ) X_j \geq 
9 c_0 \vmax \log p / \alpha,
\end{align*}
for some $k_0 \in S^* \setminus  S$.
\end{lemma}
\begin{proof}
We denote  $X_j = Z_j + \epsilon_j$, $Z_j = X_{S^*} (B^*_{j})_{S^*}$, where $B_j^*$ is $j$-th column of the true weighted adjacency matrix $B^*$. Let $k_{0} = \argmax_{k \in S^* \backslash S} Z_j^\T (\proj_{S \cup \{k\}} - \proj_{S}) Z_j$. By the triangle inequality, 
\begin{equation}\label{lm2:tri}
\begin{aligned}
    X_j^\T (\proj_{S \cup \{k_0\} } - \proj_{S} ) X_j &\; \geq  (|| (\proj_{S \cup \{k_0\} } - \proj_{S}) Z_j || - || (\proj_{S \cup \{k_0\} } - \proj_{S} ) \epsilon_j ||  )^2. 
\end{aligned}
\end{equation}
On the event set $\cC$, we can use $c_0 > \alpha \rho$ from condition (C\ref{A:prior}) to obtain that 
\begin{align*}
    || (\proj_{S \cup \{k_0\} } - \proj_{S}) \epsilon_j ||^2 \leq \rho \omega_{j}^* \log p \leq \rho \vmax \log p < \frac{c_0}{\alpha} \vmax \log p, 
\end{align*}
and thus by Lemma E2 of~\citet{zhou2021complexity},  
\begin{align*}
    || (\proj_{S \cup \{k_0\} } - \proj_{S}) Z_j ||^2   \geq  
 \frac{||B^*_{S^*\backslash S} ||^2 }{|S^* \backslash S|}  \frac{ n \vmin^2 }{ \vmax} \geq 16 c_0 \frac{\vmax^2 \log p}{\alpha \vmin^2 n} \frac{ n \vmin^2 }{ \vmax} \geq \frac{16 c_0}{\alpha} \vmax \log p. 
\end{align*}
The second inequality follows from condition (C\ref{A:beta-min}). 
Plugging the above two displayed bounds into~\eqref{lm2:tri}, we obtain the asserted result. 
\end{proof}
 
\begin{proof}[Proof of Proposition~\ref{prop:freq}]
On  the event $\cA \cap \cB \cap \cC$ defined in Section~\ref{subsec:event}, we will show that all the three events stated in the proposition happen. For a non-negative integer $d$, define
\begin{align*}
\cG_p^*(d)  = \bigcup_{\sigma \in [\sigma^*]} \cG_p^\sigma(d). 
\end{align*}
 
\noindent \textit{Event (i).}  Fix an arbitrary $G \in  \cG_p^*(2\din) $ such that $\Pa_j(G^*) \subset  \Pa_j(G)$  for some $j \in [p]$. We prove that we can remove all the redundant parents of node $j$. This is slightly stronger than the asserted result, but it will be useful later for proving the claim for event (iii).  
Pick an arbitrary $k \in \Pa_j(G) \setminus \Pa_j(G^*)$ and define $G' = G \setminus \{k \rightarrow j\}$.    
On the event $\cB \cap \cC$, we have 
\begin{gather*}
    X_j^\T (\oproj_{\Pa_j(G')} - \oproj_{\Pa_j(G)}) X_j = \epsilon_j^\T (\proj_{\Pa_j(G)} - \proj_{\Pa_j(G')}) \epsilon_j \leq  \omega^*_{j} \rho \log p ,   \\
    \SSE_i(G) = X_i^\T \oproj_{\Pa_i(G)} X_i   \geq   \epsilon_i^\T \oproj_{\Pa_i(G) } \epsilon_i \geq \frac{n \omega^*_{i}}{2} \text{ for } i \in [p].  
\end{gather*} 
Since   $1+x \leq \exp(x)$ for   $x \in \bbR$ and  $ \sqrt{1 + \alpha/\gamma} > 1$, we find that 
\begin{align*}
    \frac{\exp(\phi(G))}{\exp(\phi(G'))} &\;= \left(p^{c_0}\sqrt{1+\alpha/\gamma}\right)^{-1} \left(\frac{\sum_{i\neq j}^p \SSE_i(G) + \SSE_j(G')}{\sum_{i=1}^p \SSE_i(G)   } \right)^{\frac{\alpha p n + \kappa}{2}} \\
    &\; <  p^{-c_0 } \left( 1 + \frac{X_j^\T (\oproj_{\Pa_j(G')} - \oproj_{\Pa_j(G)}) X_j}{\sum_{i=1}^p \SSE_i(G) } \right)^{\frac{\alpha p n + \kappa}{2}} \\
    &\; \leq p^{-c_0  } \exp \left( \frac{\alpha n p + \kappa}{2} \frac{X_j^\T (\oproj_{\Pa_j(G')} - \oproj_{\Pa_j(G)}) X_j}{\sum_{i=1}^p \SSE_i(G) } \right) \\
    &\; \leq p^{-c_0 } \exp \left\{    \frac{ (\alpha n p + \kappa) \omega_j^*  \rho \log p }{ (\min_{i}\omega_i^*) np  } \right\} \\
 &\; \leq p^{  \{ \max_{i\neq j} (\omega_j^*/ \omega_i^*) \} (\alpha+1) \rho - c_0 } < 1. 
\end{align*}
In the last line, we have used $\kappa \leq np$ and $c_0 > \max_{i\neq j} (\omega_j^*/ \omega_i^*)  (\alpha + 1)\rho$ from   condition (C\ref{A:prior}).   
The same argument implies that if we define  $G_0$ such that $\Pa_j(G_0) = \Pa_j(G^*)$ and $\Pa_i(G_0) = \Pa_i(G)$ for $i \neq j$, then we have 
\begin{align*}
     \frac{\exp(\phi(G))}{\exp(\phi(G_0))} < p^{(|\Pa_j(G)| - |\Pa_j(G^*)|) \{\max_{i\neq j} (\omega_j^*/ \omega_i^*) (\alpha+1) \rho  - c_0\}} < 1. 
\end{align*} 

\medskip 
\noindent \textit{Event (ii).}   Fix an arbitrary $G \in \cG_p^*(\din) $ such that $\Pa_j(G^*) \not\subseteq \Pa_j(G)$  for some $j \in [p]$. 
Since there exists some $\sigma \in [\sigma^*]$ such that $G, G^* \in \cG_p^\sigma(\din)$, we can apply Lemma~\ref{L:lower_bound_betamin} to show that there exists some $k \in \Pa_j(G^*) \setminus \Pa_j(G)$ such that the DAG $G' = G \cup \{k \rightarrow j\}$ satisfies
$ X_j^\T (\proj_{\Pa_j(G')} - \proj_{\Pa_j(G)}) X_j \geq 9  c_0 \vmax \log p / \alpha$.
Further, on the event $\cA$, we have $\SSE_i(G) \leq X_i^\T X_i \leq n \vmax$. 
Now using $\sqrt{1+\alpha/\gamma} \leq p$, which follows from condition (C\ref{A:prior}), we find that 
\begin{align*} 
    \frac{\exp(\phi(G))}{\exp(\phi(G'))} &\;= \left(p^{c_0}\sqrt{1+\alpha/\gamma}\right)  \left(\frac{\sum_{i\neq j}^p  \SSE_i(G) +  \SSE_j(G') }{\sum_{i=1}^p \SSE_i(G) } \right)^{\frac{\alpha p n + \kappa}{2}} \\
    &\; \leq  p^{  (c_0 + 1 )} \left( 1 - \frac{X_j^\T (\oproj_{\Pa_j(G)} - \oproj_{\Pa_j(G')}) X_j}{\sum_{i=1}^p \SSE_i(G) } \right)^{\frac{\alpha p n + \kappa}{2}} \\
    &\; \leq p^{   (c_0 + 1 )} \exp \left( - \frac{\alpha n p + \kappa}{2} \frac{X_j^\T (\proj_{\Pa_j(G')} - \proj_{\Pa_j(G)}) X_j}{\sum_{i=1}^p \SSE_i(G)} \right) \\
    & \;  \leq p^{  (c_0 + 1)} \exp\left\{   - \frac{\alpha n p + \kappa}{2} \frac{9   c_0 \vmax \log p / \alpha }{n p \vmax }\right\} \leq p^{ (- 7c_0/2 + 1)}. 
\end{align*} 
This implies $\exp(\phi(G)) < \exp(\phi(G'))$ since $c_0  > 4 \din + 6 > 2/7$ . 
The same argument shows that if we define $G_1 \in \cG_p^\sigma$ such that $\Pa_j(G_1) = \Pa_j(G^*) \cup \Pa_j(G)$  and $\Pa_i(G_1) = \Pa_i(G)$ for $i \neq j$, then we have 
\begin{equation}\label{eq:under}
    \frac{\exp(\phi(G))}{\exp(\phi(G_1))}   \leq p^{|\Pa_j(G^*) \setminus \Pa_j(G)| (- 7c_0/2 + 1)}. 
\end{equation} 

\medskip 

\noindent \textit{Event (iii).}  Consider an arbitrary $G \in \cG_p^*(\din)$ such that $G \neq G^*$. Then, there exists some $j \in [p]$ such that $\Pa_j(G) \neq \Pa_j(G^*)$.  If the node $j$ is overfitted (i.e., $\Pa_j(G^*) \subset  \Pa_j(G)$), event~(i) shows that there exists some $G_0 \in \cG_p^*(\din)$ such that $\phi(G_0) > \phi(G)$. 
If the node $j$ is underfitted, i.e., $\Pa_j(G^*) \not\subseteq \Pa_j(G)$,  inequality~\eqref{eq:under} shows that there exists some $G_1 \in \cG_p^*(2\din)$ such that $\phi(G_1) > \phi(G)$ and node $j$ is overfitted. But event (i) again implies that there exists some $G_2 \in \cG_p^*(\din)$ such that $\phi(G_2) > \phi(G_1)$. 
Hence, $G$ cannot be the maximizer of $\phi$ in $\cG_p^\sigma(\din)$; that is, 
$G^*$ is the unique DAG in $\cG_p^*(\din)$ that maximizes $\phi$, which completes the proof. 
\end{proof}

\subsection{Proof of Theorem~\ref{thm.consistency}}\label{proof:consistency}

For $\tau \notin [\sigma^*]$, the ratio of $\exp( \phi(\hat{G}_\tau) )$ to $\exp( \phi(G^*) )$ is
\begin{align}\label{eq:ratio}
    \frac{ \exp( \phi(\hat{G}_\tau) )}{\exp( \phi(G^*) ) } = \left(p^{c_0}\sqrt{1+\alpha/\gamma}\right)^{| G^*|-|\hat{G}_{\tau}|} \left( \frac{\sum_{j=1}^p  \SSE_j(\hat{G}_\tau) }{\sum_{j=1}^p \SSE_j(G^*) } \right)^{-\frac{\alpha p n + \kappa}{2}}. 
\end{align} 
On the event $\cD \cap \cE$ defined in Section~\ref{subsec:event}, we have 
\begin{align*}
    \frac{\sum_{j=1}^p \mathrm{RSS}_j\left(\hat{G}_\tau\right)}{\sum_{j=1}^p \mathrm{RSS}_j\left(G^*\right)} & \geq \frac{\sum_{j=1}^p X_j^{\mathrm{T}} \Phi_{\mathrm{Pa}_j\left(\hat{G}_\tau\right) \cup \mathrm{Pa}_j\left(G_\tau^*\right)}^{\perp} X_j}{\sum_{j=1}^p X_j^{\mathrm{T}} \Phi_{\mathrm{Pa}_j\left(G^*\right)}^{\perp} X_j} \\
& =\frac{\sum_{j=1}^p(\epsilon_j^\tau)^{\mathrm{T}} \Phi_{\mathrm{Pa}_j\left(\hat{G}_\tau\right) \cup \mathrm{Pa}_j\left(G_\tau^*\right)}^{\perp} \epsilon_j^\tau}{\sum_{j=1}^p \epsilon_j^{\mathrm{T}} \Phi_{\mathrm{Pa}_j\left(G^*\right)}^{\perp} \epsilon_j} \\
& \geq \frac{\mathrm{tr}(\Omega^*_{\tau})}{\mathrm{tr}(\Omega^*_{\sigma^*})} \cdot \frac{(1-1 /(2 \eta))}{1+1 /(4 \eta)} 
\end{align*}
where the error vectors $\epsilon_j, \epsilon_j^\tau$ are as defined in~\eqref{st.eq:true} and~\eqref{st.eq:minimap}. 
Without loss of generality, we can assume $\eta > 3$ in Assumption~\ref{A:omega}, from which we obtain that 
\begin{align*}
       \frac{\sum_{j=1}^p  \SSE_j(\hat{G}_\tau) }{\sum_{j=1}^p \SSE_j(G^*) }  \geq 
     \frac{(1+1/\eta)(1 - 1/(2\eta))}{1+1/(4 \eta)}  > \frac{1+1/(3\eta)}{1+1/(4 \eta)} > 1+\frac{1}{\eta'},
\end{align*}
for some universal $\eta' > 0$.  Hence, 
\begin{align*}
    \frac{ \exp( \phi(\hat{G}_\tau) )}{\exp( \phi(G^*) ) }  
    \leq p^{c_0 |G^*|} \left(1 + \frac{1}{\eta'} \right)^{-\frac{\alpha p n + \kappa}{2}}  
    \leq p^{c_0 p \din } \left(1 + \frac{1}{\eta'} \right)^{-\frac{\alpha p n + \kappa}{2}}.   
\end{align*} 
Using  $\din \log p = o(n)$ and Stirling's formula, we get 
\begin{align*}
    \frac{\sum_{\tau \notin [\sigma^*] }\exp( \phi(\hat{G}_\tau) )}{\exp( \phi(G^*) )} \leq p! \frac{ \exp( \phi(\hat{G}_\tau) )}{\exp( \phi(G^*) ) }    \leq e^{- C n p}, 
\end{align*}
for some universal  $C > 0$.  
For sufficiently large $n$,  by Assumption~\ref{A:freq} and Lemma~\ref{L:event_theorem},  the event $\cD \cap \cE \cap \left( \cap_{\sigma \in [\sigma^*]} \{ \hat{G}_{\sigma} = G^*  \} \right)$ happens with probability at least $1 - \zeta(p) - 2 e^{-c'n}$, on which we have 
\begin{align*}
    \post( G^* ) = \frac{ \sum_{\sigma \in [\sigma^*]} e^{\phi(G^*)} }{ \sum_{\tau \in \bbS^p} e^{\phi(\hat{G}_\tau)} } \geq 1 - \frac{ \sum_{\tau \notin [\sigma^*]} e^{\phi(\hat{G}_\tau)} }{ \sum_{\sigma \in [\sigma^*]} e^{\phi(G^*)} } \geq 1 - e^{-Cnp}. 
\end{align*} 
That is, $\post(G^*)$ converges to $1$ in probability. 
\hfill\qedsymbol{}

\subsection{Proof for the case of sub-Gaussian errors}\label{proof:subgaussian}
Let $X$ be an $n \times p$ random matrix, each of whose rows is an i.i.d. copy of $p$-dimensional sub-Gaussian random vector with mean zero and covariance matrix $\Sigma^*$ with a sub-Gaussian parameter bounded by a universal constant $C_\mathrm{sub}$. We define $\Sigma^*_{S}$ as the submatrix of $\Sigma^*$ with both rows and columns indexed by the set $S$. Let $\Sigma_{j|S}^* = \Sigma^*_{j,j} - \Sigma^*_{j,S}(\Sigma^*_{S})^{-1}\Sigma^*_{S, j}$ denote the partial covariance and let
$\hat{\Sigma}_{j|S} =  n^{-1} X_j \oproj_S X_j$ be its estimator for $|S| \leq \din$ and $j \notin S$. Denote $\left\| \cdot \right\|_{\mathrm{op}}$ as the operator norm. 

In the sub-Gaussian case, zero correlation does not imply independence anymore, and thus we need more stringent assumptions. The first condition is that
\begin{align}\label{eq:assume_sub1}
     \frac{\vmax^4d_{\mathrm{in}} \log p}{\vmin^6n} \rightarrow 0,
\end{align}
as $n$ goes to infinity. Second, we need $\Pa_j(\hat{G}_\tau) \subseteq \Pa_j(G^*_\tau)$ for $\tau \notin [\sigma^*]$, which means that the stepwise selection method should estimate the minimal I-map $G^*_\tau$ sparser and should not include an edge that is not in $G^*_\tau$. 
For the consistency result, the ratio $\hat{\Sigma}_{j|S}/\Sigma^*_{j|S}$ need to be controlled. To this end, we need the following lemmas. 
\begin{lemma}\label{L:subG1}
Suppose $\din \log p = o(n)$. There exists a constant $K_0$, which only depend on $C_\mathrm{sub}$, satisfying for sufficiently large $n$,
\begin{align*}
     \max _{S \in \mathcal{M}_{p}\left(2 d_{\mathrm{in}},  [p]\right)}\left\|n^{-1} X_{S}^{\top} X_{S}-\Sigma_{S}^{*}\right\|_{\mathrm{op}} \leq K_{0} \sqrt{\frac{d_{\mathrm{in}} \log p}{n}},
\end{align*}
with probability at least $1- 2p^{-\din}$.
\end{lemma}
\begin{proof}
See Lemma F3 in~\citet{zhou2021complexity}.
\end{proof}
\begin{lemma}\label{L:subG2}
Suppose $\din \log p = o(n)$ and a set $S$ and $j$ satisfy $|S| \leq \din$ and $j \notin S$. Let $K_0$ be the constant in Lemma~\ref{L:subG1}. Then, for sufficiently large $n$, we have
\begin{align*}
    |\hat{\Sigma}_{j|S} - \Sigma_{j|S}^{*}| \leq  K_{0} \frac{\vmax^2}{\vmin^2} \sqrt{\frac{d_{\mathrm{in}} \log p}{n}},
\end{align*}
with probability at least $1- 2p^{-\din}$.
\end{lemma}
\begin{proof} 
Apply the proof of Lemma E4 of~\citet{zhou2021complexity} by setting $T = \{j\}$, where $T$ is a set defined in Lemma E4 of~\citet{zhou2021complexity}. 
\end{proof}
Now, we are ready to prove the sub-Gaussian case. It is sufficient to show
\begin{align*}
    \frac{\sum_{j=1}^p  \SSE_j(\hat{G}_\tau) }{\sum_{j=1}^p \SSE_j(G^*) } > 1+\frac{1}{\eta'}.
\end{align*}
For fixed $\eta > 0$, by the condition \eqref{eq:assume_sub1}, a sufficiently large $n$ satisfies $K_{0} (\vmax^2/\vmin^2) \sqrt{d_{\mathrm{in}} \log p/n}$ $< \vmin/(4\eta)$. It follows that 
\begin{align*}
    \hat{\Sigma}_{j|S} & > \Sigma_{j|S}^{*} - K_{0} \frac{\vmax^2}{\vmin^2} \sqrt{\frac{\din \log p}{n}} \\
     & >\Sigma_{j|S}^{*} - \frac{\vmin}{2 \eta},
\end{align*}
which implies that $\hat{\Sigma}_{j|S}  /\Sigma_{j|S}^{*} > 1 - (2\eta)^{-1}$ by the fact $\vmin \leq \Sigma_{j|S}^{*}$. 
The other direction can be obtained by 
\begin{align*}
    \hat{\Sigma}_{j|S} & < \Sigma_{j|S}^{*} + K_{0} \frac{\vmax^2}{\vmin^2} \sqrt{\frac{\din \log p}{n}} \\
     & < \Sigma_{j|S}^{*} + \frac{\vmin}{4 \eta},
\end{align*}
which yields $\hat{\Sigma}_{j|S}  /\Sigma_{j|S}^{*} < 1 + (4\eta)^{-1}$. Therefore, 
\begin{align*}
    \frac{\sum_{j=1}^p \mathrm{RSS}_j\left(\hat{G}_\tau\right)}{\sum_{j=1}^p \mathrm{RSS}_j\left(G^*\right)} & \geq \frac{\sum_{j=1}^p X_j^{\mathrm{T}} \Phi_{\mathrm{Pa}_j\left(G_\tau^*\right)}^{\perp} X_j}{\sum_{j=1}^p X_j^{\mathrm{T}} \Phi_{\mathrm{Pa}_j\left(G^*\right)}^{\perp} X_j} \\
& =\frac{\sum_{j=1}^p \hat{\Sigma}_{j|\mathrm{Pa}_j\left(G_\tau^*\right)}}{\sum_{j=1}^p \hat{\Sigma}_{j|\mathrm{Pa}_j\left(G^*\right)}} \\
& \geq \frac{\mathrm{tr}(\Omega^*_{\tau})}{\mathrm{tr}(\Omega^*_{\sigma^*})} \cdot \frac{(1-1 /(2 \eta))}{1+1 /(4 \eta)} \\
& \geq \frac{(1+1/\eta) (1-1 /(2 \eta))}{1+1 /(4 \eta)} > 1 + \frac{1}{\eta'},
\end{align*}
for some universal constant $\eta' > 0$. 
The rest of the proof is identical to the Gaussian case.
\hfill\qedsymbol{}

\subsection{Proof of Proposition~\ref{prop:lower_bound}}\label{proof:lower_bound}
By $(C\ref{c2.beta.min}')$ , we have $\omega^*_1 = \dots = \omega^*_p = \omega^*$ in~\eqref{st.eq:true} for the true data generating model.
Without loss of generality, assume that $\id = (1, \dots, p)$ is a true ordering. Define
\begin{align*}
    \theta = \din^{2} \frac{\vmax^2 \log p}{\vmin^3 n}.
\end{align*}
\begin{lemma}\label{L:asym}
Under the setting of Proposition~\ref{prop:lower_bound}, 
\begin{align*}
    \Sigma_{ii}^*   = \omega^* + O(\theta/\din),  \quad \quad 
     \Sigma_{ij}^*   = O( \sqrt{\theta} / \din), 
\end{align*}
for all $i, j \in [p]$ and $i \neq j$.
\end{lemma}
\begin{proof}
For ease of notation, in this proof we write $B = B^*$, and without loss of generality, we assume the true error variance $\omega^*$ equals 1.
Since $B$ is a strictly upper triangular matrix, its operator norm is zero and $B^p = 0$. So we can expand $\Sigma$ using the Neumann series by 
\begin{align*}
    \Sigma =&\; (I - B^\T)^{-1} (I - B)^{-1}  = \sum_{k = 0}^\infty (B^\T)^k \sum_{k = 0}^\infty B^k \\
    &\; = \sum_{k=0}^\infty \sum_{r+s = k} (B^\T)^r B^s  =  \sum_{k=0}^{2p-2} \sum_{\substack{r+s = k \\ r,s < p}} (B^\T)^r B^s.
\end{align*} 
We can calculate $B^s$ and $(B^\T)^r$ by treating $B^*$ and $(B^*)^\T$ as weighted transition matrices for a random walk on the DAG  with weighted adjacency matrix $B$. Explicitly, define the set of all paths from node $i$ to node $j$ with $s$ steps by 
\begin{align*}
    \mathrm{PATH}_{ij}^s &\; = \{q = (q_0, q_1, \dots, q_s)\colon  
    B_{q_k q_{k+1}} \neq 0,  \text{ for } k = 0, \dots, s-1, q_0 = i, q_s = j \},
\end{align*}
and the weight $W_q$ of an $s$-length path $q = (q_0, \dots, q_s)$ by  $W_q =\prod_{k=1}^{s} B_{q_{k-1} q_{k}}.$ 
We have $| \mathrm{W}_q | = O(\theta^{s/2}/ \din^{s})$, since $|B_{ij}| = O( \sqrt{\theta} / \din)$ for any $i, j$ by the condition (C\ref{c2.beta.min}'). It follows that the $(i,j)$-th entry of $(B^\T)^r B^s$ is given by 
\begin{align*}
      \left( (B^\T)^r B^s \right)_{ij} =& \sum_{k \in [p]} (B^\T)^r_{ik} B^s_{kj} = \sum_{k \in [p]}\left( \sum_{q \in \mathrm{PATH}_{kj}^s} W_q \right) \left(\sum_{q \in \mathrm{PATH}_{ki}^r} W_q  \right) \\
      =&  \sum_{k \in [p]}  \sum_{q \in \mathrm{PATH}_{kj}^s, q' \in \mathrm{PATH}_{ki}^r}  W_{q'}W_q = N^{r,s}(i, j) O( \theta^{(r + s)/2} / \din^{r+s} ),  
\end{align*}
where $N^{r,s}(i, j)$ denotes the number of possible ``paths''  that start from node $i$, move backwards for $r$ steps, move forwards for $s$ steps and arrive at node $j$; such paths are called treks~\citep{uhler2013geometry,sullivant2010trek} and we denote them by $q = (q'_0, q'_{1},  \dots, q'_{r-1}, q'_{r} = q_{s}, q_{s - 1}, \dots, q_1, q_0)$, where $q'_0 = i, \, q_0 = j.$  
Since $d$ is the maximum number of parent nodes, given $i, j$, there are at most $\din$ different choices for $q'_1$ and $q_1$. Similarly, given $q'_1$ and $q_1$, there are at most $\din$ choices for $q'_2$ and $q_2$. Repeating this argument yields that $N^{r, s}(i, j) \leq \din^{r + s - 1}$, and it follows that $ \left( (B^\T)^r B^s \right)_{ii} = O( \theta^{(r + s) / 2} / \din).$  
Therefore,  for sufficiently large $n$,
\begin{align*}
    \Sigma_{ii} &\; = \sum_{k=0}^{2p-2} \sum_{\substack{r+s = k \\ r,s < p}} ( (B^\T)^r B^s)_{ii}\\
    &\; = 1 + \sum_{k=2}^{p} \sum_{1 \leq r \leq k-1} ( (B^\T)^r B^{k-r})_{ii} + \sum_{k=p+1}^{2p - 2} \sum_{k - p + 1 \leq r \leq p - 1 } ( (B^\T)^r B^{k - r})_{ii} \\
    &\; = 1 + \sum_{k=2}^{p} \din^{-1} (k-1) O (\theta^{k/2})  + \sum_{k=p+1}^{2p-2} \din^{-1}(2p- 1 -k) O ( \theta^{k/2}) \\
    &\; =  1 + \sum_{k=2}^\infty \din^{-1} O(2^{k-2} \theta^{k/2}) 
    = 1 + O(\theta/\din). 
\end{align*} 
Similarly, for any $i < j$,
\begin{align*}
    \Sigma_{ij} &\; = \sum_{k=0}^{2p-2} \sum_{\substack{r+s = k \\ r,s < p}}  ((B^\T)^r B^s)_{ij}  \\
    &\; = B_{ij} +  \sum_{k=2}^{p} \din^{-1} (k-1) O (\theta^{k/2})  + \sum_{k=p+1}^{2p-2} \din^{-1}(2p- 1 -k) O ( \theta^{k/2}),
\end{align*}
from which we obtain that $\Sigma_{ij} = O( \sqrt{\theta} / \din) + O(\theta/\din) = O( \sqrt{\theta} / \din). $
\end{proof}

\begin{figure}[t]
    \centering
    \includegraphics[width=.5\linewidth]{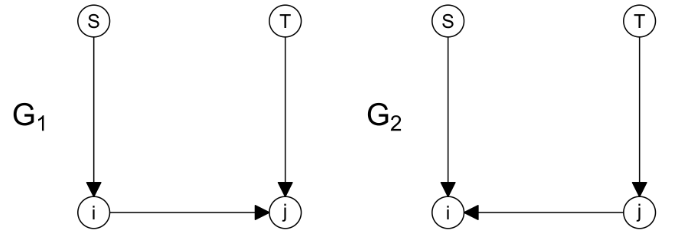}
    \caption{Local structure of $G_1, G_2$ in the proof of  Proposition~\ref{prop:lower_bound}.}
    \label{fig:rev.edge}
\end{figure}

\begin{proof}[Proof of Proposition~\ref{prop:lower_bound}]
Define $\dag(\din) = \cup_{\sigma \in \bbS^p} \cG_p^\sigma(\din)$. 
Let $G_1, G_2 \in \dag(\din)$ be such that $\{i \rightarrow j\} \in G_1$ and $G_2$ can be obtained from $G_1$ by reversing $i \rightarrow j$. Let $S = \Pa_i(G_1)$ and $T = \Pa_j(G_2)$; see Fig.~\ref{fig:rev.edge}. The sets $S$ and $T$ may not be disjoint. 

Assume we are on the event $ \cB' \cap \cJ$ defined in Section~\ref{subsec:event}. Since $G_1, G_2$ have the same number of edges, the posterior ratio of $G_1$ to $G_2$ is 
\begin{align*}
    \frac{\exp(\phi(G_1))}{\exp(\phi(G_2))} &\;= \left(\frac{\sum_{k=1}^{p} X_k^\T \oproj_{\Pa_k(G_2)} X_k}{\sum_{k=1}^{p} X_k^\T \oproj_{\Pa_k(G_1)} X_k} \right)^{\frac{\alpha p n+\kappa}{2}} \\
    &\;= \left(1 + \frac{X_j^\T(\proj_{T\cup\{i\}} -\proj_T)X_j - X_i^\T(\proj_{S\cup\{j\}} -\proj_S)X_i}{\sum_{k=1}^{p} X_k^\T \oproj_{\Pa_k(G_1)} X_k} \right)^{\frac{\alpha p n+\kappa}{2}} \\
    &\; \leq \exp \left( \frac{\alpha p n+\kappa}{2} \frac{X_j^\T(\proj_{T\cup\{i\}} -\proj_T)X_j - X_i^\T(\proj_{S\cup\{j\}} -\proj_S)X_i}{np \vmin /2}\right) \\
    &\; \leq \exp \left\{ \frac{ \alpha+1 }{\vmin } [X_j^\T(\proj_{T\cup\{i\}} -\proj_T)X_j - X_i^\T(\proj_{S\cup\{j\}} -\proj_S)X_i]  \right\},
\end{align*}
where the first inequality follows from the inequality $1+x \leq \exp(x)$ for all $x \in \bbR$ and the second follows from the observation that $X_k^\T \oproj_{\Pa_k(G_1)} X_k \geq n \vmin / 2$  for any $k \in [p]$ on the event $\cB'$.  
To conclude the proof, we need to show 
\begin{equation}\label{eq:bound}
    X_j^\T(\proj_{T\cup\{i\}} -\proj_T)X_j - X_i^\T(\proj_{S\cup\{j\}} -\proj_S)X_i  =    o( (\vmax^2/\vmin^2) \log p ).  
\end{equation} 
By Lemma~\ref{L:asym} and condition (C\ref{c2.din}'), on the event $\cJ$, we have 
\begin{align*}
    \frac{X_i^\T X_i}{n} &\; = \Sigma_{ii} + O(\vmin \sqrt{\theta}/ \din) = \omega^* + O(\theta / \din) + O(\vmin \sqrt{\theta}/ \din)= \omega^* + o(1), \\
    \frac{X_i^\T X_j}{n} &\; = \Sigma_{ij} + O(\vmin \sqrt{\theta}/ \din)= O(\sqrt{\theta}/ \din) = o(1).
\end{align*}
Hence, by Neumann series, for any $S \subseteq [p]$  such that $|S| \leq \din$, we have $(n^{-1}X_S^\T X_S)^{-1}  = (\omega^*)^{-1}I + R_S$  where $R_S$ is a matrix with all entries being $O(\sqrt{\theta} / \din )$. This yields, for all $i,j \in [p] \setminus S$, 
\begin{align*}
     \frac{X_i^\T \proj_S X_j}{n}  
    &\; = \frac{X_i^\T X_S}{n} \left( \frac{X_S^\T X_S}{n} \right)^{-1} \frac{X_S^\T X_j}{n} \\
&\; = \begin{bmatrix}
O(\sqrt{\theta}/ \din) & \cdots & O(\sqrt{\theta}/ \din)
\end{bmatrix} ((\omega^*)^{-1}I + R_S ) \begin{bmatrix}
O(\sqrt{\theta}/ \din)\\
\vdots \\
O(\sqrt{\theta}/ \din))
\end{bmatrix}  \\
&\; = \din  O( \theta / \din^2) + \din^2 O( \theta^{3/2} / \din^3 ) = O( \theta / \din ) = o(1). 
\end{align*}
It follows that
\begin{align*}
   &\; X_j^\T(\proj_{T\cup\{i\}} -\proj_T)X_j - X_i^\T(\proj_{S\cup\{j\}} -\proj_S)X_i = \frac{(X_j^\T\oproj_TX_i)^2}{X_i^\T\oproj_TX_i} - \frac{(X_j^\T\oproj_SX_i)^2}{X_j^\T\oproj_S X_j} \\
    &\;= n \frac{\left[\frac{X_j^\T X_i}{n} - \frac{X_j^\T \proj_T X_i}{n} \right]^2}{\frac{X_i^\T X_i}{n} - \frac{X_i^\T \proj_T X_i}{n}} - n \frac{\left[\frac{X_j^\T X_i}{n} - \frac{X_j^\T \proj_S X_i}{n} \right]^2}{\frac{X_j^\T X_j}{n} - \frac{X_j^\T \proj_S X_j}{n}} \\
    &\; = n (\omega^*)^{-1} \left\{ (1+o(1)) \left[\frac{X_j^\T X_i}{n} - \frac{X_j^\T \proj_T X_i}{n} \right]^2 - (1+o(1)) \left[\frac{X_j^\T X_i}{n} - \frac{X_j^\T \proj_S X_i}{n} \right]^2\right\} \\
    &\; = n (\omega^*)^{-1} \left\{ - \frac{2 X_j^\T X_i}{n} \left[\frac{X_j^\T \proj_T X_i}{n} - \frac{X_j^\T \proj_S X_i}{n} \right] + \left( \frac{X_j^\T \proj_T X_i}{n} \right)^2  - \left( \frac{X_j^\T \proj_S X_i}{n} \right)^2+ o(\theta/ \din^2) \right\}\\
    &\; = n \left\{  O(\sqrt{\theta}/ \din)  O(\theta/ \din)+  O(\theta^2 / \din^2) + o(\theta/ \din^2) \right\} = n   o(\theta/ \din^2)= o((\vmax^2/\vmin^2) \log p),
\end{align*}
which completes the proof of \eqref{eq:bound}. 
\end{proof}

\subsection{Proof of Theorem~\ref{thm:ITD}}\label{proof:ITD}

Let $\delta =  \vmin^2 \betamin   (\din + 1)^{-1} (\vmin \betamin + 3 \omega^*(1+\betamin))^{-1}$ and $\hat{\Sigma}_{ij} = X_i^\T X_j / n$ for each $(i, j)$.
Define 
     $\cK = \left\{ \max_{i,j \in [p]} | \hat{\Sigma}_{ij} - \Sigma_{ij}^*| \leq \delta \right\}.$   
For any $\epsilon > 0$, using Lemma 1  of \citet{ravikumar2011high} and our Lemma~\ref{L:set_ordering}, we can show that $\bbP^*( \cA \cap \cB \cap \cC \cap \cK ) \geq 1 - \epsilon$ and 
\begin{align*}
     \bbP^*\{|\hat{\Sigma}_{ij} - \Sigma_{ij}^*| > \delta\} \leq 4 \exp\left\{  - \frac{n \delta^2}{3200 \max_k(\Sigma^*_{ij})^2} \right\} \leq \frac{\epsilon}{p(p+1)}. 
\end{align*}
Further, from the proof of Proposition~\ref{prop:freq}, we know that on the event $\cA \cap \cB \cap \cC$, we have 
$$\argmax_{S \subset P_j \colon  |S| \leq \din} \phi_j\left(S,  \sum_{i \neq j} \SSE_i(G) \right) = \Pa_j(G^*),$$ 
for any $j \in [p]$,  $P_j \supseteq \Pa_j(G^*)$, and $G \in \cG^*_p(2\din)$. 
Observe that Theorem~\ref{thm:ITD} holds if we can show that for any $G \in \cG^*_p(\din)$, Algorithm~\ref{alg:STD} with input $\SSE = (\SSE_1(G), \dots, \SSE_p(G))$ returns some $\sigma \in [\sigma^*]$, but this follows by an argument completely analogous to the proof of Theorem 2 of~\citet{chen2019causal}. 
\hfill\qedsymbol{} 

\subsection{Derivation of the posterior distribution }\label{subsec:post}
Let $L(B, \omega)$ be the likelihood function in~\eqref{eq:ln.str.eq}. The $\alpha$-fractional posterior distribution of $B, \omega$, given the prior distributions in~\eqref{prior:B} and \eqref{prior:omega}, is 
\begin{align*}
    \pi_n (B, \omega \mid G, \sigma) & \propto \pi_0 (B, \omega \mid G, \sigma) L(B, \omega)^\alpha \\
    &= \frac{\pi_0 (B, \omega \mid G, \sigma)}{L(B, \omega)^{1-\alpha}}L(B, \omega),
\end{align*}
where the first term in the last equation can be regarded as the effective prior distribution for $(B, \omega) \mid (G, \sigma)$. By the normal-inverse-gamma conjugacy, the $\alpha$-fractional marginal likelihood of $(G, \sigma)$ is given by
\begin{align*}
    & f_\alpha(G, \sigma) \propto \int \pi_0 (B, \omega \mid G, \sigma) L(B, \omega)^\alpha d(B, \omega) \\
    & = \int \pi_0 (B \mid  \omega, G, \sigma) \pi_0 ( \omega \mid G, \sigma) L(B, \omega)^\alpha d(B, \omega) \\
    & \propto \int \left(\frac{\omega}{\gamma}\right)^{-|G|/2} \prod_{j=1}^p \mathrm{det}\left(X_{\Pa_j}^\T  X_{\Pa_j}\right)^{1/2} \exp\left\{-\frac{\gamma}{2\omega} \sum_{j=1}^p ( B_{\Pa_j,j} - \hat{B}_{\Pa_j,j} )^\mathrm{T}(X_{\Pa_j}^\mathrm{T}X_{\Pa_j})(B_{\Pa_j,j} - \hat{B}_{\Pa_j,j})  \right\}  \times \\ 
    & (\omega^{-\frac{\kappa}{2}-1}) \left[\omega^{-\frac{\alpha n p}{2}}  \exp\left\{-\frac{\alpha}{2\omega} \sum_{j=1}^p (X_j -B_{\Pa_j,j}^\mathrm{T} X_{\Pa_j} )^\mathrm{T}(X_j -B_{\Pa_j,j}^\mathrm{T} X_{\Pa_j} )  \right\}\right] d(B,\omega)  \\
    & \propto  \int \left(\frac{\omega}{\gamma}\right)^{-|G|/2} \omega^{-\frac{\alpha np + \kappa}{2} -1} \exp \left\{ - \frac{\alpha}{2\omega} \sum_{j=1}^p X_j^\mathrm{T} \oproj_{\Pa_j} X_j \right\}  \left(\frac{\alpha +\gamma}{\omega}\right)^{-|G|/2} \times \\
    & \int \left(\frac{\omega}{\alpha +\gamma}\right)^{-|G|/2} \prod_{j=1}^p \mathrm{det}\left(X_{\Pa_j}^\T  X_{\Pa_j}\right)^{1/2} \times \\
    & \exp\left\{-\frac{\alpha + \gamma}{2\omega} \sum_{j=1}^p ( B_{\Pa_j,j} - \hat{B}_{\Pa_j,j} )^\mathrm{T}(X_{\Pa_j}^\mathrm{T}X_{\Pa_j})(B_{\Pa_j,j} - \hat{B}_{\Pa_j,j})\right\} dB d\omega  \\
    & = \left(1+\frac{\alpha}{\gamma}\right)^{-|G|/2} \int  \omega^{-\frac{\alpha np + \kappa}{2} -1} \exp \left\{ - \frac{\alpha}{2\omega} \sum_{j=1}^p X_j^\mathrm{T} \oproj_{\Pa_j} X_j \right\}  d\omega\\
    & \propto \left(1+\frac{\alpha}{\gamma}\right)^{-|G|/2} \left(\sum_{j=1}^p \mathrm{RSS}_j(G)\right)^{-\frac{\alpha np + \kappa}{2}}. 
\end{align*}
Given the prior distribution~\eqref{prior:G}, 
we obtain the posterior distribution of $(G, \sigma)$ as 
\begin{align*}
    \pi_n(G, \sigma) & \propto f_\alpha(G, \sigma)\pi_0 (G, \sigma) \\
    & = \left(1+\frac{\alpha}{\gamma}\right)^{-|G|/2} \cdot \left(\sum_{j=1}^p \mathrm{RSS}_j(G)\right)^{-\frac{\alpha np + \kappa}{2}} \cdot p^{-c_0 \log p} \cdot \ind_{\{ \hat{G}_\sigma\}}(G) \\
    &  = e^{\phi(G)} \ind_{\{ \hat{G}_\sigma\}}(G). 
\end{align*}

\section{Simulation results}\label{subsec:simulation}

\subsection{Mixing behavior}\label{subsec:mixing_supp} 
In Fig.~\ref{fig:mixing1} we examine the mixing behavior of the three types of proposals for a moderately small sample size. 
We repeat the simulation studies shown in panels (a), (b), and (c) of Fig.~\ref{fig:mixing} in Section~\ref{subsec:mixing} by choosing $n = 100$ and keeping all the other simulation settings unchanged. 
We confirm that all $90$ trajectories have reached the red line, which appears to be the global mode.   
Figure~\ref{fig:mixing2} shows the mixing behavior of our method and the minimal I-MAP MCMC for the heterogeneous case where, for each $j \in [p]$, we sample error variance $\omega_j$ for node $j$ uniformly from $[0.5, 1.5]$. 
We still observe that some trajectories of the minimal I-MAP MCMC get stuck at local modes, while the mixing behavior of the proposed method is consistently good despite of the model misspecification.  

\begin{figure}[t!]
    \centering
    \includegraphics[width=0.32\textwidth]{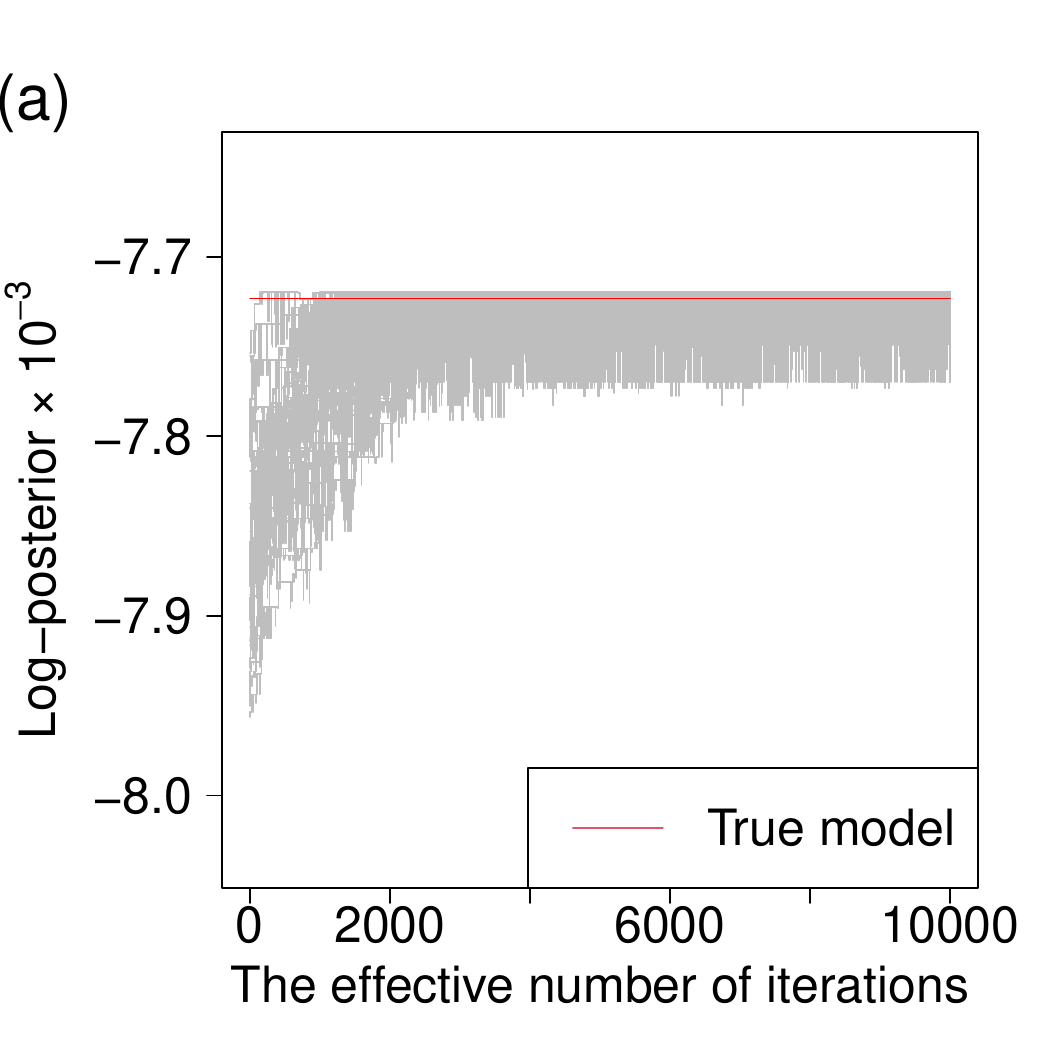}
    \includegraphics[width=0.32\textwidth]{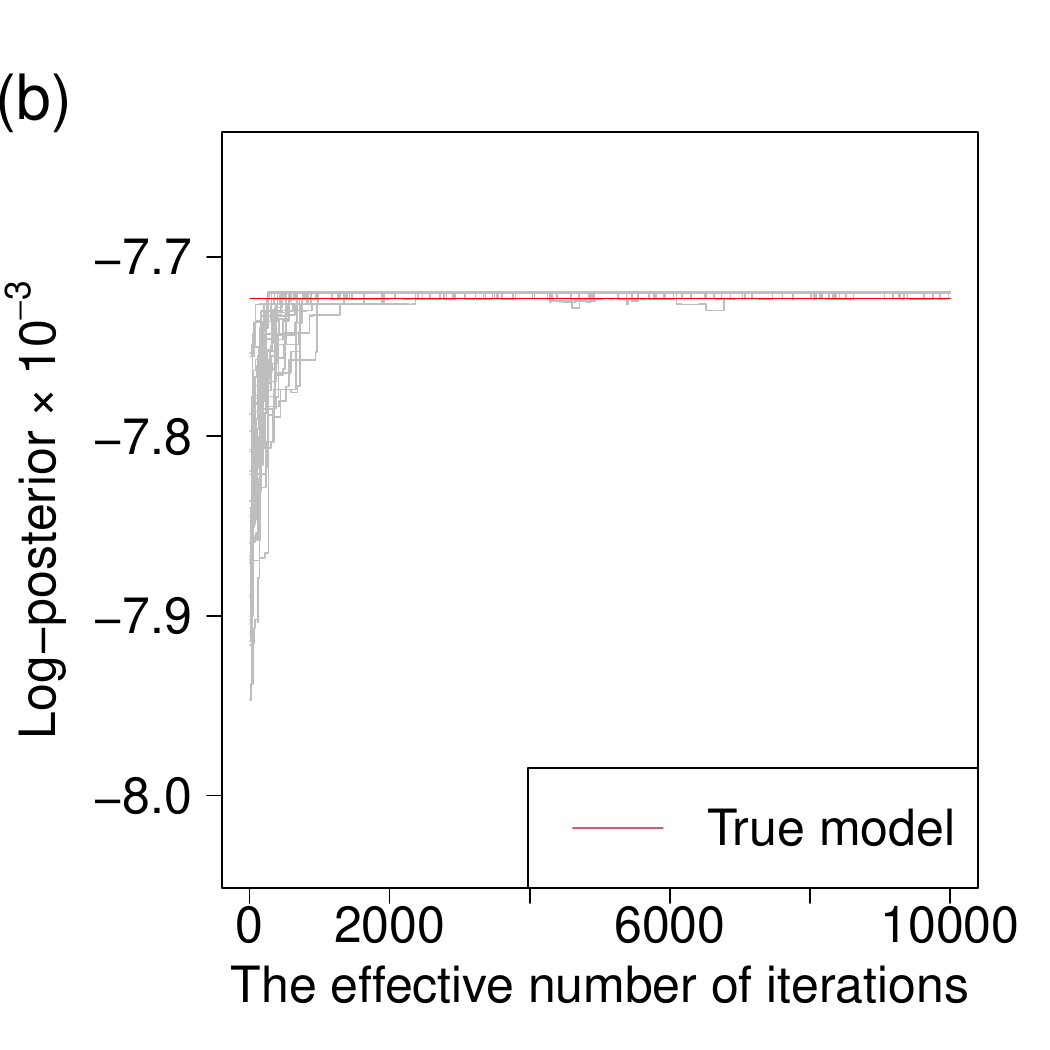}
    \includegraphics[width=0.32\textwidth]{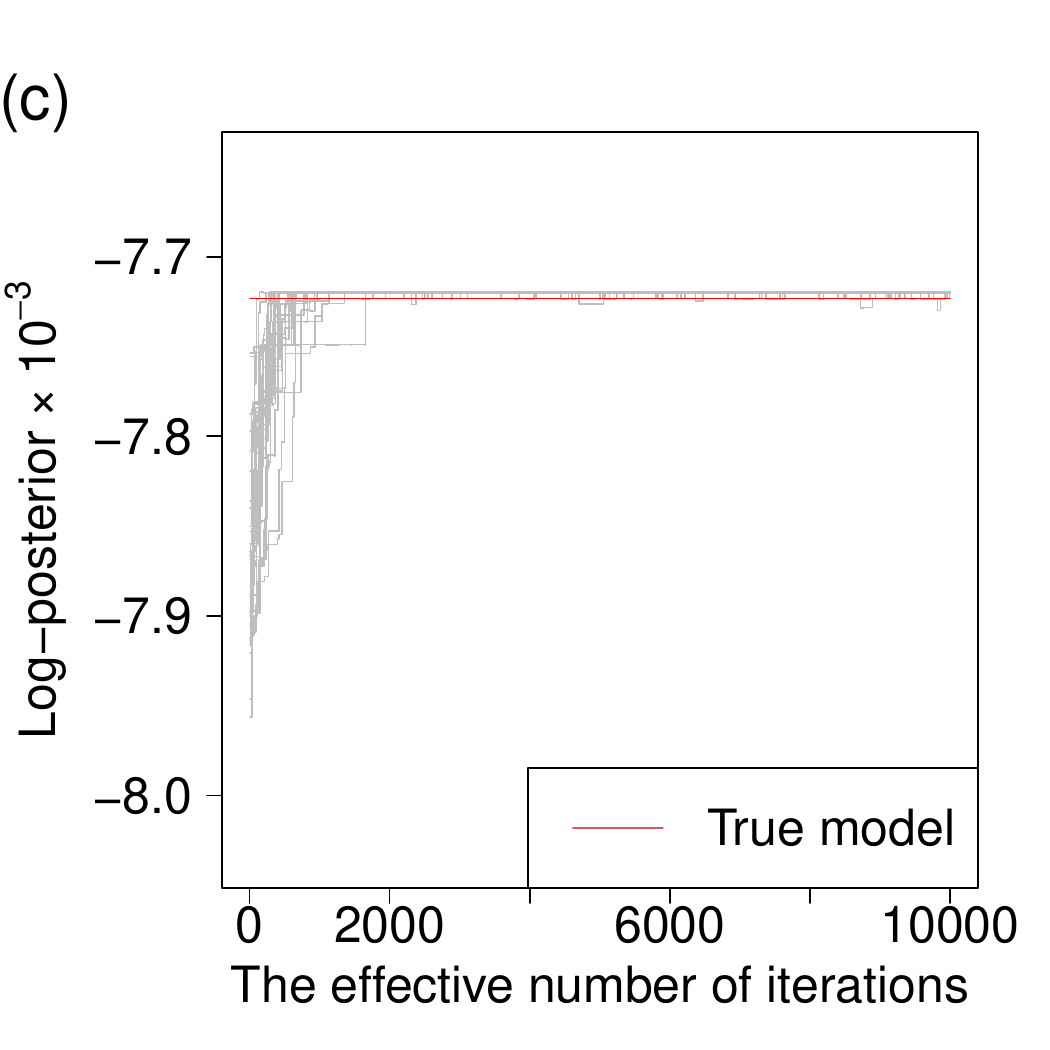}
    \caption{Log posterior probability times $10^{-3}$ versus the effective number of iterations of 30 MCMC runs for $p = 20$ and $n = 100$.  The red line represents the  true ordering $\sigma^*$.}
    \label{fig:mixing1}
\end{figure} 

\begin{figure}[!t]
    \centering
    \includegraphics[width=0.32\textwidth]{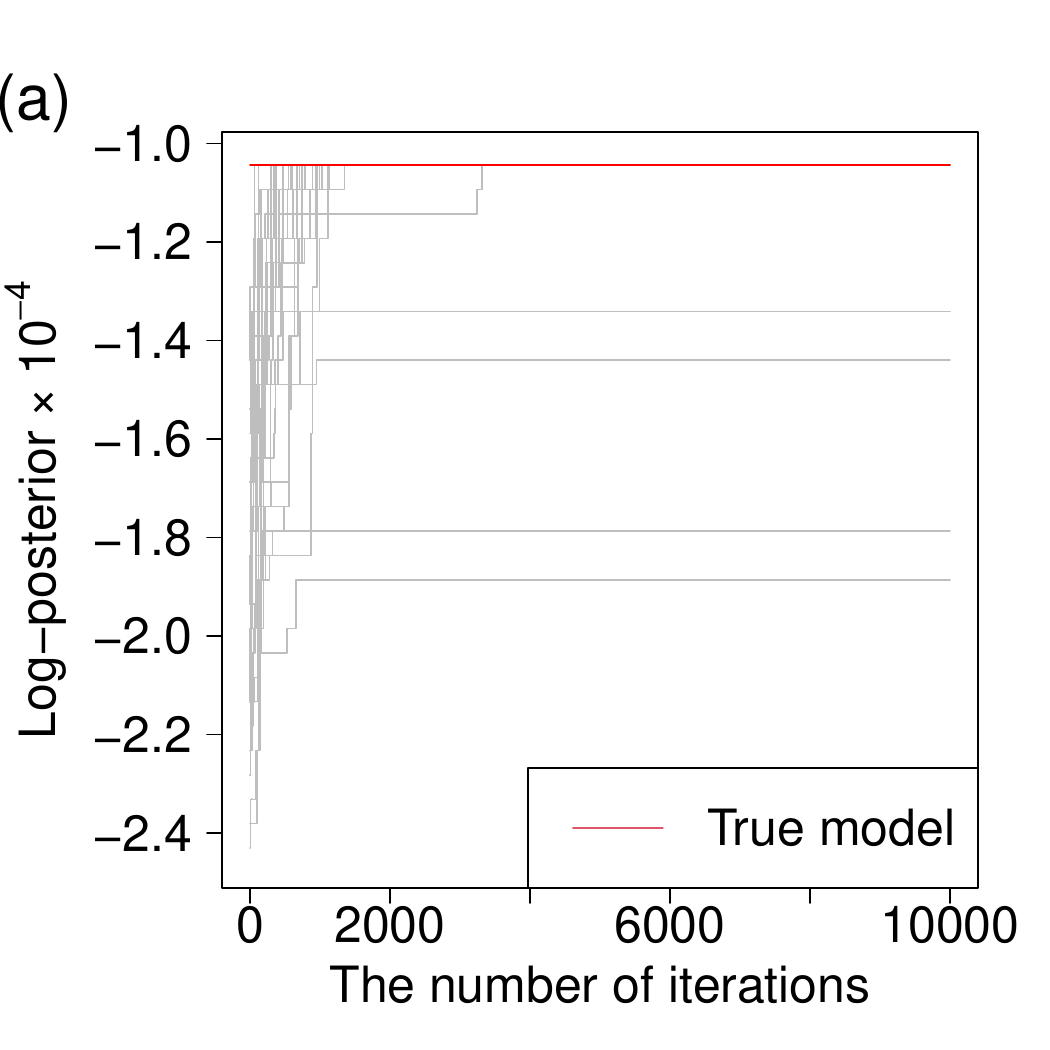}
    \includegraphics[width=0.32\textwidth]{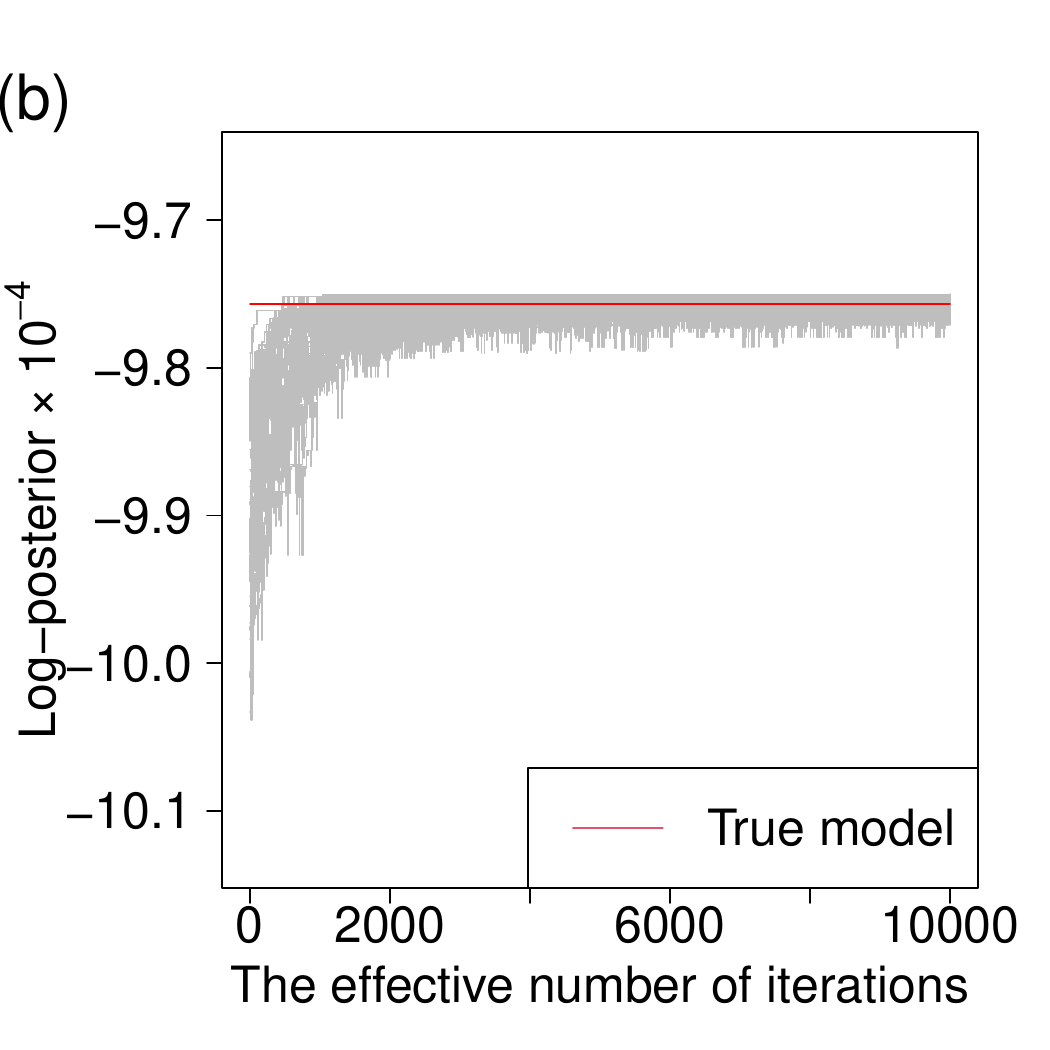} 
    \caption{Log posterior probability $\times 10^{-4}$ versus the  effective number of iterations of 30 MCMC runs with random initialization for the heterogeneous case with $p = 20$ and $n = 1000$: (a) minimal I-MAP MCMC, (b) the proposed method.  
    The red line represents the true ordering $\sigma^*$.  }
    \label{fig:mixing2}
\end{figure} 

\subsection{Performance evaluation}\label{subsec:compare_supp} 
We consider more scenarios for the simulation study described in Section~\ref{subsec:compare}.  We always fix $p = 40$. 
In Table~\ref{table:normal}, we still generate $X$ under the equal variance assumption but we sample each $B^*_{ij}$ for each edge $i \rightarrow j$ in the DAG $G^*$ from the standard Gaussian distribution.  The advantage of the proposed method is as significant as in Table~\ref{table:uniform} presented in the main text. 
In Table~\ref{table:hetero}, we sample the error variance $\omega_j$ for each $j$  uniformly from $[0.7, 1.3]$ and sample each $B^*_{ij}$ from the uniform distribution on $[-1, -0.3] \cup [0.3, 1]$. 
Comparing Table~\ref{table:hetero} with the left column of   Table~\ref{table:uniform}, we see that the advantage of our method over the competing ones becomes more substantial. 

\newpage

\begin{table}[ht]
\centering
\begin{adjustbox}{width=0.7\textwidth}
\small
\begin{tabular}{ccccc}
  & Signal & \multicolumn{3}{c}{$\mathrm{N}(0,1)$}         \\  
Method   & $n$      & 100              & 500             & 1000    \\ 
Proposed 
         & HD     & \textbf{10.4$\pm$0.8}    & \textbf{5.2$\pm$0.5}    & \textbf{4.2$\pm$0.4}    \vspace{-1mm}\\
         & FNR & \textbf{34.2$\pm$1.7}    & 17.0$\pm$1.7     & 13.8$\pm$1.2    \vspace{-1mm}\\
         & FDR    & \textbf{2.3$\pm$0.5}     & \textbf{1.7$\pm$0.5}    & \textbf{1.6$\pm$0.4}   \vspace{-1mm} \\
         & Flip   & 0.8$\pm$0.3       & \textbf{1.2$\pm$0.3}    & \textbf{1.0$\pm$0.3}    \vspace{-1mm}\\
         & Time   & 12.8$\pm$0.2      & 13.2$\pm$0.2     & 13.2$\pm$0.2   \\
TD      
         & HD     & 12.0$\pm$0.8      & 6.3$\pm$0.6      & 6.4$\pm$0.6       \vspace{-1mm} \\
         & FNR & 39.3$\pm$1.8      & 18.1$\pm$1.5     & 15.5$\pm$1.1      \vspace{-1mm} \\
         & FDR    & 3.7$\pm$0.8       & 4.8$\pm$1.2      & 6.8$\pm$1.2    \vspace{-1mm} \\
         & Filp   & 1.3$\pm$0.4       & 2.4$\pm$0.6      & 3.1$\pm$0.6      \vspace{-1mm} \\
         & Time   & 0.6$\pm$0.0       & 0.5$\pm$0.0      & 0.5$\pm$0.0      \\
LISTEN   
         & HD     & 12.5$\pm$0.8      & 6.5$\pm$0.6      & 5.9$\pm$0.6    \vspace{-1mm} \\
         & FNR & 39.3$\pm$1.8      & 18.8$\pm$1.5     & 15.3$\pm$1.2    \vspace{-1mm} \\
         & FDR    & 6.6$\pm$1.1       & 4.8$\pm$1.1      & 5.8$\pm$1.1     \vspace{-1mm} \\
         & Flip   & 2.0$\pm$0.4       & 2.6$\pm$0.6      & 2.8$\pm$0.5     \vspace{-1mm} \\
         & Time   & 0.5$\pm$0.0       & 0.5$\pm$0.0      & 0.5$\pm$0.0      
\end{tabular}
\end{adjustbox}
\caption{Standard Gaussian signal case with $p = 40$. Each entry gives mean $\pm$ 1 standard error.  The best performance with a margin of more than one $\mathrm{se}$ is highlighted in boldface. Time is measured in seconds.}
\label{table:normal}

\end{table}

\begin{table}[t]
\centering
\begin{adjustbox}{width=0.7\textwidth}
\small
\begin{tabular}{ccccc}
  & Signal &  \multicolumn{3}{c}{Heterogeneity}          \\ 
Method   & $n$            & 100           & 500           & 1000    \\     
Proposed 
         & HD      & \textbf{10.3$\pm$0.6} & \textbf{3.2$\pm$0.5}  & \textbf{4.4$\pm$0.8}  \vspace{-1mm}\\
         & FNR  & \textbf{33.1$\pm$1.6} & \textbf{6.0$\pm$1.0} & \textbf{6.0$\pm$0.8} \vspace{-1mm}\\
         & FDR      & \textbf{4.4$\pm$0.7}  & \textbf{6.1$\pm$1.2}  & \textbf{8.9$\pm$1.5} \vspace{-1mm} \\
         & Flip     & \textbf{2.8$\pm$0.5}  & \textbf{5.4$\pm$1.0}  &\textbf{ 6.0$\pm$0.8}  \vspace{-1mm}\\
         & Time   & 12.0$\pm$0.2   & 11.6$\pm$0.2   & 12.3$\pm$0.2   \\
TD      
         & HD        & 15.8$\pm$1.0   & 6.8$\pm$0.8    & 8.0$\pm$1.2   \vspace{-1mm} \\
         & FNR   & 45.5$\pm$2.0   & 10.0$\pm$1.1   & 9.1$\pm$1.2  \vspace{-1mm} \\
         & FDR      & 14.8$\pm$1.6   & 13.4$\pm$1.6   & 16.3$\pm$2.3  \vspace{-1mm} \\
         & Filp        & 7.5$\pm$0.9    & 9.2$\pm$1.1    & 9.0$\pm$1.2   \vspace{-1mm} \\
         & Time       & 0.5$\pm$0.0    & 0.5$\pm$0.0    & 0.5$\pm$0.0    \\
LISTEN   
         & HD        & 16.0$\pm$1.0   & 8.4$\pm$1.0    & 8.9$\pm$1.2   \vspace{-1mm} \\
         & FNR   & 46.2$\pm$1.9   & 11.3$\pm$1.0   & 10.0$\pm$1.1  \vspace{-1mm} \\
         & FDR        & 15.2$\pm$1.8   & 16.4$\pm$1.8   & 17.9$\pm$2.3  \vspace{-1mm} \\
         & Flip    & 7.1$\pm$0.8    & 10.5$\pm$1.0   & 9.7$\pm$1.1   \vspace{-1mm} \\
         & Time       & 0.5$\pm$0.0    & 0.6$\pm$0.0    & 0.5$\pm$0.0    
\end{tabular}
\end{adjustbox}
\caption{Heterogeneous error variance case with $p = 40$. Each entry gives mean $\pm$ 1 standard error. The best performance with a margin of more than one $\mathrm{se}$ is highlighted in boldface. Time is measured in seconds.}
\label{table:hetero} 
\end{table}

We also conduct simulation studies on the proposed algorithm with weakly increasing error variances.  We fix $n = 1,000$ and $p = 40$, and sample the error variance $\omega_j \sim \mathrm{Uniform}([1-b, 1+b])$ for 6 different heterogeneity levels $b$. We set $\sigma^* = (1, \dots, p)$ to be the true ordering and sort the error variances in ascending order to make them weakly increasing in $\sigma^*$. We generate $G^*$ by adding $i \rightarrow j$ for $i < j$ with probability $p_\mathrm{edge} = 3/(2p-2)$ and draw the edge weight $B_{ij}^*$ independently from some distribution $F$. In Table~\ref{table:increasing}, we present the results with 4 metrics: Hamming distance (HD), the false negative rate (FNR), false discover rate (FDR), and the percentage of flipped edges (Flip). The rows of $\mathrm{Uniform}$ and $\mathrm{Gaussian}$ indicate the result for $F$ being $\mathrm{Uniform}([-1, -0.3] \cup [0.3, 1])$ and that for $F$ being the standard normal distribution, respectively. 
Notably, the Flip rate is always very low, which indicates that  the algorithm can accurately identify the true ordering. When $b = 0.9$, FNR tends to be significantly larger. This is because some nodes may have very large error variances when $b = 0.9$, and thus the signal-to-noise ratio  is low, making it challenging for the algorithm to detect edges. 

\begin{table}[t]
\begin{tabular}{cccccccc}
Signal   &      & $b =$ 0           & $b =$ 0.1         & $b =$ 0.3         & $b =$ 0.5         & $b =$ 0.7         & $b =$ 0.9         \\\vspace{-1mm} 
Uniform  & HD   & 0.2$\pm$0.1  & 0.1$\pm$0.1  & 0.1$\pm$0.1  & 0.1$\pm$0.1  & 0.2$\pm$0.1  & 1.2$\pm$0.2  \\\vspace{-1mm} 
         & FNR  & 0.3$\pm$0.2  & 0.3$\pm$0.2  & 0.2$\pm$0.2  & 0.3$\pm$0.2  & 0.5$\pm$0.2  & 4.0$\pm$0.6  \\\vspace{-1mm} 
         & FDR  & 0.3$\pm$0.2  & 0.2$\pm$0.1  & 0.1$\pm$0.1  & 0.2$\pm$0.1  & 0.2$\pm$0.1  & 0.3$\pm$0.2  \\ 
         & Flip & 0.3$\pm$0.2  & 0.2$\pm$0.1  & 0.1$\pm$0.1  & 0.2$\pm$0.1  & 0.2$\pm$0.1  & 0.2$\pm$0.1  \\ \vspace{-1mm}
Gaussian & HD   & 4.9$\pm$0.5  & 4.3$\pm$0.4  & 4.5$\pm$0.4  & 4.7$\pm$0.4  & 5.1$\pm$0.4  & 6.0$\pm$0.4  \\\vspace{-1mm} 
         & FNR  & 15.4$\pm$1.4 & 15.3$\pm$1.3 & 14.9$\pm$1.2 & 15.5$\pm$1.2 & 17.1$\pm$1.3 & 20.2$\pm$1.3 \\\vspace{-1mm} 
         & FDR  & 2.2$\pm$0.6  & 0.4$\pm$0.2  & 0.5$\pm$0.2  & 0.3$\pm$0.2  & 0.3$\pm$0.2  & 0.4$\pm$0.2  \\\vspace{-1mm} 
         & Flip & 1.4$\pm$0.4  & 0.3$\pm$0.1  & 0.4$\pm$0.2  & 0.2$\pm$0.1  & 0.2$\pm$0.1  & 0.2$\pm$0.1 
\end{tabular}
\caption{A table for increasing error variances with heterogeneity level $b = 0, 0.1, \dots, 0.9$ with $p = 40$. We sample error variance from $\mathrm{Uniform}([1-b, 1+b])$ and sort in ascending order. Nonzero edge weights are from $\mathrm{Uniform}([-1, -0.3] \cup [0.3, 1])$ in Uniform case and $N(0,1)$ in Gaussian case.  Each entry gives mean $\pm$ 1 standard error.}
\label{table:increasing}
\end{table}

\subsection{Single-cell real data analysis}\label{subsec:realdata_supp} 
Figure~\ref{fig:real_data1} shows the result of the minimal I-MAP MCMC (with decomposable score) for the real data analysis. See Section~\ref{sec:real} in the main text for details. 

\begin{figure}[h]
    \begin{minipage}[c]{0.6\linewidth}
    \includegraphics[width=0.97\textwidth]{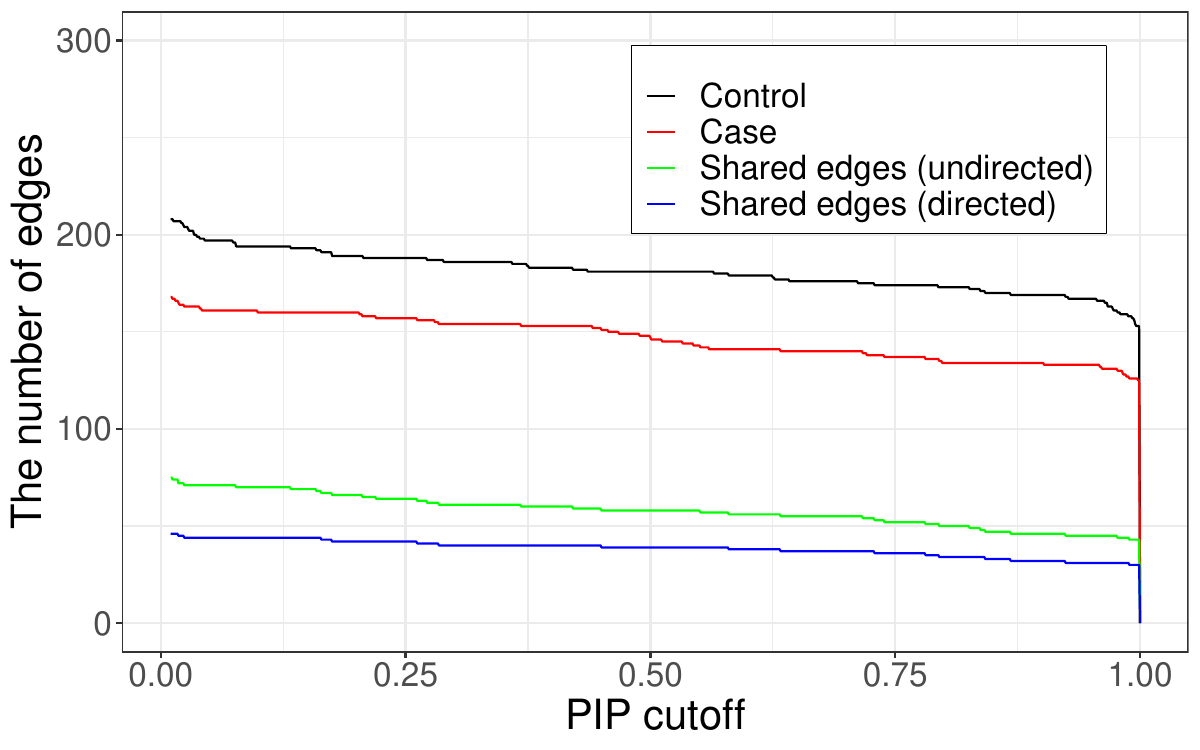}
\end{minipage}\hfill 
\begin{minipage}[c]{0.38\linewidth}
    \caption{Result of the minimal I-MAP MCMC for the real case-control data analysis. Given an estimate $\hat{\Gamma}_{ij}$ from MCMC samples, we infer the edge $i \rightarrow j$ exists in the DAG if $\hat{\Gamma}_{ij} > c$ where $c$ is the posterior inclusion probability cutoff. For each $c$, we count the number of edges occurring in the DAG for control samples (black), the number of edges in the DAG for case samples (red), the number of edges (edge direction ignored) in both DAGs (green), and the number of directed edges in both DAGs (blue).}
    \label{fig:real_data1}
    \end{minipage}
\end{figure}

\end{document}